%% file: arXiv.tex
\documentclass[11pt]{article}

\usepackage[toc,page]{appendix}

\usepackage{caption}
\usepackage{subcaption}

\usepackage{mathptmx,amsthm,amsmath,tikz}
\usepackage{amsfonts, amssymb}
\usepackage{cite}
\usepackage{cmap}
\usepackage[T1]{fontenc}
\usepackage{subfig}
\usepackage{authblk}
\usepackage{fancyhdr}

\include{header}

\begin{document}

\title{The Computational Complexity of Ball Permutations}

\author[1]{Scott Aaronson\thanks{Email: aaronson@cs.utexas.edu.}}
\author[2]{Adam Bouland\thanks{Email: adam@csail.mit.edu. Partially supported by the NSF GRFP under Grant No. 1122374 and by Scott Aaronson's NSF Waterman award under grant number 1249349.}}
\author[3]{Greg Kuperberg\thanks{Email: greg@math.ucdavis.edu}}
\author[2]{Saeed Mehraban\thanks{Email: mehraban@mit.edu. Partially supported by Scott Aaronson's NSF Waterman award under grant number 1249349.}}

\affil[1]{Department of Computer Science, University of Texas at Austin}
\affil[2]{CSAIL, Massachusetts Institute of Technology, Cambridge, MA}
\affil[3]{Department of Mathematics, UC Davis, Davis, CA}

\date{}
\maketitle{}

\begin{abstract}

Inspired by connections to two dimensional quantum theory, we define several models of computation based on permuting distinguishable particles (which we call balls), and characterize their computational complexity. 
In the quantum setting, we find that the computational power of this model depends on the initial input states. More precisely, with a standard basis input state, we show how to approximate the amplitudes of this model within additive error using the model $\DQC1$ (the class of problems solvable with one clean qubit), providing evidence that the model in this case is weaker than universal quantum computing. 
However, for specific choices of input states, the model is shown to be universal for $\BQP$ in an encoded sense. We use representation theory of the symmetric group to partially classify the computational complexity of this model for arbitrary input states.
Interestingly, we find some input states which yield a model intermediate between $\DQC 1$ and $\BQP$. 
Furthermore, we consider a restricted version of this model based on an integrable scattering problem in $1+1$ dimensions. We show it is universal under postselection, if we allow intermediate destructive measurements and specific input states. Therefore, the existence of any classical procedure to sample from the output distribution of this model within multiplicative error implies collapse of polynomial hierarchy to its third level. Finally, we define a classical version of this model in which one can probabilistically permute  balls. We find this yields a complexity class which is intermediate between $\L$ and $\BPP$. 
Moreover, we find a nondeterministic version of this model is $\NP$-complete.
\end{abstract}

\newpage

\section{Introduction}

The standard model of quantum computing is defined using quantum circuits acting on qubits. The computational power of this model is captured by the complexity class $\BQP$, which resides somewhere between the complexity classes $\BPP$ and $\PP$ \cite{FortnowRogers}. However, physical systems can exist in Hilbert spaces which are not described by tensor products of qubits. For example, systems of noninteracting fermions, noninteracting bosons, or anyons in a 2+1 dimensional quantum field theory all live in Hilbert spaces with different mathematical descriptions.

A natural problem is to explore the computational complexity of these alternative physical systems. 
There are several reasons to study this problem. 
First, many of these alternative models of quantum computing seem to be intermediate in power between classical and quantum computing \cite{jordan2009permutational,aaronson2011computational,DQC1ishard,JozsaShepherd}. 
Therefore it is interesting to study their power from a purely complexity-theoretic standpoint, as they help delineate the boundary between classical and quantum computation.
Second, these models sometimes have special properties from the perspective of mathematical physics. 
For instance, they might be ``solvable" or ``integrable" systems, which are regarded as simple to mathematical physicists.
It is interesting to compare notions of simplicity in mathematical physics (solvability and integrability) to the notion of simplicity in computational complexity (efficient classical simulablity).

Motivated by the above, 
in this work we consider an alternative model of quantum computing based on permuting quantum balls, and study its computational complexity. More specifically, we consider a system of $n$ distinguishable particles (``balls'') on a line. The computational basis of this Hilbert space consists of all permutations of the particles. We denote this Hilbert space by $\C S_n$.
The quantum operations in this model act on two balls at a time, and map the state $|x,y\rangle $ to the state $|x,y\rangle \rightarrow c |x,y\rangle + i s |y,x\rangle$, where $x$ and $y$ distinct integers from $1$ to $n$, and $c$ and $s$ are real numbers with $c^2+s^2=1$. 
For example, if one has the state $\ket{123}$, and applies the above operation to the first two particles, the resulting state would be $c\ket{123}+is\ket{213}$. 
This is a quantum analog of exchanging the particles in that location with probability $s^2$; hence we call this the ``partial swap'' gate.
Physically, this ``ball-permuting'' model captures the scattering problem of distinguishable particles on a line with short-range interactions\footnote{In the interactions we consider, contact between particles is penalized with a delta function. In this sense the interactions are ``hard". However, in the physics literature these are referred to as ``soft short-range" interactions since they allow permutation of particles, whereas ``hard" interactions forbid the permutation of particles.}.

We study the computational complexity of several models based on the above formalism. First, we consider a model in which one starts in the state $|123\ldots n \rangle$, and then applies polynomially many partial swap gates. We show that one can approximate amplitudes in this model, within $1/\operatorname*{poly}$ additive error, within the one-clean-qubit model of Knill and LaFlamme, also known as the complexity class $\DQC 1$ \cite{knill1998power}. This class captures the power of quantum computers in which all input qubits are in the maximally mixed state (i.e. a uniformly random basis state unknown to the experimenter) except one, which is in a pure initial state. This model is widely believed to be substantially weaker than $\BQP$; in fact it is open whether or not $\DQC 1$ is even capable of universal \emph{classical} computation.Therefore, the power of this model seems to be substantially weaker than $\BQP$. Our result shows that exchange based quantum computing with distinguishable particles starting in the state $\ket{12\ldots n}$ yields a weak computational model. This in turn suggests that indistinguishability is a crucial computational resource for exchange based computations.

Next, we consider the computational power of this model when we begin with arbitrary initial states (i.e. states more complicated than $\ket{1,2,\ldots,n}$). First, we show that if the initial state is selected according to certain irreducible invariant subspaces, then this model can efficiently simulate $\BQP$, using an encoded universality. We also mention an explicit construction based on the result that the exchange interaction on qubits is encoded-universal \cite{divincenzo2000universal,fong2011universal}. Therefore allowing arbitrary initial states substantially boosts the power of this model.

Furthermore, we obtain a partial classification of the computational complexity of this model on different input states. In order to achieve this goal, we use the representation theory of the symmetric group. In particular, we use the Young-Yamanouchi orthonormal basis \cite{james1981representation} to describe subspaces of our Hilbert space that are invariant and irreducible under the action of ball permuting gates. Therefore to understand the complexity of our model starting from an arbitrary input state, one merely needs to analyze its components in the Young-Yamanouchi basis using representation theory.  We make progress towards this goal in Section \ref{arbitrary}. 
 
One interesting finding of this classification is the discovery of a natural model of quantum computing which seems to be intermediate between $\DQC 1$ and $\BQP$. 
In this model, we initially start with the input state:

$$
\dfrac{1}{\sqrt{2^n}}\displaystyle\sum_{x\in\{0,1\}^n} \ket{x}\ket{x},
$$

\noindent consisting of $n$ active qubits maximally entangled to $n$ inert qubits. We imagine we can apply an arbitrary quantum circuit to the active qubits only (i.e. the left half of the state only), and then we measure both active and inert qubits together in the computational basis. 
If one were to only examine the computation on the left half of the state only\footnote{As well as the first qubit of the right half of the state}, this model would be equivalent to $\DQC 1$ - because the left half of the state considered in isolation is maximally mixed. 
However, in this new model one additionally gets to observe the right half of the state. In other words, although the qubits are initialy in an (unknown) uniformly random basis state at the start of the computation, one gets to learn what basis state they started in, but only \emph{at the end of the computation}. 
We find that this model arises naturally in ball-permuting circuits starting with input states of specific ``reducible subspaces''.
This model may be of independent interest.

Third, motivated by physical considerations, we consider a restricted version of this model in which the exchange operators are required to satisfy the (parameter dependent) Yang-Baxter equation. The Yang-Baxter equation arises numerous places in mathematical physics, including two-dimensional quantum field theory, statistical mechanics, and topological quantum field theory. The equation describes sets of operators which form representations of the braid group. \color{black} In our setup the Yang-Baxter equation arises naturally from quantizing a system of $n$ distinguishable particles with their own momenta and short-range \cite{yang1967some} \color{black} interactions\footnote{ Specific to the integrable theories we consider in $1+1$ dimensions, the Yang-Baxter equation is a consistency equation which implies that the unitary evolution depends only on the initial momenta of the particles. In the language of mathematical physics this last condition is also known as Bethe Ansatz solvability.}. Less formally, this arises if one imagines that the balls have their own velocities, and upon colliding, either exchange velocities (as in classical physics) or else pass through each other, with quantum amplitudes $c$ and $is$, respectively. 
In this case one cannot exchange arbitrary particles due to the velocity constraints.

We show that if intermediate measurements are added to this model, then one cannot sample from the same probability distribution (up to multiplicative error) efficiently with a classical computer unless the polynomial hierarchy collapses to the third level. The proof uses a postselection \cite{aaronson2005quantum} argument; we show that postselecting on possibly exponentially-unlikely measurement outcomes allows the model to efficiently solve any problem in the complexity class $\Post \BQP$. Then using the same line of reasoning as in Aaronson-Arkhipov \cite{aaronson2011computational} and Bremner-Jozsa-Shepherd \cite{JozsaShepherd}, one can infer that the existence of an efficient procedure to sample from the distribution of outcomes in the proposed model within multiplicative error implies the collapse of polynomial hierarchy to the third level.

This result might be somewhat surprising to mathematical physicists because systems obeying the Yang-Baxter equation in $1+1$ dimensions are considered simple because they are ``integrable''. 
Integrable systems are exactly solvable in the sense that the number of conserved quantities\footnote{For example, in classical mechanics total energy and total momentum are conserved quantities.} exceeds the number of degrees of freedom.
So by specifying the conserved quantities, one has fully specified the evolution of the system; there is no need e.g. to solve a differential equation to find the evolution of the system. Moreover our model is Bethe ansatz solvable, which in the language of mathematical physics means that their unitary evolution can be fully specified with \emph{linearly} many real parameters; this is a particular property of our model which makes it very easy to describe, and is not known to apply to integrable theories in higher dimensions. This is true in our Yang-Baxter ball-permuting model, since the initial particle velocities fully determine the unitary evolution of the system. As a result, the dimensionality of the set of unitary matrices allowed in this model grows only linearly in the number of particles, while the dimensionality of the space of all unitary evolutions grows exponentially. Intuitively, it seems difficult to ``program" such models to do any interesting computations, since there are few parameters in the system.
Our result therefore shows that even 1+1 dimensional integrable models may be difficult to computationally simulate classically, at least in the presence of intermediate measurements. 
Therefore integrable models (even those which are Bethe ansatz solvable) are not necessarily computaionally simple. Note however in general there are solutions to the Yang-Baxter equation that are known to be $\BQP$ universal, e.g.,  topological quantum computing in $2+1$ dimensions \cite{freedman2003topological}. However these are not Bethe ansatz solvable. Also, prior work has shown that some solutions to the Yang-Baxter equation on qubits \cite{alagic2014classical} are efficiently classically simulable. However in \cite{alagic2014classical} this is not due to the small number of parameters present, but for other reasons.

Finally, we consider classical versions of this model, where a base $\AC^0$ machine can query a deterministic, probabilistic, or non-deterministic ball-permuting oracle. Inputs to a deterministic ball permuting oracle are lists of swaps, and outputs are the permutations that are resulted from the application of swaps in order. A randomized ball permuting oracle also takes a list of probabilities as input, applies the swaps probabilistically and outputs the final permutation. The model corresponding to the deterministic ball permutation is proved to be equivalent to $\L$ (log-space) Turing machines. The randomized ball permutation model, on the other hand, can simulate $\BPL$ (Bounded-error probabilistic log-space) machines efficiently. We also show that a machine from the class $\Almost \L$ ($\L$ relative to a random oracle) can efficiently simulate randomized ball permutation. So the power the randomized ball permuting model lies between $\BPL$ and $\Almost\L$. We raise the open problem of whether this can be simulated in $\P$. We also consider a nondeterministic version of the ball-permuting model. We find this to be equivalent to $\NP$, unless the swaps are between adjacent balls only, in which case it can be simulated in $\P$ using nontrivial planar graph algorithms.

\section{Models and Motivations}
\label{model}

In this section we define the quantum ball-permuting model as a model of quantum computation, and then briefly explain how it models actual physical processes. The basic operations of this model are quantum swaps, as we will define them shortly. To obtain a point of comparison with the classical world we analyze the power of deterministic, randomized and non-deterministic swaps separately in section \ref{classical}.

The computational basis states in the ball permuting model are all $n!$ possible
permutations on an $n$-element set (\i.e., $S_n$).  We start out in the
state:

$$
|1,2,\ldots,n\rangle.
$$

At each time step, the rule is: we get to pick any adjacent\footnote{We get the same model if we allow swaps between any pairs. This is because we can simulate general swaps with adjacent ones.} pair of the $n$
registers in our quantum state, and then apply an $(n^2-n)\times (n^2-n)$ unitary
transformation that, for every pair of distinct labels $x\neq y$, maps:

$$
|x,y\rangle \rightarrow c|x,y\rangle + is|y,x\rangle,
$$

\noindent where $c$ and $s$ are any two real numbers satisfying $c^2+s^2=1$.  (We get
to pick whichever $c$ and $s$ we like for each gate operation.  However, $c$
and $s$ can't vary depending on the labels $x$ and $y$. If one does allow for $c$ and $s$ to depend on $x$ and $y$, then in Appendix \ref{zqball} we show one recovers $\BQP$). 

Finally, we measure all $n$ registers in the computational basis, and
feed them to a classical computer for postprocessing.

We can represent the computational basis by the kets $|\sigma\rangle$ for any permutation $\sigma\in S_n$. We denote this Hilbert space by $\C S_n$. A quantum swap between $t$ and $t+1$'th registers depends on one real parameter $\theta$ and is represented by the operator:

$$
X(\theta,t)= \cos \theta I + i \sin \theta L_{(t,t+1)}.
$$

\noindent Here $I$ is the identity operator, and $L_{\sigma}$ is a representation of $\C S_n$ with the action\footnote{The operator $L$ acts from left to right on the basis of the Hilbert space. We can also talk about right multiplications $R(\tau)$ which map $\ket{\sigma} \mapsto \ket {\sigma \circ \tau^{-1}}$ for permutations $\sigma$ and $\tau$. The importance of these operators become clear in Section \ref{partialclassification1}}:

$$
L_\sigma |\tau\rangle = |\sigma \circ \tau \rangle.
$$

The idea of this model is to capture $n$ distinguishable particles
moving around on a line (\i.e., in $1+1$ dimensions).  The only state
that we care about is the order of the particles in the line.  The
$x$'th register of the quantum computer stores the label of the $x$'th
particle, if the particles are listed in order from left to right.
Whenever two particles meet, one of two things can happen: the
particles can reflect, or they can pass through each other.  The first happens with amplitude $c$, while the second happens with
amplitude $is$.

We also consider a restricted version of this model, in which each particle $1\ldots n$ has its own velocity $v_1\ldots v_n$, and particles can only interact if they collide. For instance if particle 2 is initially moving to the right, and all other particles are stationary, then particle 2 will only interact with particles to its right, and cannot interact with particles to its left. This model captures the power of scattering experiments with n distinguishable particles and repulsive (hard-core) interactions. This is illustrated in Figure \ref{fig22_v0}.
\begin{figure}[h]
\begin{center}
\includegraphics[height=2.0in]{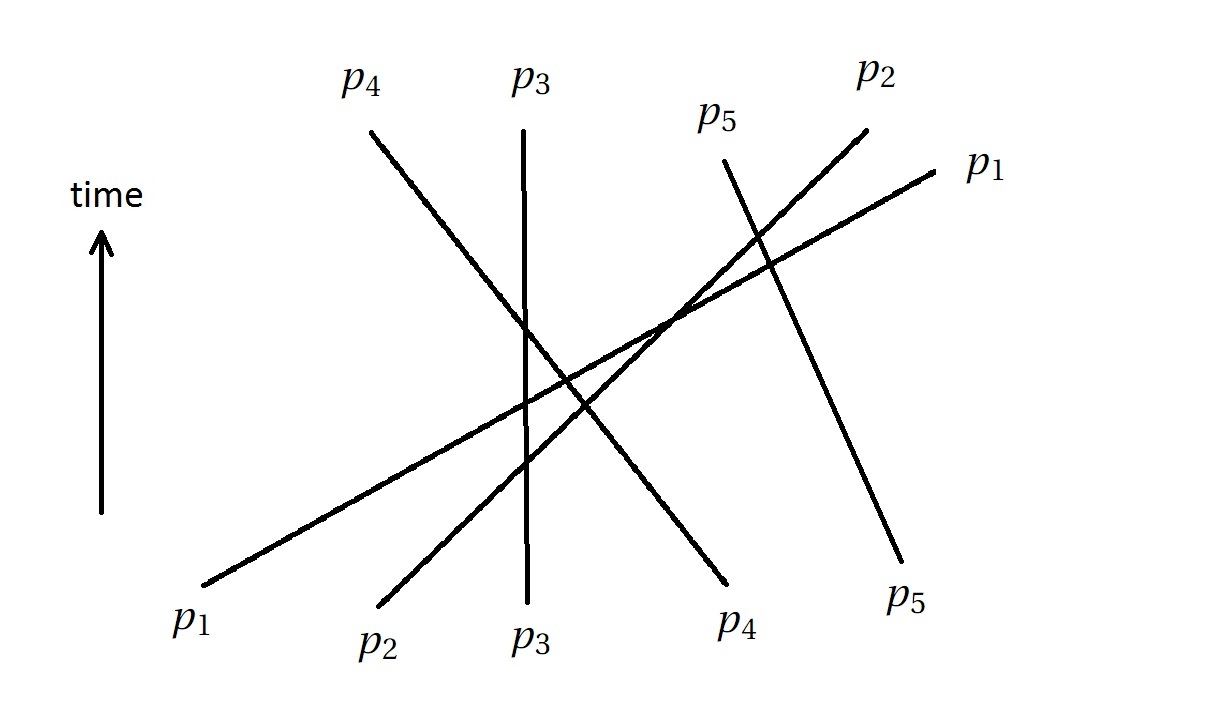}
\caption
{Five particles with momenta $p_1, p_2, \ldots, p_5$ moving on a line. At each intersection between these lines there is an interaction modeled by a ball-permuting gate. No other interactions are allowed.}
\label{fig22_v0}
\end{center}
\end{figure}
Furthermore, as described in Appendix \ref{models}, if one rigorously defines this model, one finds the parameter $\theta$ in the exchange interactions cannot be arbitrary. Suppose we define a new parameter $z$, called the rapidity, which is the relative velocity of the colliding particles. Then we find that the exchange interaction must have $\theta=\tan^{-1}(z)$. Furthermore, if we let $R(z,t)=X(\tan^{-1}(z),t)$, then we additionally find our quantum operations must satisfy:
$$
R(z_1,1) R(z_2,2) R(z_3,1) = R(z'_3,2) R(z'_2,1) R(z'_1,2).
$$
This is known as the parameter-dependent Yang-Baxter equation. We can confirm that the only nontrivial solution to this equation is, $z_1=z'_1$, $z_3=z'_3$, and $z_2=z'_2=z_1+z_2$. Therefore this equation constrains the values of $z$ (and therefore the values of $\theta$) allowed in interactions in the model.

As a result, we will consider two types of models in this paper: one in which we imagine we can set $\theta$ to any value we choose for each interaction, and another based on the physics of colliding particles in which the interactions are constrained by both their velocities and the Yang-Baxter equation. The former model is much simpler to define and work with, while the latter model is more directly related to the physics of interacting particles on a line.

\section{The quantum ball permuting model}
\label{modelIntro}

Based on the model introduced in Section \ref{model}, we will now formally define quantum ball-permuting complexity classes. We analyze their power with standard and arbitrary initial states in Sections \ref{sep} and \ref{arbitrary}, respectively. These sections assume that one can set the amplitudes $c$ and $s$ to arbitrary values for each interaction. Then, in Section \ref{Post}, we analyze the complexity of this model when the interactions are constrained to obey the velocity constraints and the Yang-Baxter equations introduced at the end of Section \ref{model}.

\begin{definition}

$\QBall$ is the class of languages $L \subseteq \{0,1\}^\star$ for which there exists a polynomial time Turing machine $M$ which on any input $x \in \{0,1\}^\star$, outputs the description of a ball permuting quantum circuit $C$ as a composition of $X$ operators, and the description of a subset $P \subseteq S_n$ of permutations, such that if $x\in L$, then the probability that the sampled permutation from $C$ is in $P$ at least $\dfrac{1}{2} + \dfrac{1}{\operatorname*{poly}(n)}$, and otherwise it is at most $\dfrac{1}{2} - \dfrac{1}{\operatorname*{poly}(n)}$.  Likewise, $\RQBall$ is defined to be the subclass of $\QBall$ where the exchange gates of the ball permuting circuits are constrained to satisfy the Yang-Baxter equation. 
\end{definition}

Note that in this definition we have not specified in which state the quantum circuit $C$ is initialized. The exact power of the class $\QBall$ (or $\RQBall$) may depend on the input state allowed. Unless otherwise specified, we will assume the initial state is $\ket{123\ldots n}$, i.e. the identity permutation.

\subsection{Standard Initial States}
\label{sep}

In this section we will consider the power of $\QBall$ when the initial state is the identity permutation $\ket{12\ldots n}$. We observe the following containments for the ball permuting complexity classes:

\begin{theorem}
$\RQBall \subseteq \QBall \subseteq \BQP$.
\end{theorem}

The containment $\RQBall\subseteq \QBall$ is trivial. For $\QBall \subseteq \BQP$ note that we can represent labels with binary strings using $O(\log n)$ qubits. By the Solovay-Kitaev theorem \cite{DawsonSolovayKitaev}, any unitary on $O( \log n )$ qubits can be implemented using a polynomial-size quantum circuit, and therefore a $\BQP$ circuit can simulate a $\QBall$ circuit. Therefore the power of $\QBall$ is upper bounded by $\BQP$; this model is no more powerful than standard quantum computing. 

We will now show that the power of $\QBall$ is likely much weaker than that of $\BQP$, because one can efficiently estimate individual amplitudes in $\QBall$ to $\dfrac{1}{\operatorname*{poly}}$ error in the complexity class $\DQC1$, i.e. the class of quantum computations that can be performed with one clean qubit and $n-1$ maximally mixed qubits. $\DQC1$ is believed to be a substantially weaker complexity class than $\BQP$. (For a discussion of $\DQC1$ see Appendix \ref{dqc1intro}.) In constrast, for circuits over qubits it is $\BQP$-hard to compute individual amplitudes to $\dfrac{1}{\operatorname*{poly}}$ accuracy. The following observation is crucial for the establishment of our result relating $\QBall$ and $\DQC1$:

\begin{lemma}
If $C$ is any composition of $X$ ball permuting operators over $\C S_n$, then $C|123\ldots n\rangle = |123\ldots n \rangle$ if and only if $C=I$.
\end{lemma}

\begin{proof}(Sketch) Suppose that $C|123\ldots n\rangle = |123\ldots n \rangle$, then permuting the labels arbitrarily gives $C|\sigma\rangle = |\sigma\rangle$ for all permutations $\sigma \in S_n$.
\end{proof}

The main result is the following:

\begin{theorem}
There is a $\DQC 1$ algorithm which takes as input a description of a ball permuting circuit $C$ over $\C S_n$, and which outputs a complex number $\alpha$ such that $|\alpha - \langle 123 \ldots n | C | 123 \ldots n\rangle| \leq \dfrac{1}{\operatorname*{poly}(n)}$, with high probability \footnote{Furthermore, the $\DQC 1$ algorithm is able to find additive approximations for both the real and and imaginary values of the amplitude, separately.}.
\label{DQC1thm}
\end{theorem}

\begin{proof} (Sketch) For any ball-permuting circuit, we know that for any two permutations $\pi$ and $\pi'$, $\langle \pi|C|\pi\rangle=\langle \pi'|C|\pi'\rangle$. This implies that $\langle 123\ldots n|C|123\ldots n\rangle = \frac{\mathrm{Tr}(C)}{n!}$. On the other hand, Knill and Laflamme \cite{knill1998power} showed for any quantum circuit $U$ on $n$ qubits composed of polynomially many gates,  a $\DQC1$ circuit can output a bit which is $1$ with probability $\frac{1}{2}+|\frac{\mathrm{Tr}(U)}{2^{n+1}}|$, and therefore estimate $|\mathrm{Tr}(U)/2^n|$ to $1/\operatorname*{poly}$ additive error with postprocessing. Therefore it suffices to create a qubit unitary $U$ on $n$ qubits with the same trace as $C$ on $n$ particles, where $2^{n+1}$ is approximately the same as $n!$.
Fortunately it is possible to do this using a carefully chosen encoding of permutations with qubits. 
In particular, we use an encoding of permutations that is both compressed and local. By compressed we mean that it uses $\log\left(n! \operatorname*{poly}(n)\right)$ bits, and by local we mean that we can use polynomial-size quantum circuits to simulate each quantum swap. We discuss this in detail in Appendix \ref{trace}. 
\end{proof}

A similar argument holds for arbitrary amplitudes $\langle \sigma |C| \sigma' \rangle$. This is followed by the reduction $C\mapsto L_\sigma C L_{\sigma'^{-1}}$ to the above problem.
Therefore $\DQC1$ can efficiently estimate amplitudes in the $\QBall$ model to $1/\operatorname*{poly}$ accuracy. In contrast, for qubits, computing individual amplitudes to $1/\operatorname*{poly}$ accuracy is $\BQP$-hard. Therefore this is evidence that $\QBall$ is a weaker model of computation than $\BQP$. 

Note however Theorem \ref{DQC1thm} does not imply $\QBall\subseteq \DQC1$ as decision languages, because we have only shown how to compute individual amplitudes of $\QBall$ in $\DQC1$, while the decision class $\QBall$ may accept or reject based on an exponential sum of amplitudes. Likewise, Theorem \ref{DQC1thm} does not immediately imply a $\DQC1$ machine could sample from the output distribution of a $\QBall$ circuit. Whether it is possible to efficiently sample the output distributions of ball permuting circuits with $\DQC 1$ computation is unknown.

We conjecture that this result can be generalized to a wide variety of quantum models based on group algebras. More precisely, consider a group $G$, with identity element $e$. Then construct the Hilbert space $\mathcal{H}_G$ with orthonormal basis $\{ |g\rangle : g \in G \}$. Let $\C G$ be the (left) group algebra, and $x \in \C G$. Then $\langle e | x | e \rangle= \dfrac{1}{|G|} \operatorname*{Tr} (x)$, which is a reduction to the computation of normalized trace. 
If one has a compressed and local encoding of $G$ as described in the proof of Theorem \ref{DQC1thm}, then this would imply that amplitudes of computations over $\C G$ can be simulated in $\DQC1$ as well. We leave this is an open problem.

\subsection{Arbitrary Initial States}
\label{arbitrary}

We already observed that the model we obtain seems to be restricted if one starts with and measures according to the computational basis. In this section, we further examine how the power of ball-permuting model depends on the input states. We first give a simple construction showing that $\QBall$ is universal for $\BQP$ when given particular input states. The proof is based on DiVincenzo \emph{et al.}'s result that the exchange interaction is universal for quantum computing \cite{divincenzo2000universal}. We then provide a partial classification of the power of $\QBall$ on different input states using the representation theory of $S_n$.This requires substantial work in representation theory, and is the main technical contribution of the paper.  In Section \ref{partialclassification1} we describe a number of input states which boost the power of $\QBall$ up to $\BQP$. We then describe some other input states which yield a model intermediate between $\DQC1$ and $\BQP$ in Section \ref{intermediate}.

\subsubsection{A simple proof that $\QBall=\BQP$ on arbitrary initial states}

First of all, building on DiVincenzo \emph{et al.}'s result that the exchange interaction is universal for quantum computing \cite{divincenzo2000universal}, we observe that when the initial state need not be a computational basis state, the quantum ball permuting model has the full power of $\BQP$. The proof uses the notion of \emph{encoded universality}: although our model does not allow arbitrary unitaries on the Hilbert space $\mathbb{C} S_n$, it can simulate arbitrary unitaries on certain subspaces of $\mathbb{C} S_n$ which encode qubits. We can therefore perform universal quantum computation on the encoded subspace, assuming our inputs are allowed to lie in the encoded subspace. This can be summarized by the following theorem:

\begin{theorem}
If $\QBall$ is allowed to have non-basis input states, then $\QBall=\BQP$.
\end{theorem}

The authors of \cite{divincenzo2000universal, bacon9909058universal} showed that one could use the exchange interaction to simulate one logical qubit using three physical qubits by the following encoding:

$$
|0_L\rangle := \dfrac{|010\rangle - |100\rangle}{\sqrt{2}}
$$ 

\noindent and,

$$
|1_L\rangle := \dfrac{|010\rangle + |100\rangle-2 |001\rangle}{\sqrt{6}}.
$$ 

\noindent Moreover the author of \cite{divincenzo2000universal} showed how to implement an approximate CNOT on logical qubits using exchange interactions. This result was further improved in \cite{fong2011universal}, where the authors found closed form expressions for this implementation. 

In our model, we mimic this encoded universality construction for qubits using permutations. 
Specifically, to encode a logical qubit, we use permutations of a three-element set. We let the permutation labels $1$ and $2$ represent the qubit state $\ket{0}$ and the permutation label $3$ represent the qubit state $\ket{1}$. We then symmetrize over the labels which represent with zeros and over the labels which represent with ones, to obtain states over $\C S_3$ which represents each basis state $\ket{001},\ket{010},\ket{100}$ used in DiVincenzo et al.'s construction. For example we represent the qubit state $|001\rangle$ with the state $|123\rangle + |213\rangle$, and we represent the state $\ket{010}$ with the state $\ket{132}+\ket{231}$.
To encode logical zero and logical 1, we use DiVincenzo \emph{et al.}'s encoding, ported over to permutations using the above correspondance.
One can check this simulation works because the only operators being used in both models are permutations. 
For $n>3$ qubits we use the same symmetrization idea to simulate exchange interactions on $(\C^2)^{\tensor n}$;  a detailed proof of this fact in given in Appendix \ref{bqpuniversality}.

Note that an alternative construction for encoded universality can be obtained using the path representations of Appendix \ref{remarksexchange}, which stands in line with the representation theory discussed in the next section. However,  DiVincenzo \emph{et al.}'s construction has the advantage of being more explicit. More specifically, we can encode $|0_L\rangle$ with the path $\ket{udud}$ and $|1_L\rangle$ with $\ket{uudd}$. Here $u$ is corresponds to one step up, and $d$ is one step down. Using the result of next Section we know that we can apply $SU(V)$ on the space of these paths. Therefore to simulate $\BQP$ first initialize in $\ket{udud\ldots udud}$, implement sequences of gates to mimic arbitrary rotation and entangling two qubit gates on the encoded qubits and sample a path from the output. Note that the construction of these gates directly follow from theorem \ref{thm:bigirreps} of next Section. This encoding of qubits is also very similar to Aharonov and Arad's \cite{aharonov2011bqp} construction for proving $\BQP$-completeness for the evaluation of Jones polynomial.

\subsubsection{Partial classification of input states which make $\QBall=\BQP$}
\label{partialclassification1}

Our particular contribution in this section is to demonstrate a partial classification for the computational power of this model according to different initial states. Our objective is to demonstrate that different input states lead to different interesting models of computation. This classification is obtained using the representation theory of the symmetric group. (For a brief review of this theory see Appendix \ref{RepTheory}).

A representation of a group $G$ is a homomorphism $:G\rightarrow GL(V)$, for some vector space $V$, which obeys the same multiplication rule as the group law. We interchangeably refer to $V$ or the homomorphism itself as the representation. The regular representation of $S_n$ with left action is according to the homomorphism $L: S_n\rightarrow GL(\C S_n)$, with the map $L_g: |h\rangle \mapsto |g . h\rangle$. Similarly, the right action $R: S_n\rightarrow GL(\C S_n)$, is according to the map $R_g: |h\rangle \mapsto |h.g^{-1}\rangle$. An invariant subspace is a subspace that is stable under the action of a particular representation, \i.e., the image of this subspace under the action of the group is equal to the subspace itself. A representation is called irreducible if its only invariant subspaces are the singleton $\{0\}$ and the representation itself.

Under these (left and right) regular representations the Hilbert space $\C S_n$ decomposes into irreducible representations as:

$$
\C S_n \cong \bigoplus_{\lambda \vdash n} V_\lambda \tensor X_\lambda,
$$

\noindent where each $\lambda$ is a partition of the number $n$: that is, a list of non-negative integers in non-ascending order summing to $n$. $\lambda \vdash n$ means that $\lambda$ is a partition of $n$. $V_\lambda \tensor X_\lambda$ is a summand of the decomposition the tensor product of two vector spaces with $\dim(V_\lambda)=\dim(X_\lambda)=:d_\lambda$ and $\sum_\lambda d^2_\lambda = n!$.. Interestingly, the left actions of $S_n$ only acts on $V_\lambda$'s and act trivially on the $X_\lambda$'s. $X_\lambda$'s are called the multiplicity spaces corresponding to each $V_\lambda$, as they enumerate all the irreducible representations isomorphic to $V_\lambda$.

Denote the Lie group generated by ball-permuting operators with $G_n$, and the projection of $G_n$ onto $V_\lambda$ with $G_n (V_\lambda)$. 

\begin{theorem}
\label{thm:bigirreps}
If $\lambda\vdash n$ or its conjugate consists of two parts, then $SU(V_\lambda) \subseteq G_n(V_\lambda)$.
\end{theorem}

\begin{proof}(Sketch) By induction, we first prove the statement for $n=3$. We then use subgroup adapted Young-Yamanouchi basis which manifests the branching rule (discussed in Appendix \ref{YoungYamanouchi}), along with decoupling lemma and bridge lemma of Aharonov and Arad \cite{aharonov2011bqp} to deduce the statement of this theorem for each subspace one-by-one. This procedure fails for subspaces corresponding to general partitions with more than two parts. 
\end{proof}

Therefore, if the input state to $\QBall$ corresponds to an irreducible representation which a) consists of two parts and b) is sufficiently large in dimension, then by Theorem \ref{thm:bigirreps} one can perform encoded universal computation on this input state. The description of these input states is given in Appendix \ref{classification}. We believe that the result can be extended to partitions with more than two rows or columns, however, the tools we used are restricted and we need more ideas to achieve this goal. We leave this result as an open question.

\subsubsection{Some new intermediate quantum computing models}
\label{intermediate}

We saw in Section \ref{sep} that for a ball permuting circuit $C$, the amplitude $\bra {123 \ldots n} C\ket {123\ldots n}$ can be additively approximated using a $\DQC 1$ algorithm. Moreover, last Section asserted that for specific partitions $\lambda$, if $\ket{\psi}$ is a separable state over the partition $V_\lambda\tensor X_\lambda$, then it is $\BQP$-complete to read the amplitude $\bra {\psi} C \ket{\psi}$ within additive error. In this Section we provide evidence that the ball-permuting model along specific subspaces of $\C S_n$ yields a model of computing that is intermediate between $\DQC 1$ and $\BQP$.

Suppose that instead of the computational basis, we initialize the ball-permuting model with the projection of the state $|123\ldots n\rangle$ onto an irrep $\lambda$. Denote this (normalized) state by $\ket{\lambda}$. Then we apply a sequence of ball-permuting gates, and at the end of the computation we sample a pair of tableus in $V_\lambda \tensor X_\lambda$ according to the Young-Yamanouchi basis (see Appendix \ref{YoungYamanouchi} for a review). Inspired by this model, we formally define the following complexity class:

\begin{definition}
$\Samp\QBall(\lambda)$ is the class of problems that are solvable in polynomial time using polynomially many samples from the above model.
\end{definition}

The exact power of this model is unknown, however, we motivate reductions to an interesting complexity class that is intermediate between $\DQC 1$ and $\BQP$.
We can also define the following computational problem:

\begin{definition}
($\QBall(\lambda)$) given a ball permuting circuit $C$, and a partition $\lambda$, evaluate an additive approximation to $\bra{\lambda}C\ket{\lambda}$.
\end{definition}

Interestingly, the initial state $|123\ldots n\rangle$ is equally supported on all of the irreps $\lambda$ of the decomposition $\C S_n \cong \bigoplus_{\lambda \vdash n} V_\lambda \tensor X_\lambda$, and moreover it is maximally entangled over each partition $V_\lambda \tensor X_\lambda$. Also, notice that (as mentioned in the last Section) all ball-permuting operators only act on $V_\lambda$ parts, and act trivially on $X_\lambda$'s, the multiplicity spaces \footnote{Intuitively, the situation resembles a quantum/classical hybrid memory in the sense of \cite{kuperberg2003capacity}, where classical bits enumerate the name $\lambda$ corresponding to the particular irrep, and the quantum memory corresponds to a bipartite system-environment Hilbert space $V_\lambda \tensor X_\lambda$; $V_\lambda$ plays the role of a system, and $X_\lambda$ is its environment which is inert and also maximally entangled with the system.}.

In short, we find that the model operates on one half of a maximally entangled state, while leaving the other half untouched. At the end of the computation, one measures both halves of the maximally entangled state. 
To investigate the power of this model, as a toy model, we consider a restricted model of quantum computing. Imagine we have access to only one half of the state:
$$
|\psi\rangle =\frac{1}{2^{n/2}} \displaystyle\sum_{x\in\{0,1\}^n} \ket{x}\ket{x}.
$$
Suppose we can perform standard unitary quantum computation on the left half of the state, but one cannot access the right half of the state. At the end of the computation, then we get to measure both halves of the state in the computational basis. 

This models seems very similar to the class $\DQC1$, because if $C$ is a quantum circuit on the $n$ active qubits, then:

$$
\langle \psi|C|\psi\rangle = \dfrac{Tr(C)}{2^n},
$$

\noindent the normalized trace. This is simply because the reduced density matrix of the left half of the state is maximally mixed. We have seen in Section \ref{sep} that evaluation of such a trace within additive error is complete for the class $\DQC 1$. Therefore if we define the trace computing class:

\begin{definition}
(Trace computing quantum polynomial time) $\TQP$ is the class of problems that are polynomial-time reducible to additive approximation of $\langle \psi|C|\psi\rangle$, the normalized trace of the matrix. Also $\Samp \TQP$ is the class of problems that are solvable with high probability using polynomially many samples from $C|\psi\rangle$, in the computational basis.
\end{definition}

\noindent Then we trivially find $\TQP=\DQC1$. However, the class $\Samp\TQP$ seems to be more powerful than $\DQC1$. The reason is that one gets to measure the right half of the state at the end of the computation; this is not an ability one has in $\DQC1$. To put it another way, $\DQC1$ is defined using maximally mixed states as inputs, which is equivalent to performing your computation on a random basis state $\ket{x}$. $\Samp\TQP$ is equivalent to performing your computation on a random $\ket{x}$, but at the end of the computation, you get to learn which $\ket{x}$ you started with. As a result, it appears that $\Samp\TQP$ is intermediate between $\DQC1$ and $\BQP$. Note that $\Samp\TQP$ is $\BQP$-universal under postselection as well.

We first observe that:

\begin{theorem} 
If $\lambda$ is a partition with two equally sized rows, then $\QBall(\lambda)\in\TQP$.
\label{tworowQBall}
\end{theorem}

\begin{proof} (Sketch) following the proof of theorem \ref{DQC1thm} we use a compressed and local representation of standard tableaux with two rows of length $n/2$ using strings of bits. Specifically, suppose we represent standard tableaux on two equal-sized rows by bit strings of length $n$. The $i$th entry represents whether the number $i$ appears on the top or bottom row. There are $2^n$ such strings and $2^n/\operatorname*{poly}(n)$ valid tableaux, so although many of the strings do not represent valid tableaux, a $1/\operatorname*{poly}$ fraction do represent valid tableaux. So this encoding is very efficient. Furthermore, it is local in the sense that to exchange two labels, one simply exchanges their corresponding bits. Finally, it is easy to test if a string is a valid encoding. Additionally, it one can compute using $O(\log n)$ ancillas the axial distance (as defined in Appendix \ref{YoungYamanouchi}) between two labels in the tableau. Following the techniques of Appendix \ref{trace}, the existence of such an encoding implies membership in $\DQC1$, as one can  simulate the action of ball-permuting gates on these basis, and use a $\DQC 1$ procedure to evaluate an additive approximation to $\bra{\lambda}C\ket{\lambda}$ .See Appendix \ref{trace} for details.

\end{proof}

We conjecture that the power of $\QBall$ on all such input states $\ket{\lambda}$ is contained in $\TQP$. However, it is difficult to prove this fact because $\DQC1$ is defined using qubits, while the $\ket{\lambda}$ basis is labeled by the Young-Yamanouchi basis (See Appendix \ref{YoungYamanouchi}). As in Section \ref{sep}, in order to prove $\QBall$ is in $\TQP$ on this input states, we would need to find an encoding of the Young-Yamanouchi basis which is both extremely compressed and local for arbitrary Young diagrams. We mention this is an open problem in Appendix \ref{openproblems}.

Next, we build a connection between the two classes $\Samp\QBall(\lambda)$ and $\Samp\TQP$:

\begin{theorem}
If $\lambda$ is a partition with two equally sized rows, $\Samp\QBall(\lambda) \subseteq \Samp\TQP$.
\end{theorem}

\begin{proof} (Sketch) as in the proof of theorem \ref{tworowQBall}, we use a compressed and local representation of tableaux with two equally sized rows using binary strings to simulate the action of ball permuting gates on $V_\lambda$ in a succinct space, and sample from the output distribution in the computational basis. After postprocessing, with high probability the sampled string corresponds to a valid sample from $\Samp\QBall(\lambda)$.
\end{proof}

$\Samp\TQP$ is a restricted model of computation on qubits, and is an interesting model on its own right. We can immediately observe the following:

\begin{theorem}
$\DQC 1\subseteq \Samp\TQP \subseteq \BQP$.
\end{theorem}

\begin{proof}
$\Samp\TQP \subseteq \BQP$ is immediate. To see $\DQC 1\subseteq \Samp\TQP$ do a $\DQC 1$ computation, assuming optimistically that the first active bit is in the pure state $\ket{0}$.  Then, at the end, when we measure we will find out whether or not the assumption was correct, and it will have been with probability $1/2$.
\end{proof}

\subsection{The power of $\RQBall$ }
\label{Post}

In this section, we consider the complexity class $\RQBall$, in which the partial swap interactions are constrained by the physics of distinguishable particles. In particular, as discussed briefly in Section \ref{model} and in detail in Appendix \ref{models}, the amplitudes $c$ and $s$ are constrained by the relative velocities of the interacting particles, and only those particles which collide can interact. This physically corresponds to a two-dimensional integrable quantum field theory of n distinguishable particles on a line.

Before presenting these results, we first note that we do not expect this model to be capable of universal quantum computing, because the dimensionality of the space of unitary operations allowed in this model is too small:

\begin{theorem}
Let $Q_n$ be the Lie group generated by \textit{planar} Yang-Baxter quantum circuits over $n$ labels, then $Q_n$ as a manifold is isomorphic to the union of $n!$ manifolds, each with dimension at most $n$.
\end{theorem}

\begin{proof}(Sketch) This follows by noticing that if the operators satisfy the Yang-Baxter equation, then the final unitary operator only depends on $n$ real parameters, namely the initial momenta of the particles.
\end{proof}

In other words, if the interactions of the ball-permuting model satisfy the Yang-Baxter equation, then the group of unitary operators generated on $n$ labels corresponds to a Lie group of linear dimension. In contrast the standard definition of $\BQP$ allows unitaries of dimension exponential in the number of qubits. This makes it unlikely that one could do universal quantum computation in this model; however, it does not imply efficient classical simulation.

To the contrary, in this section we establish some hardness results for classical simulation of the class $\RQBall$ following the strategy of Aaronson and Arkhipov \cite{aaronson2011computational}, Bremner-Shephard-Josza ~\cite{JozsaShepherd} and others.
In particular, we show that a classical computer cannot sample from the same distribution as these circuits unless the polynomial hierarchy collapses to the third level. 
Our argument makes use of two technical concepts.
The first is intermediate demolition measurements. A \emph{demolition measurement} is a measurement of the label at a particular location, where after the label is measured, it is removed from the computation (``demolished"), i.e. the Hilbert space becomes $\C S_{n-1}$ rather than $\C S_n$.
The second is the concept of postselection (discussed in greater detail in Appendix \ref{Postselection}). 
Here postselection is the (non-physical) ability to specify the outcome of a measurement in advance. 
We show that postselected $\RQBall$ circuits with intermediate demolition measurements are capable of universal postselected quantum computation.
By the techniques of \cite{aaronson2011computational,JozsaShepherd}, this immediately implies that a classical computer cannot efficiently sample from the same probability distribution unless the polynomial hierarchy collapses.  
The basic reason is that postselected quantum computation (denoted by the complexity class $\Post\BQP$) is more powerful that postselected classical computation (denoted by the class $\Post\BPP$) assuming $\PH$ does not collapse.
Therefore if a classical computer could sample from the same distribution, then by postselecting the simulation of the quantum device, that would solve a $\Post\BQP$ problem in $\Post\BPP$, which causes $\PH$ to collapse.

More formally, we consider the following complexity class:
\begin{definition}
$\Post \RQBall$ is the class of decision problems that are solvable in polynomial-time using postselected Yang-Baxter ball collision circuits with intermediate demolition measurements, where in the input state is a special inital state described in Appendix \ref{PostProgramming}. That is, $\Post\RQBall$ is the set of languages $L$ such that there exists a postselected quantum $\RQBall$ circuit, and a poly-time-computable set of permutions $P$, such that if $\sigma$ is the permutation measured at the end of the computation, and $C$ is the event all intermediate measurments yield their postselected values, then

\begin{itemize} 
\item[] 1) $\operatorname*{Pr} [C]>0$. 

\item[] 2) if $x\in L$, $\operatorname*{Pr} [\sigma \in \tilde {P} | C]\geq 2/3$. 

\item[] 3) if $x\notin L$, $\operatorname*{Pr} [\sigma \in \tilde {c} | C]\leq 1/3$.
\end{itemize}

Likewise define $\Meas\RQBall$ to be the class of sampling problems that are solvable by polynomial time uniform Yang-Baxter ball collision circuits on arbitrary initial states and intermediate demolition measurements.
\label{PostScatdef}
\end{definition}

The central observation is the following theorem, stating that although the ball scattering model might be strictly weaker than $\BQP$, the postselected version is equal to $\Post\BQP$, which is equal to $\PP$ \cite{aaronson2005quantum}.

\begin{theorem}
$\Post \RQBall=\Post \BQP= \PP$.
\label{PPeqPostRQBall}
\end{theorem}

\begin{proof} (sketch). In Appendix \ref{PostProgramming} we show how to use intermediate demolition measurements in $\RQBall$ circuits to simulate standard $\QBall$ circuits, assuming one can postselect the outcomes of the intermediate demolition measurements. Since $\QBall=\BQP$ this implies $\BQP \subseteq \Post\RQBall$ in an encoded sense. By postselecting this encoded $\BQP$ simulation we can simulate arbitrary $\Post\BQP$ computations in $\Post\RQBall$ as well.
\end{proof}

From this, using the techniques of \cite{aaronson2011computational,JozsaShepherd}, we immediately have the following corollary:

\begin{corollary}
There is no polynomial time randomized algorithm to simulate $\Meas\RQBall$ to within multiplicative constant error, unless $\PH$ collapses to its third level.
\end{corollary}
\begin{proof} (sketch). If a $\BPP$ machine could sample from the same distribution as a $\Meas\RQBall$ circuit, then by postselecting the simulation, one could simulate $\Post\RQBall=\PP$ in $\Post\BPP$, and hence $\Post\BPP=\PP$. But by \cite{aaronson2011computational,JozsaShepherd} this would imply the collapse of $\PH$, because then $\PH\subseteq \P^\PP = \P^{\Post\BPP} \subseteq \Delta_3$. Here the first containment is Toda's theorem \cite{toda1991pp} and the third follows from the fact  $\Post\BPP \subseteq  \Sigma^3_\P$ \cite{arora2009computational}.
\end{proof}

Therefore, despite the fact that the dimensionality of the space of unitaries in $\RQBall$ is small, $\RQBall$ computations cannot be efficiently simulated classically (up to multiplicative error) unless the polynomial hierarchy. Hence the $\RQBall$ model seems to be intermediate in power between $\BPP$ and $\BQP$.

\section{Computational complexity of the classical ball permuting model}
\label{classical}

To understand which features of the quantum ball-permuting  model actually depend on quantum mechanics, in this section, we define a classical model of ball permuting and analyze its computational power.  In this model input is a list of swaps in $S_n$, and the output is a permutation of $n$ labels. We consider three models - deterministic, randomized and nondeterministic for the application of swaps.

Before introducing our classical ball permuting models, we will first review some classical complexity theory.
Recall that $\L$ is the class of decision problems that are solvable on a log-space Turing machine. $\BPL$ is the class of decision problems that are solvable with bounded probability of error on a probabilistic logarithmic-space Turing machine, while $\Almost \L$ is defined as the class of languages that are decidable with bounded probability of error by a logarithmic-space Turing machine with access to a random oracle with probability $1$. Interestingly, $\Almost\L$ contains $\BPL$, but the converse is not known. The issue is that an $\Almost\L$ machine can use the same specific polynomially-long random string over and over, which is not an ability that a $\BPL$ machine obviously has. Indeed, $\Almost\L$ is not even known to be contained in $\P$, though it is contained in $\BPP$.

A reversible Turing machine is a Turing machine such that for any of its computations, the nodes of its infinite configuration graph have in-degree and out-degree of at most $1$. Reversible log-space or $\Rev\L$ is the class of decision problems that are solvable by reversible log-space Turing machines. Lange, McKenzie, and Tapp \cite{lange2000reversible} proved the following important result:
\begin{theorem}
 $\L=\Rev\L$.
\label{L=RevL}
\end{theorem}

We now define the classical ball-permuting model.

\begin{definition}

\begin{itemize}
\item[]
\item $\DBall$ is the class of languages that are decidable by an $\AC^0$ machine that can make a single query to a``deterministic ball-permuting oracle''. Such an oracle takes as input a polynomially-long list of swaps $(i_1,j_1),(i_2,j_2),\ldots$ and returns the permutation obtained by applying those swaps in order to the initial permutation $(1,2,\ldots,n)$.

Also, $\DBall_{adj}$ is the subclass of $DBall$ where all swaps are adjacent elements only.

\item $\RBall$ is the class of languages that are decidable, with bounded probability of error, by an $AC^0$ machine with a single query to a ``randomized ball-permuting oracle''.  Such an oracle takes as input a polynomially-long list of swaps $(i_1,j_1),(i_2,j_2),\ldots$ along with a list of probabilities $p_1,p_2,\ldots$. It returns the permutation obtained by swapping $(i_1,j_1)$ with independent probability $p_1$, then swapping $(i_2,j_2)$ with independent probability $p_2$, and so on, starting from the initial permutation $(1,2,\ldots,n)$.

$\RBall_{adj}$ is the subclass where the swaps are restricted to be adjacent ones only, and $\RBall^\star_{adj}$ is the subclass where all probabilities are nonzero and also not equal to one.
\end{itemize}

\end {definition}

We find that deterministic ball permuting captures $\L$, while randomized ball permuting defines a class which is intermediate between $\BPL$ and $\Almost\L$. We include sketches of these proofs, but details can be found in Appendix \ref{ballpermutingoracles}.

\begin{theorem}
$\DBall=\DBall_{adj}=\L=\Rev\L$ 
\end{theorem}

\begin{proof}(Sketch) For $\DBall_{adj} = \DBall$, just observe that we can easily implement any swap we want as a sequence of adjacent swaps. For $\L \subseteq\DBall$, we show how to simulate reversible log-space with swaps and then appeal to theorem \ref{L=RevL}. For $\DBall\subseteq\L$, we use $\L$ subroutines to keep track of each label one-by-one.
\end{proof}

\begin{theorem}
$\L \subseteq \BPL \subseteq \RBall=\RBall_{adj}\subseteq \Almost \L\subseteq \BPP$. However, if we let $\RBall (2)$ be the class where the $\AC^0$ machine is allowed to make two adaptive queries to the ball-permuting oracle, then $\RBall(2)=\Almost\L$.
\end{theorem}

\begin{proof}(Sketch) $\RBall=\RBall_{adj}$ is proved similarly to the previous theorem. Also, for $\RBall\subseteq\Almost\L$: we just run the $\DBall\subseteq\L$ simulation, except that any time we need to apply a probabilistic swap, we use the random oracle to decide whether to make it or not.

Finally, for $\BPL\subseteq \RBall$: we run the $\L \subseteq\DBall$ simulation, except
that whenever the $\L$ machine wants a random bit in its memory, we use
random swaps on every adjacent pair of elements $((1,2), (3,4), \ldots)$
to produce one.
\end{proof}

Next, in order to understand the power of nondeterministic ball-permutation, we investigate the problem of deciding whether or not the output of an $\RBall$ computation has positive support on a specific permutation:

\begin{definition}
($\Ball$) Given a polynomially-long list of swaps $(i_1, j_1), \ldots, (i_m,j_m)$, and a list of indepedent probabilities $p_1, \ldots, p_m$ as inputs of a randomized ball-permuting oracle, decide if there is a positive probability that the oracle outputs a specific traget permutation $\sigma \in S_n$. $\Ball_{adj}$ is the restriction of the language when all the input swaps are adjacent ones, and $\Ball^\star_{adj}$ is the language when we further restrict the probabilities to be strictly between $0$ and $1$.
\end{definition}

First we prove that the problem is $\NP$-complete for general swaps. We prove this by a reduction from ``word problem for the product of permutations ($\WPPP$)'' which is known to be $\NP$-complete. More precisely:

\begin{definition}
($\WPPP$) Given the set $\{1,2,3, \ldots, n\}=: [n]$, an ordered list of subsets $S_1, S_2, \ldots, S_m \subseteq [n]$ with $m=\Poly(n)$, and a target permutation $\tau$ on $[n]$, the problem is to decide whether there exist permutations $\pi_1, \pi_2, \ldots, \pi_m$, with each $\pi_j$ acting on the labels of $S_j$ only and identity on the others, such that the combination $\pi_1\circ \pi_2\circ \ldots\circ \pi_m=\tau$.
\end{definition}

\begin{theorem}
(Garey, Johnson, Miller, and Papadimitriou \cite{garey1980complexity}) $\WPPP$ is $\NP$-complete.
\label{GJMP}
\end{theorem}

\begin {theorem}
$\Ball_{adj}$ and $\Ball$ are polynomial-time reducible to each other, and furthermore they are complete for the class $\NP$. 
\end{theorem}

\begin{proof}(Sketch) To see polynomial-time equivalence between the two languages, we use adjacent swaps with probability $1$ to simulate nonadjacent ones. To prove completeness for $\NP$ we use theorem \ref{GJMP}.
\end{proof}

On the other hand, we find that the non-deterministic class $\Ball^*_{adj}$ is contained in $\P$. Our proof uses a reduction to planar case of the edge disjoint path problem.

\begin{definition}
(The edge disjoint path problem $\EDP$) Given a directed graph $G$, with source and sink nodes $(s_1, t_1), (s_2, t_2),\ldots,(s_m, t_m),$, decide if there are paths from $s_i$ to $t_i$ for all $i \in [m]$, such that all the paths are edge-disjoint.
\end{definition}

\begin{theorem}
(Wagner and Weihe\cite{wagner1995linear}) $\EDP$ for the case of planar graphs is decidable in linear time.
\label{Wagner}
\end{theorem}

\begin{theorem}
$\Ball^\star_{adj}\in\P$
\end{theorem}

\begin{proof}(Sketch) we use a linear time reduction to $\EDP$. Given a list of swaps we construct a planar directed graph for which an edge disjoint path exists if and only if the target permutation is generated in the non-deterministic swaps.
\end{proof}

In other words, given a list of swaps and a target permutation, it is in general $\NP$-hard to decide if there is a way of constructing the target permutation out of the given swaps. If all the swaps are transpositions, the problem is decidable in linear time. Observe that our result about the language $\Ball$ gives an independent proof for:

\begin{corollary}
The edge disjoint path problem in the non-planar case is $\NP$-complete.
\end{corollary}

\section{Discussion and Conclusions}
We have explored several models of ball-permuting computation. We found the power of $\QBall$ depends on the input states allowed in the model, and we partially classified which input states allow $\QBall$ to perform universal quantum computation using the representation theory of $S_n$. Additionally, we found that even the highly restricted model of $\RQBall$ cannot be efficiently simulated classically unless $\PH$ collapses. This is surprising as the model $\RQBall$ admits unitary transformations of only polynomial dimension. Classically, we found that $\RBall$ is of power intermediate between $\BPL$ and $\Almost\L$.

A natural open problem is to finish classifying the power of $\QBall$ on all input states. While we nearly completed this task, the tools we used for the classification are restricted in the sense that they can only take care of special subspaces and not the others.
Another natural open problem is to determine if $\QBall$ with the starting state $\ket{12\ldots n}$ is contained in $\DQC1$ as a decision class.
A list of further open problems can be found in Appendix \ref{openproblems}.

\section{Acknowledgements}
We thank Robin Kothari for suggesting connections to the class $\DQC 1$ and Luke Schaeffer for helpful discussions regarding succinct encodings of permutations and tableaux. We thank an anonymous referee for helpful comments about integrable models. S.M. thanks Hossein Esfandiary for suggesting connections to factorial number system and Usman Naseer \& Hamed Pakatchi for useful comments on the physics.

\begin{appendices}

\section{Models of the two dimensional integrable quantum theory}
\label{models}

In this appendix, we review the basic structure of the solution to the two dimensional integrable model and demonstrate how this model is captured by the quantum ball-permuting model.

We imagine that in the far past, a number of free particles are initialized on a line, moving towards each other, and in the far future, an experimenter measures the asymptotic wave-function that results from scattering. Throughout this work we assume that all particles are distinguishable.

The basic model we focus on is given by the repulsive delta interactions model ~\cite{yang1967some} of quantum mechanics. In this model, elastic hard particles with known momenta are scattered off each other. The conserved quantities of this model are $\sum_j p^{2 k+1}_j$, and $\sum_j p_j^2/m_j$, for $k\geq 0$, where the $m_j$'s are the masses of the particles. Therefore, particles of different mass don't interact, and after the interaction the set of output momenta equals the set of input momenta.

\begin{figure}[tp]
\begin{center}
\includegraphics[height=4.0in]{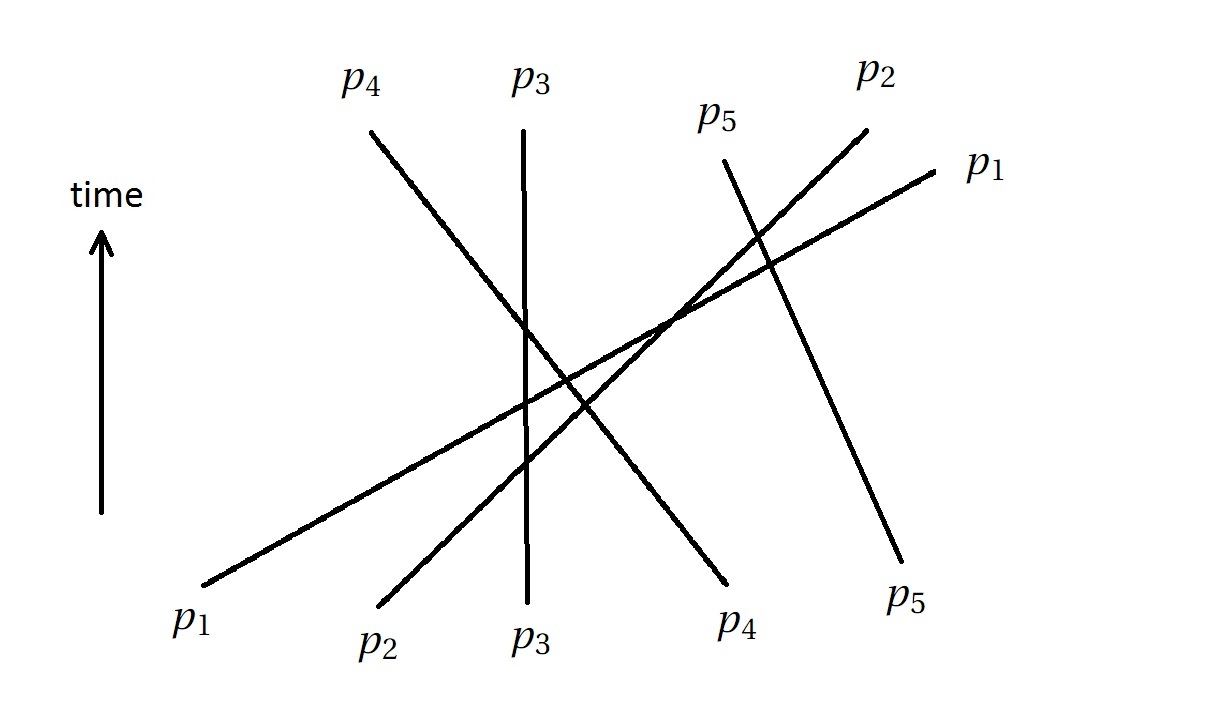}
\caption
{Five particles with momenta $p_1, p_2, \ldots, p_5$ moving on a line. At each intersection between these lines there is an interaction modeled by a ball-permuting gate $R(\cdot, \cdot)$. Each line is assigned to a fixed momentum throughout and the rapidity parameter of the gates is equal to the difference between the momenta incident on each intersection.}
\label{fig22}
\end{center}
\end{figure}

\begin{figure}[tp]
\centering
\includegraphics[height=3.5in]{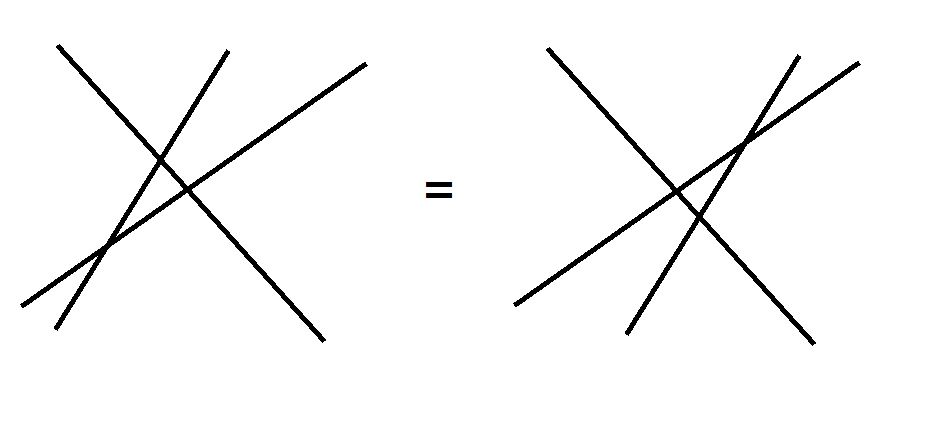}
\caption{Pictorial demonstration of the Yang-Baxter equation. Both interactions give the same unitary operator as quantum circuits, if the input momenta are equal.
}
\label{fig23}
\end{figure}

The particles are asymptotically free, meaning that they move freely and do not interact until reaching to the short range of contact. Denote the position of these particles by $x_1, x_2, \ldots, x_n$ and the range of interaction as $r_0$, then for the asymptotically free regime, we assume $|x_i-x_j|\gg r_0$. The interaction consists of at most $\dfrac{n(n-1)}{2}$ terms, one for each pair of particles. For each pair of particles, the interaction is modeled by the delta function of the relative distance between them. If no particles are in contact, then we just have a free Hamiltonian, and each contact is penalized by a delta function. The functional form of the Schr\"odinger's equation is written as:

$$
i\hbar\dfrac{\partial}{\partial t} \psi(x_1, x_2, \ldots, x_n, t) = \Big [\sum_{j=1}^n-\dfrac{\partial ^2}{\partial x_j^2}+2 c\sum_{i<j}\delta(x_i-x_j)\Big ]\psi(x_1, x_2, \ldots, x_n, t)
$$

\noindent Here $c>0$ is the strength of the interactions. As species of unequal mass do not interact, only particles of same mass are considered. The Hilbert space is indeed $(\mathbb{C}^\infty)^{\otimes n}$. Using the Bethe ansatz ~\cite{staudacher2012review} for spin chain models, a solution for the eigenfunction with the following form is considered:

$$
\psi(x_1,\ldots,x_n)=\sum_{\tau,\pi \in S_n} \mathcal{A}^\tau_\pi \theta_\tau (x_1, \ldots, x_n) \exp{[i(x_{\tau_1} p_{\pi_1}+\ldots+x_{\tau_n} p_{\pi_n})]}
$$

\noindent $\theta_\pi : \mathbb{R}^n \rightarrow \{0,1\}$ is an indicator function which is set to $1$ whenever its input $x$ satisfies $x_{\pi(1)}<x_{\pi(2)}<\ldots <x_{\pi(n)}$, and otherwise zero. $p_j$ for $j\in [n]$ are constant parameters, and can be viewed as the momenta. The proposed solution must be a continuous function of the positions and also one can impose a boundary condition for $x_{\pi(t)}=x_{\pi(t+1)}$ for $t\in [n]$ on the derivative of the wave-function. Applying these boundary conditions, one can get linear relations between the amplitudes:

$$
\mathcal{A}^{t \circ \tau}_{\pi}=\dfrac{-ic \mathcal{A}^\tau_\pi + V_{\tau, t} \mathcal{A}^\tau_{t \circ \pi}}{ic + V_{\tau, t}}
$$

\noindent Here $t \circ \pi$ is a new permutation that results from the swapping of the $t$ and $t+1$'th labels in the permutation $\pi$, while $V_{\tau, t}=p_{\tau(t)}-p_{\tau(t+1)}$. The above linear map has a simple interpretation: two particles with relative velocity $V$ collide with each other and reflect with amplitude $\dfrac{-ic}{ic+V}$, or otherwise, with amplitude $\dfrac{V}{ic+V}$ they tunnel through without any interaction. In any case, the higher momentum passes through the lower momentum and starting from a configuration $x_1 < x_2 < \ldots <x_n$ for $n$ particles with any momenta, the wave-function will end up in a configuration with momenta in the increasing order. 

Each of these pairwise scatterings can be viewed as a local quantum gate, and the collection of scatterings can be viewed as a quantum circuit. To see this, consider an $n!$ dimensional Hilbert space for $n$ particles with orthonormal basis $\{ |\sigma \rangle : \sigma \in S_n \}$. Assume an initial state of $|1,2,3,\ldots, n\rangle$, with momenta $p_1, p_2, \ldots, p_n$. These momenta and the initial distances between the particles determine in what order the particles will collide. It is instructive to view the trajectory of the particles as $n$ straight lines for each of these particles in an $x/t$ plane. Time goes upwards and the intersection between any two lines is a collision. In each collision, either the labels of the two colliding particles are swapped or they are left unchanged.  Particles with zero relative velocity do not interact, as lines with equal slope do not intersect. Suppose that the first collision corresponds to the intersection of line $t$ with line $t+1$. So in particlular, $p_t>p_{t+1}$. Then the initial state is mapped to:

$$
\Big |1,2, \ldots, n\Big \rangle \rightarrow \dfrac{-ic}{ic + V_{t,t+1}}\Big |1,2, \ldots t, t+1, \ldots, n\Big \rangle + \dfrac{V_{t,t+1}}{ic + V_{t,t+1}}  \Big |1,2, \ldots t+1, t, \ldots, n\Big \rangle,
$$

\noindent where, $V_t= p_t - p_{t+1}$. This map can be viewed as a $n!\times n!$ unitary matrix:

$$
R(p_t- p_{t+1}, t):=R(p_t, p_{t+1},t):=\dfrac{-ic}{ic + V_{t, t+1}} I + \dfrac {V_{t,t+1}}{ic + V_{t, t+1}} L_{(t, t+1)}.
$$

\noindent Here $I$ is the $n! \times n!$ identity matrix, and $L_{(t,t+1)}$ is the $n!\times n!$ matrix which transposes the $t$ and the $t+1$'th labels of the basis states. $R(u,t )$ acts only on the $t$ and the $t+1$'th particles only. $u$ is the velocity of the $t$'th particle relative to the $t+1$'th particles. These are exactly the unitary ball-permuting gates.

Given $n$ particles, with labels $i_1, i_2, \ldots, i_n$, and momenta $p_1, p_2, \ldots, p_n$, we can obtain a quantum circuit with gates $R(u_1,t_1), R(u_2,t_2),\ldots, R(u_m,t_m)$, one for each intersection of the straight lines. The scattering matrix in this theory is then given by the product $S=R(u_m,t_m), R(u_{m-1},t_{m-1}),\ldots, R(u_1,t_1)$. 

An important ingredient of these quantum gates is the so-called Yang-Baxter equation ~\cite{yang1967some, baxter1972partition}, which is essentially the factorization condition, and is the analogue of the associativity of Zamalodchikov algebra. The Yang-Baxter equation is the following three particle condition:

$$
R(u,t) R(u+v, t+1) R(v,t)= R(v, t+1) R(u+v, t) R(u, t+1).
$$

Basically, the Yang-Baxter equation asserts that the continuous degrees of freedom like the initial positions of the particles do not change the outcome of a quantum process, and all that matters is the particles' relative configurations. See Figure \ref{fig23}. For the line representation of the Yang-Baxter equation, see Figure \ref{fig22}. Consider three balls with labels $1, 2, 3$, initialized with velocities $+u$, $0$ and $-u$, respectively. If we place the middle ball very close to the left one, the order of collisions would be $(1,2)\rightarrow (2,3) \rightarrow (1,2)$. However, if the middle one is placed very close to the third ball, the order would be $(2,3)\rightarrow (1,2) \rightarrow (2,3)$. The Yang-Baxter equation asserts that the output of the collisions is the same for the two cases. Therefore, the only important parameters are the relative configurations, and the relationships between the initial velocities.

\subsection{Semi-classical Model}

It is not conventional to do a measurement at the middle of a scattering process. However, in the discussed integrable models, particle interactions occur independently from each other, and the scattering matrix is a product of smaller scattering matrices. Moreover, in the regime that we consider, no particle creation or annihilation occurs. Therefore, it is reasonable to assume a semi-classical model, where the balls move according to actual trajectories, and it is possible to track and measure them in between and stop the process whenever we want between collisions. According to this assumption quantum effects occur only at the collisions and measurements.

\section{Complexity Classes with Postselection}
\label{Postselection}

Here we define complexity classes with postselection. Intuitively, these are the complexity classes with efficient verifiers with free retries. That is an algorithm which runs on the input, and in the end will tell you whether the computation has been successful or not. The probability of successful computation can be exponentially small.

\begin{definition}
Fix an alphabet $\Sigma$. $\Post\BQP$ ($\Post \BPP$) is the class of languages $L \subset \Sigma^\star$ for which there is a polynomial time quantum (randomized) algorithm $\mathcal{A}: \Sigma^\star \rightarrow \{0,1\}^2$, which takes a string $x\in \Sigma^\star$ as an input and outputs two bits, $\mathcal{A}(x)= (y_1, y_2)$, such that:

\begin{itemize}

\item[] 1) $\forall x \in \Sigma^\star, \operatorname*{Pr}(y_1 (x)=1)> \dfrac{1}{\exp(n^{O(1)})}$.

\item[] 2) If $x \in L$ then $\operatorname*{Pr}(y_2 (x)=1| y_1(x)=1)\geq \dfrac{2}{3}$

\item[] 3) If $x \notin L$ then $\operatorname*{Pr}(y_2 (x)=1| y_1(x)=1)\leq \dfrac{1}{3}$
\end{itemize}
\end{definition}

\noindent Here $y_1$ is the bit which tells you if the computation has been successful or not, and $y_2$ is the actual answer bit. The conditions $2)$ and $3)$ say that the answer bit $y_2$ is reliable only if $y_1=1$. In this work, we are interested in the class $\Post\BQP$. However, $\Post \BPP$ is interesting on its own right, and is equal to the class $\BPP_{path}$, which a modified definition of $\BPP$, where the computation paths do not need to have identical lengths. $\Post \BPP$ is believed to be stronger than $\BPP$, and is contained in $\BPP^{\NP}$, which is the class of problems that are decidable on a $\BPP$ machine with oracle access to $\NP$ or equivalently $\SAT$.

Due to a result by Aaronson, $\Post \BQP$ is equal to the complexity class $\PP$:

\begin{theorem}
(Aaronson \cite{aaronson2005quantum}) $\Post\BQP=\PP$.
\end{theorem}

\noindent Firstly, because of $\P^{\PP}=\P^{\#\P}$ as a corollary to the theorem, with oracle access to $\Post\BQP$, $\P$ can solve intricate counting tasks, like counting the number of solutions to an $\NP$-complete problem. The implication of this result for the current work is that if a quantum model, combined with postselection is able to efficiently sample from the output distribution of a $\Post\BQP$ computation, then the existence of a  randomized scheme for approximating the output distribution of the model within constant multiplicative factor is ruled out unless $\PH$ collapses to the third level. This point is going to be examined in Section \ref{phcol}.

\section{The Classical Yang-Baxter Equation}
\label{classicalyangbaxter}
The classical ball permuting model can be alternatively viewed as a stochastic matrices. Consider the set of probability distributions on permutations $\mathbb{V}_n:=\{V\in \mathbb{R^{+}}^{n!}, \sum_j V_j=1\}$. Each entry of $V$ corresponds to a permutation, and its content is the probability that the permutation appears in the process. Denote the basis $\{|\sigma\rangle: \sigma \in S_n\}$ for this vector space. The basis $|\sigma\rangle$ has probability support $1$ on $\sigma$ and $0$ elsewhere. Now a permutation can be viewed as a doubly stochastic $n! \times n!$ matrix with the map $L(\pi)|\sigma\rangle= |\pi\circ \sigma\rangle$. Denote the swap matrix corresponding to $(i,j)$ by $L_{i,j}:=L(i,j)$. A probabilistic swap $(i,j)$ with probability $p$ is therefore given by the convex combination:

$$
R_{i,j} (p)=(1-p) I+p L_{i,j}
$$

One can then formalize the Yang-Baxter equation for three labels (balls). Yang-Baxter (YB) equation is a restriction on the swap probabilities in such a way that the swaps of order $(1,2), (2,3) , (1,2)$ give the same probability distribution as the swaps $(2,3), (1,2), (2,3)$:

\begin{equation}
R_{1} (p_1) R_{2} (p_2)  R_{1} (p_3) =  R_{2} (p'_1) R_{1} (p'_2)  R_{2} (p'_3)
\label{YB}
\end{equation}

We want to solve this equation for $1\geq p_1, p_2, \ldots, p'_6 \geq 0$. If we rewrite $R_{i,j} (p)$ with the parameter $x\in [0,\infty)$:

$$
R_{i,j} (x)=\dfrac{1+x L_{i,j}}{1+x},
$$

\noindent then: 

\begin{theorem}
The following is a solution to equation ~\ref{YB}:

$$
p_1=\frac{x}{1+x} \hspace{1cm} p_2=\dfrac{x+y}{1+x+y} \hspace{1cm}p_1=\dfrac{y}{1+y} \hspace{1cm}
$$
$$
p'_1=\dfrac{y}{1+y} \hspace{1cm} p'_2=\dfrac{x+y}{1+x+y} \hspace{1cm} p'_3=\frac{x}{1+x} \hspace{1cm}
$$

\end{theorem}

\begin{proof}
First we need to show that these are indeed a solution to YB equation. Expanding the equation $R_1 (x) R_2 (x+y) R_1 (y)$
we get:

\begin{equation}
\dfrac{1+xy+(x+y)(L_1+L_2)+(x+y)(x L_1 L_2 + y L_2 L_1) + xy (x+y) L_1 L_2 L_1}{(1+x)(1+x+y)(1+y)}
\label{expand}
\end{equation}

Exchanging $L_1$ with $L_2$ and $x$ with $y$ in the left hand side of equation ~\ref{expand} gives the right hand side. Using the identity $L_1 L_2 L_1 = L_2 L_1 L_2$, however equation ~\ref{expand} is invariant under these exchanges.
\end{proof}

\subsection{The Quantum Yang-Baxter Equation}
\label{YangBaxter}

We can discuss the restriction on the angles of the $X$ operators in such a way that they respect Yang-Baxter equation (YBE) of three particles. Therefore, this restricted version can capture the scattering matrix formalism of particles on a line. A solution to the YBE in the scattering models is based on the amplitudes which depend only on the initial state and momenta of the particles. The aim of this section is to justify that the only non-trivial solution to the YBE on the Hilbert space of permutations is the one in which the amplitudes are selected according to the velocity parameters.

Given a vector space $\mathbb{V}^{\otimes n}$, let $H_{ij}$ for $i<j \in [n]$ be a family of two-local operators in $GL(\mathbb{V}^{\otimes n})$ such that each $H_{ij}$ only affects the $ij$ slot of the tensor product, and acts trivially on the rest of the space. Then, $H$ is said to satisfy the parameter independent YBE if they are constant and:

$$
R_{ij} R_{jk} R_{ij}= R_{jk} R_{ij} R_{jk}
$$

Sometimes, we refer to the following as the YBE:

$$
(R\otimes I) (I\otimes R) (R \otimes I)=(I\otimes R) (R \otimes I) (I\otimes R)
$$

Both sides of the equation act on the space $\mathbb{V}\otimes \mathbb{V}\otimes \mathbb{V}$, and $H\otimes I$ acts effectively on the first two slots, and trivially on the other one. Similarly, we can define a parameter dependent version of the YBE, wherein the operator $R:\C \rightarrow GL(\mathbb{V}\otimes \mathbb{V})$ depends on a scalar parameter, and $H$ is said to be a solution to the parameter dependent YBE is according to:

$$
(R(z_1)\otimes I) (I\otimes R(z_2)) (R(z_3)\otimes I)=(I\otimes R(z'_1)) (R(z'_2)\otimes I) (I\otimes R(z'_3))
$$

for some $z_1, z_2, \ldots, z'_3$. We are interested in a solution of parameter dependent YBE with $X(\cdot, \cdot )$ operators. For simplicity of notations, in this part, we use the following operator:

$$
R(z , k):=\dfrac{1}{\sqrt{1+z^2}}+\dfrac{i z}{\sqrt{1+z^2}} L_{(k, k+1)} = X(\tan^{-1} (z), k)
$$

instead of the $X$ operators. The following theorem specifies a solution to the parameter dependent YBE:

\begin{theorem}
Constraint to $z_1 z_2 \ldots z'_3\neq 0$, the following is the unique class of solutions to the parameter dependent YBE, with the $R (\cdot, \cdot) = X (\tan^{-1}(\cdot), \cdot )$ operators:

$$
R(x,1) R(x+y,2) R(y,1)=R(y,2) R(x+y,1) R(x,2),
$$

for all $x,y \in \mathbb{R}$.

\end{theorem}

\begin{proof}
We wish to find the class parameters $z_1, z_2, \ldots, z'_3$ such that the following equation is satisfied:

\begin{equation}
R(z_1,1) R(z_2,2) R(z_3,1)=R(z'_1,2) R(z'_2,1) R(z'_3,2)
\label{PDYBE}
\end{equation}

It is straightforward to check that if $z_1=z'_3$, $z_3=z'_1$ and $z_2=z'_2=z_1+z_2$, then the equation is satisfied. We need to prove that this is indeed the only solution. Let:

$\Gamma:= \sqrt{\dfrac{(1+z'^2_1)(1+z'^2_2)(1+z'^2_3)}{(1+z^2_1)(1+z^2_2)(1+z^2_3)}}.$

If equation ~\ref{PDYBE} is satisfied, then the following are equalities inferred:

\begin{eqnarray*}
&1)&\hspace{3mm}\Gamma . (1- z_1 z_3) = (1- z'_1 z'_3)\\
&2)&\hspace{3mm}\Gamma . (z_1 + z_3) = z'_2\\
&3)&\hspace{3mm}\Gamma . z_2 = (z'_1+z'_3)\\
&4)&\hspace{3mm}\Gamma . z_1 z_2 = z'_2 z'_3\\
&5)&\hspace{3mm}\Gamma . z_2 z_3 = z'_1 z'_2\\
&6)&\hspace{3mm}\Gamma . z_1 z_2 z_3 = z'_1 z'_2 z'_3\\
\end{eqnarray*}

Suppose for now that all of the parameters are nonzero; We will take care of these special cases later. If so, dividing $6)$ by $5)$ and $6)$ by $4)$ reveals:

\begin{eqnarray}
\nonumber
z_1&=&z'_3\\
z_3&=&z'_1.
\label{blahh}
\end{eqnarray}

Again suppose that $z_1 z_3\neq 1$ and $z'_1 z'_3\neq 1$. Then using the equivalences of ~\ref{blahh} in $2)$, one gets $\Gamma=1$, from $2)$ and $3)$:

$$
z_2=z'_2=z_1+z_3,
$$

\noindent which is the desired solution. Now suppose that $z_1 z_3 = 1$. This implies also $z'_1 z'_3 = 1$. Using these in $6)$  one finds $\Gamma . z_2 =z'_2$ and substituting this in $2)$ and $3)$ reveals $\Gamma=1$ as the only solution, and inferring from equations $2), 3), \ldots, 6)$ reveals the desired solution.
\end{proof}

If one of the parameters is indeed $0$, we can find other solutions too, but all of these are trivial solutions. The following is the list of such solutions:

\begin{itemize}
\item If $z_1=0$, then:
\begin{itemize}
\item either $z'_3=0$, which implies $z_3=z'_2$ and $z_2=z'_1$
\item or $z'_2=0$ that implies $z_3=0$ and $\tan^{-1}(z_2)=\tan^{-1}(z'_1)+\tan^{-1}(z'_3)$.
\end{itemize}
\item If $z_2=0$ then $z'_1=z'_3=0$ and $\tan^{-1}(z,_2)=\tan^{-1}(z_1)+\tan^{-1}(z_3)$.
\item If $z_3=0$, then:
\begin{itemize}
\item either $z'_1=0$, which implies $z_1=z'_2$ and $z_2=z'_3$
\item or $z'_2=0$ that implies $z_1=0$ and $\tan^{-1}(z_2)=\tan^{-1}(z'_1)+\tan^{-1}(z'_3)$.
\end{itemize}
\end{itemize}

The solutions corresponding to $z'_j=0$ are similar, and we can obtain them by replacing the primed rapidities with the unprimed rapidities in the above table. There is another corresponding to the limit $z_j \rightarrow \infty$:

$$
X(0,1) X(0, 2) X(0, 1)=X(0,2) X(0, 1) X(0, 2)
$$

Which corresponds to the property $L_{(1,2)}L_{(2,3)}L_{(1,2)}=L_{(2,3)}L_{(1,2)}L_{(2,3)}$ of the symmetric group. From now on, we use the following form of the $H$-matrices:

$$
R(v_1, v_2, k) := R(v_1 -v_2, k),
$$

for real parameters $v_1, v_2$, and the YBE is according to:

$$
R(v_1, v_2, 1) R(v_1, v_3, 2) R(v_2, v_3, 1)=R(v_2, v_3, 2) R(v_1, v_3, 1) R(v_1, v_2, 2).
$$

The parameters $v_j$ can be interpreted as velocities in the scattering model. We can now extend the three label Yang-Baxter circuit to larger Hilbert spaces. 

\begin{definition}
An $m$ gate Yang-Baxter circuit over $n$ labels is a collection of $n$ smooth curves $(x_1(s), s)$ $, (x_2(s), s)\ldots (x_n(s), s)$ where $s\in[0,1]$, with $m$ intersections, inside the square $[0,1]^2$, such that, $0<x_1(0)< x_2(0) <\ldots <x_n (0)<1$, and $x_i(1)$ are pairwise non-equal.

If $\sigma \in S_n$, and $x_{\sigma(1)}(1)<x_{\sigma(2)}(1)<\ldots<x_{\sigma(n)}(1)$, then $\sigma$ is called the permutation signature of the circuit.

We say a Yang-Baxter circuit consists of line trajectories if all of the smooth curves are straight lines.
\end{definition}

Each such Yang-Baxter circuit can be equivalently represented by a set of adjacent permutations. When only line trajectories are considered, the circuit is related to the particle scattering models discussed in Section \ref{models}. The permutation signature in this case is obtained by the momenta of the particles.

\begin{definition}
Let $C$ be a Yang-Baxter circuit of $m$ gates, each corresponding to a transposition $(k_t, k_t +1), t \in [m]$, and the permutation signature $\pi_t, t \in [m]$ at each of these gates. Then if one assigns a real velocity $v_j$ to each line, then the Yang-Baxter quantum circuit for $C$ is a composition of $H(\cdot, \cdot, \cdot)$ operators:

\begin{center}
$$
\hspace{-0.8cm}
 R(v_{\pi_m(k_m)}-v_{\pi_m(k_m)+1}, k_m) \ldots R(v_{\pi_2(k_1)}-v_{\pi_2(k_1+1)}, k_2) R(v_{k_1}-v_{k_1+1}, k_1).
$$
\end{center}

Each of these unitary $H$-matrices is a quantum gate.
\end{definition}

\begin{theorem}
Let $Q_n$ be the Lie group generated by \textit{planar} Yang-Baxter quantum circuits over $n$ labels, then $Q_n$ as a manifold is isomorphic to the union of $n!$ manifolds, each with dimension at most $n$.
\label{YBnonu}
\end{theorem}

\begin{proof}
Fix the velocities $v_1, v_2, \ldots, v_n$. The idea is to demonstrate an embedding of the group generated with these fixed velocities into the symmetric group $S_n$.  Consider any two planar Yang-Baxter quantum circuits $C$ and $C'$, with permutation signatures $\sigma$ and $\tau$, respectively. We show that if $\sigma=\tau$, then $C=C'$.

The underlying circuit of $C$ corresponds to a sequence of transpositions $k_1, k_2, \ldots , k_M$ and $C'$ corresponds to another sequence $l_1, l_2, \ldots , l_N$, such that $k_M\circ \ldots \circ k_2\circ k_1= \sigma$, and $l_N\circ \ldots \circ l_2\circ l_1= \tau$. Then the unitary operators $C$ and $C'$ can be written as a sequence of $H$ operators:

$$
C= R(z_M, k_M) \ldots R(z_2, k_2) R(z_1, k_1)
$$

and,

$$
C'= R(z'_N, l_N) \ldots R(z'_2, l_2) R(z'_1, l_1).
$$

Where the $z$ parameters are the suitable rapidities assigned to each two-particle gate based on the velocities $v_1, v_2, \ldots, v_n$, and the underlying Yang-Baxter circuits. If two sequence of transpositions $k_M\circ \ldots \circ k_2\circ k_1$ and $l_N\circ \ldots \circ l_2\circ l_1$ amount to the same permutation, then there is a sequence of substitution rules among:

1) $b_i ^2 \Leftrightarrow e$

2) $b_i b_j \Leftrightarrow b_j b_i$  if $|i-j|>1$

3) $b_i b_i+1 b_i \Leftrightarrow b_i+1 b_i b_i+1$, for all $i\in [n-1]$

Such that if we start with the string $k_M\circ \ldots \circ k_2\circ k_1$ and apply a sequence of substitution rules, we end up with $l_N\circ \ldots \circ l_2\circ l_1$. All we need to do is to prove that the sequences of unitary gates are invariant under each of the substitution rules. The invariance under each rule is given in the below:

$1)$ If we apply two successive quantum transpositions on the labels $i, i+1$ we will end up with the identity operator. This follows from unitarity $H(z) H(-z)=I, \forall z\in \mathbb{R}$, and planarity of the circuits.

$2)$ Clearly $H(\cdot, i) H(\cdot, j) = H(\cdot, j) H(\cdot, i)$ for $|i-j|>1$, since these are $2$-local gates.

$3)$ This part also follows from the Yang-Baxter equation.

We can then start with the unitary $C= H(z_M, k_M) \ldots H(z_2, k_2) H(z_1, k_1)$  and apply the same substitution rules and end up with $C= H(z'_N, l_N) \ldots H(z'_2, l_2) H(z'_1, l_1)$.

Now let $Q_n (\sigma)$ be the Lie group corresponding to all Yang-Baxter quantum circuits with permutation signature $\sigma$. For each choice of velocities, there is exactly one unitary in this group, so $Q_n (\sigma)$ is locally diffeomorphic to $\mathbb{R}^n$, and $Q_n=\cup_{\sigma\in S_n} Q_n(\sigma)$.
\end{proof}

In the following we show that even with postselection in the end, a quantum planar Yang-Baxter circuit still generates a sparse subset of unitary group. In other words any attempt to prove post-selected universality for the particle scattering model without intermediate measurements will probably fail.

\begin{theorem}
The set of unitary operators generated by $\HQBALL$ with postselection in particle label basis in the end of computation, correspond to the union of (discrete) $n!^{O(1)}$ manifolds, each with linear dimension.
\label{YBnonu1}
\end{theorem}

\begin{proof}
We follow the proof of theorem ~\ref{YBnonu}. Consider the planar YB circuits on $n$ labels. If the input velocities are fixed, then the unitary operators generated by the model constitute a finite set of size at most $n!$. There are finite $n!^{O(1)}$ to do a postselection on the output labels of each circuit. So for each fixed set of velocities, the unitary matrices obtained by postselection and proper normalization still constitute a set of size $n!^{O(1)}$. Therefore, label the manifolds with the permutation signature of the circuits and the type of final postselection. Then the points in each of these manifolds are uniquely specified by $n$ velocity parameters, which is an upper-bound on the dimension for each of them.
\end{proof}

\noindent Notice that result of these theorems still hold if we allow the circuit models to start with arbitrary initial states.

\section{A modified quantum ball-permuting model}
\label{zqball}

In the definition of $X (\cdot, \cdot)$ operators, the angle $\theta$ is independent of the labels that are being swapped. This in turn leads to the following observation suggesting that the output of this model is relatively restricted:

\begin{theorem}
Let $U=X(\theta_m, k_m) \ldots X(\theta_2, k_2) X(\theta_1, k_1)$ be any composition of the $X$ operators, then columns of $U$ as a matrix in $S_n$ basis, are obtainable by permuting the entries of the top-most column.
\label{ind}
\end{theorem}

\begin{proof}
Consider the first column of $U$ spanned by:

$$
U|123\ldots n\rangle =\sum_{\sigma\in S_n} \alpha_\sigma |\sigma\rangle
$$

Where $\alpha_\sigma$'s are the amplitudes of the superposition. Now consider any other column marked by $\pi$:

$$
U|\pi\rangle=\sum_{\sigma\in S_n} \beta_\sigma |\sigma\rangle
$$

Clearly, $|\pi\rangle=R(\pi)|123\ldots n\rangle$, and since $[U,R(\pi)]=0$:

$$
\sum_{\sigma\in S_n} \beta_\sigma |\sigma\rangle= \sum_{\sigma\in S_n} \alpha_\sigma |\sigma\circ \pi\rangle,
$$

\noindent which is the desired permutation of columns, and in terms of entries $\beta_{\sigma\circ \pi}=\alpha_\sigma$.

\end{proof}

Here we introduce another local unitary, $Z(\tilde{\theta},k)$, wherein the transposition angles depend on the color of the labels. Here $\tilde {\theta} = \{\theta_{ij}\}$ is a list of angles, one element per each $i\neq j\in [n]$. By definition $Z(\tilde{\theta},k)$ acts on the labels $|ab\rangle$ in the locations $k$ and $k+1$ with the following map:

$$
Z(\tilde{\theta},k)|ab\rangle =\cos \theta_{ab} I + i\sin \theta_{ab} L_{(a,a+1)}
$$

If we assume real valued angles with $\theta_{ij}=\theta_{ji}$, then the operator $Z$ becomes unitary. Clearly, the $X$ operators are the special case of the $Z$ operators. In order to see this, consider any basis $|\sigma\rangle, \sigma \in S_n$, and suppose $\sigma(k)=a, \sigma(k+1)=b$ then:

$$
Z^\dagger (\tilde{\theta},k)Z(\tilde{\theta},k)|\sigma\rangle=(\cos \theta_{ab}- i \sin \theta_{ab} L_{(k,k+1)})(\cos \theta_{ab}+ i \sin \theta_{ab} L_{(k,k+1)})|\sigma\rangle=|\sigma\rangle
$$

Next we use a simple encodings of qubits using labels $1,2,3\ldots, n$, and the $Z$ operators to operate on them as single and two qubit gates and prove that this modified model can simulate $\BQP$. More specifically, we prove that using a sequence of $Z$ operators, one can encode any element in the special orthogonal group. For an example of encoded universality see \cite{gottesman2001encoding, bacon2001encoded}. We encode each qubit using two labels. Given two labels $a < b$ we define the encoded (logical) qubits as:

$$
|0\rangle := | a b \rangle
$$

and,

$$
|1\rangle := i |b a\rangle.
$$

\noindent Using simple $X(\theta,1)$ we can apply arbitrary rotation of the following form:

$$
|0\rangle \rightarrow \cos \theta |0\rangle+\sin \theta |1\rangle
$$

and,

$$
|1\rangle \rightarrow \cos \theta |1\rangle-\sin \theta |0\rangle.
$$

\noindent We are dealing with orthogonal matrices which are represented over the field or real numbers. Using the $Z$ operators, we can discuss a controlled swap of the form:

$$
S(i,j,k,l) := Z(\pi/2 \delta_{i,j}, k , l).
$$

In simple words, $S(i,j,k,l)$ applies the swap $i L_(k,l)$, on the $k$ and $l$'th labels if and only if the content of these label locations are $i$ and $j$ ( $j$ and $i$). We can also extend it to the following form:

$$
S(\{(i_1,j_1)^{s_1}, (i_2,j_2)^{s_2}, \ldots, (i_t,j_t)^{s_t}\},k,l) := Z(\pi/2 \delta_{i,j}, k , l).
$$

\noindent Where $s_m$ can be a symbol $\star$ or nothing. Given $(i_m, j_m)^\star$ in the list means that the swap $(i L_{(k, l)})^\dagger= - i L_{(k, l)}$ is applied if the content of $k$ and $l$ are $i_m$ and $j_m$. And given plain $(i_m, j_m)$ in the list means $i L_{(k, l)}$ if the content of $k$ and $l$ are $i_m$ and $j_m$.

Suppose that one encodes one qubit with labels $a < b$ and another one with $x < y$, we wish to find a unitary operator which applies a \textit{controlled not} on the two qubits, that is the following map:

\begin{eqnarray*}
|00\rangle &:=& | a,b, x,y\rangle \rightarrow  | a,b, x,y\rangle=|00\rangle\\
|01\rangle &:=& i| a,b, y,x\rangle \rightarrow  i| a,b, y,x\rangle=|01\rangle\\
|10\rangle &:=& i| b,a, x,y\rangle \rightarrow -| b,a, y,x\rangle=|11\rangle\\
|11\rangle &:=& -| b,a, y,x\rangle \rightarrow i | b,a, x,y\rangle =|10\rangle
\end{eqnarray*}

\noindent It can be confirmed that the following operator can do this:

$$
\hspace{-1cm}
C:= S (\{(a , x), (a , y)^\star\}, 1, 3) S (\{(a , x), (a , y)\}, 2, 3) S (\{(a , x), (a , y)\}, 1, 2)
$$

\noindent Given these two operators, one can simulate special orthogonal two-level systems, that is for each orthonormal $|\psi\rangle$ and $|\phi\rangle$ in the computational basis of $n$ qubits we can apply an operator which acts as:

$$
|\psi\rangle \rightarrow \cos \theta |\psi\rangle + \sin \theta |\phi\rangle 
$$

and,

$$
|\phi\rangle \rightarrow \cos \theta |\phi\rangle - \sin \theta |\psi\rangle 
$$


\section{The One-Clean-Qubit Model}
\label{dqc1intro}
While state of a quantum system is a pure vector in a Hilbert space, most of the time the actual quantum state is unknown; instead, all we know is a classical probability distribution over different quantum states, \i.e., the given quantum state is either $|\psi_1\rangle$ with probability $p_1$, or $|\psi_2\rangle$ with probability $p_2$, and so on. In other words, the state is an ensemble of quantum states $\{(p_1, |\psi_1\rangle) ,(p_2, |\psi_2\rangle),\ldots, (p_n, |\psi_n\rangle) \}$, for $p_1 + p_2 + \ldots p_n =1$. Such an ensemble is a mixture of quantum probability and classical probability distributions at the same time, it is also called a mixed state, and is described by a density matrix $\rho$:

$$
\rho = \sum_{j \in [n]} p_j |\psi_j \rangle \langle \psi_j |.
$$

\noindent A density matrix is a Hermitian operator, with nonnegative eigenvalues and unit trace. A quantum state is called pure, if it has a density matrix of the form $|\psi\rangle \langle \psi |$. In other words, a quantum state $\rho$ is pure if and only if $tr(\rho^2)=1$.

If some quantum state is initially prepared in the mixed state $\rho_0$, then given a unitary evolution $U$, the state is mapped to $U\rho_0 U^\dagger$. Let $\{|j\rangle: j \in [n]\}$ be some orthonormal basis of a Hilbert space. The maximally mixed state of this Hilbert space has the form $\dfrac{1}{n} \sum_j  |j\rangle  \langle j|= \dfrac{I}{n}$, is a quantum state which contains zero quantum information in it. That is, the outcome of any measurement can be simulated by a uniform probability distribution on $n$ numbers. Also, a maximally mixed state is independent of the selection of the orthonormal basis. Quantum computing on a maximally mixed state is hopeless, since $\dfrac{I}{n}$ is stable under any unitary evolution.

Consider the situation where we can prepare a pure qubit along with $n$ maximally mixed qubits to get $|0\rangle \langle 0| \tensor \dfrac{I}{n}$. The state $|0\rangle \langle 0|$ is also referred to as a clean qubit. In this case, the quantum state has one bit of quantum information in it. It is also believed that there are problems in $\DQC1$ that are not contained in the polynomial time. One example of such problem, is the problem of deciding if the trace of a unitary matrix is large or small. No polynomial time algorithm is known for this problem. We are going to point out to the trace computing problem later in Section ~\ref{trace}. Moreover, if we consider the version of $\DQC1$ where we are allowed to measure more than one qubits, then it is shown that there is no efficient classical simulation in this case, unless the polynomial Hierarchy collapses to the third level.  In the version of my definition, since we used a polynomial Turing machine as a pre-processor, $\DQC1$ immediately contains $\P$. Pre-processing can be tricky for the one-clean-qubit model. For example, as it appears, if instead of $\P$, we used $\NC^1$, the class $\P$ and $\DQC1$ are incomparable.

\begin{definition}
(The one-clean-qubit model \cite{knill1998power}) $\DQC1$ is the class of decision problems that are efficiently solvable with bounded probability of error using a one-clean-qubit and arbitrary amount of maximally mixed qubits. More formally, it is the class of languages $L\subseteq \{0,1\}^\star$, for which there is a polynomial time Turing machine $M$, which on any input $x \in \{0,1\}^\star$, outputs the description of a unitary matrix $\langle U \rangle$ with the following property: if $x \in L$, the probability of measuring a $|0\rangle$ on the first qubit of $U |0\rangle \langle 0| \tensor \dfrac{I}{n} U^\dagger$ is $\geq 2/3$, and otherwise it is $\leq 1/3$. Here $U$ is a $2 n \times 2 n$ unitary matrix.
\end{definition}

Notice that if we allow intermediate measurements we will obtain the original $\BQP$; just measure all qubits in $\{|0\rangle, |1\rangle \}$ basis, and continue on a $\BQP$ computation. Clearly, $\DQC1$ is contained in $\BQP$; in order to see this, just use Hadamrds and intermediate measurements to prepare the maximally mixed state, and continue on a $\DQC1$ computation. It is unknown whether $\BQP \subseteq \DQC_1$, however, we believe that this should not be true.

\section{A Brief Review of the Representation Theory of the Symmetric Group}
\label{RepTheory}

Most of the mathematical review is borrowed from \cite{james1981representation}. We are interested in two mathematical structures, the group algebra of the symmetric group $\C S_n$, and the unitary regular representation of the symmetric group. As it turns out, the two structures are closely related to each other, and also to the group generated by the ball permuting gates. Group algebra is an extension of a group to an algebra, by viewing the members of the group as linearly independent basis of a vector space over the field $\C$. Therefore, in addition to the group action an action of $\C$ on $S_n$ is needed, by the map $(\alpha, \sigma) \mapsto \alpha S_n$, and also addition of vectors in the usual sense. Therefore, a group algebra consists of all elements that can ever be generated by vector on vector composition and linear combination of vectors over $\C$.  Any element of $\C S_n$ can be uniquely written as $\sum_{\sigma \in S_n} \alpha_\sigma \sigma$, with $\C$ coefficients $\alpha_\sigma$. If we add a conjugation convolution $^\dagger$ with maps $\sigma^\dagger=\sigma^{-1}$, and $\alpha^\dagger = \alpha^\star$, then for any element $v\in \C S_n$, $v^\dagger v=0$, if and only if, $v=0$. In order to see this, let $v=\sum_{\sigma \in S_n} \alpha_\sigma \sigma$. Then, $v^\dagger v= \sum_\sigma |\alpha_\sigma|^2 e+ \ldots=0$. A zero on the right hand side implies zeroth of all the vector components, including the component along $e$, which implies $\alpha_\sigma=0$ for all $\sigma \in S_n$, and therefore $v=0$. Let $e$ be the identity element of $S_n$, consider an element $p\in \C S_n$ to be a projector if it has the property $p^2=p$. Two projectors $p$ and $q$ are called orthogonal if $p. q =0$. Then $(e-p)^2=e-p$ is also a projector, and also $p (e-p) =0$ are orthogonal projectors. $0$ is trivially a projector. Therefore, the group algebra decomposes as:

$$
\C S_n = \C S_n e = \C S_n (e-p) + p = \C S_n (e-p) \oplus \C S_n p.
$$

\noindent A projector is called minimal if it cannot be written as the sum of any two others projectors other than $0$ and itself. Let $p^\mu$ be a list of minimal projectors summing $\sum_\mu p^\mu=e$, then the decomposition of the group algebra into minimal parts is according to:

$$
\C S_n = \bigoplus_\mu \C S_n p^\mu.
$$

\noindent $p^\mu$ are known as Young symmetrizers, and we are going to mention them later.

A (finite) representation $\rho$ of a group $G$ is a homomorphism from $G$ to the group of isomorphisms of a linear space $: G \rightarrow GL(V, \C)$, for some vector space $V$. Let $g$ be any element of $G$, with its inverse $g^{-1}$, and $e$ and $1$ as the identity elements of $G$ and $GL(V,\C)$, respectively. Given the definition, $\rho (g^{-1})= \rho(g)^{-1}$, and $\rho(e)=1$ are immediate. One can observe that $\rho : G \rightarrow \{I\in GL(V,\C)\}$, is immediately a representation, and is called the trivial representation of $V$. A dual representation of $G$ is a homomorphism from $G$ into the group of linear maps $: V \rightarrow \C$. As we discussed before, this is called the dual space $V^\star$, and $V$ is viewed as the space of column vectors, then its dual space is a row space. For any vector spaces $V$ and $W$, the two can be combined into a larger linear structure, $V\otimes W^\star$, as the set of linear maps from $W$ to $V$. Let $M_1$ and $M_2$ be two elements of $GL(V,\C)$ and $GL(W,\C)$, respectively. Then, viewing $V\otimes W^\star$ as a vector space, the object $(M_1, M_2)$ acts on $x \in V\otimes W^\star$ with $M_1 x M^{-1}_2$. Then, if $M_1$ and $M_2$ are two representations of $G$ on $V$ and $W$, then $(M_1, M_2)$ is a representation of $G$ on $V\otimes W^\star$, as a vector space. Notice that the inverse on $M_2$ is needed in order to have $(M_1, M_2)$ act as a homomorphism. The dual representation $M$ of $V$ is then the representation on $\C\otimes V^\star$, when $M_2=M$, and $M_1$ is the one dimensional trivial representation. This is just saying that the dual representation $M^\star$ of $M$ on $V^\star$, maps $\langle \psi |$ to $\langle \psi| M(g^{-1})$, if we view the dual space as the usual row space. If we define an inner product as the action of the dual of a vector on itself, then $G$, as a representation, sends orthonormal basis to orthonormal basis. This suggests that every representation of a finite group is isomorphic to a unitary representation. That is, any non-unitary representation becomes unitary after a change of basis. Let $M$ be a representation on $V$. Then, we say $W\subseteq V$ is called stable under $M$, if for any $x\in W$, $M x \in W$. Then, $M$ restricted to $W$ is called a sub-representation.  A representation $M$ on $V$ is called an irreducible representation (irrep), if it has no stable subspaces other than $0$ and $V$. Two representations $M_1$ and $M_2$ on $V_1$ and $V_2$ are isomorphic if $M_1$ resembles $M_2$ after a suitable change of basis within $V_1$. Then, if $V$ is reducible, it can be decomposed as $V_1 \oplus V_2 \oplus \ldots \oplus V_n$, for $n>1$. Some of the sub-representations can be isomorphic, and the multiplicity of a sub-representation is the number of sub-representations isomorphic to it. Then, the isomorphic subspaces can be grouped together to $V\cong m_1 V_1 \oplus m_2 V_2 \oplus \ldots m_k V_k$. Then $\dim V= \sum_j m_j \dim V_j$. The structure of such decomposition is isomorphic to $\bigoplus_j V_j \otimes X_j$, where $X_j$ is the multiplicity space of $V_j$ and is a vector space of dimension $m_j$. Decomposition of a representation onto the irreducible ones is unique up to isomorphism and multiplicities and dimensionality of irreducible representations do not depend on the decomposition. Canonical ways to find a decomposition are also known.

The regular representation of $S_n$, also denoted by $\C S_n$, is the unitary representation of $S_n$ onto the usual Hilbert space $\C S_n$ spanned by the orthonormal basis $\{|\sigma\rangle : \sigma \in S_n\}$. It is well known that for any regular representation, the dimension of each irrep is equal to the multiplicity of the irrep, and therefore $\C S_n$ decomposes into irreducible representations of the following form:

$$
\C S_n \cong \bigoplus_\lambda V_\lambda \otimes X_\lambda,
$$

\noindent with $\dim X_\lambda = \dim V_\lambda =: m_\lambda$, and indeed $\sum_\lambda m^2_\lambda= n!$. Here $X_\lambda$ is again the multiplicity space, and $V_\lambda$ corresponds to each irrep. It is tempting to make a connection between the group algebra and regular representation of the symmetric group. As described earlier, $S_n$ can act on the Hilbert space $\C S_n$ in two ways; the left and right, $L, R: S_n \rightarrow U(\C S_n)$, unitary regular representation, with the maps $L(\sigma) |\tau\rangle = |\sigma \circ \tau\rangle$ and $R(\sigma) |\tau\rangle = |\tau\circ \sigma^{-1}\rangle$. Also, similar left and right structure can be added to the group algebra. Clearly, $L$ and $R$ representations commute, and it can be shown that the algebra generated by $L$ is the entire commutant of the algebra generated by $R$. Putting everything together, inspired by the theory of decoherence free subspaces, and the defined structures, one can show that the left ($A$) and right ($B$) algebras and the Hilbert space $\C S_n$ decompose according to:

$$
A\cong \bigoplus_\lambda M(m_\lambda)\otimes I(m_\lambda),
$$

$$
B\cong \bigoplus_\lambda I(m_\lambda) \otimes M(m_\lambda),
$$

and,

$$
\C S\cong \bigoplus_\lambda V(m_\lambda) \otimes X(m_\lambda).
$$

\noindent This is indeed a nice and symmetric structure. Indeed each irrep $V_\lambda$ is an invariant subspace of the $X$ operators, and it cannot be reduced further. It remains to demonstrate the structure of the irreps $\lambda$, and to study the action of $X$ operators on these subspaces.

The irreducible representations of the symmetric group $S_n$ are marked by the partitions of $n$. Remember that a partition of $n$ is a sequence of non-ascending positive numbers $\lambda_1 \geq \lambda_2 \geq \lambda_3 \geq \ldots \lambda_k$ summing to $n$, \i.e., $\sum_j \lambda_j = n$. The number of partitions of $n$ grows like $\exp \Theta(\sqrt{n})$. Each as described earlier each partition $\lambda= (\lambda_1, \lambda_2, \ldots, \lambda_k)$ is related to a diagram, called the Young diagram, which consists of $k$ horizontal rows of square boxes $r_1, r_2, \ldots, r_k$. The Young diagram is then created by paving the left-most box of $r_1$ to the left-most box of $r_2$, and so on. For a Young diagram $\lambda$, the dual diagram $\tilde{\lambda}$, is another Young diagram, whose rows are the columns of $\lambda$. A Young tableau $t^\lambda$ with the shape $\lambda$, is a way of bijective assigning of the numbers in $[n]$ to the boxes of $\lambda$. We will use $t^\lambda$ and simply $t$ with the shape $\lambda$ interchangeably. A permutation $\pi \in S_n$ can act on a Young tableau $t^\lambda$ by just replacing the content of each box to the its image under $\pi$, \i.e., if a box contains $j$, after the action of $\pi$ it will be replaced with $\pi(j)$. A tableau is called standard, if the numbers in each row and column are all in ascending orders. The number of standard tableau for each partition of shape $\lambda$ is denoted by $f^\lambda$.

Let $t$ be a tableau with shape $\lambda$. Define $P(t)$ and $Q(t)\subseteq S_n$ to be sets of permutations that leave each row and column invariant, respectively. Then the projectors of the $\C S_n$ group algebra are according to the Young symmetrizers, one for each standard tableau:

$$
p^t = \frac{1}{f^\lambda} \sum_{\pi \in C(t)}\sum_{\sigma \in R(t)} sgn(\pi) \pi \circ \sigma.
$$

These subspaces correspond to all of the irreducible invariant subspaces of $S_n$. The dimension for each of these subspaces is the number of standard tableaus of each partition, and it is computable using the hook lengths. The hook of each box in a partition of shape $\lambda$ is consists of the box itself along with all boxes below and at the right of the box. The hook length of each box is the number of boxes contained in that hook, and the hook length $h^\lambda$ of the shape $\lambda$ is the multiplication of these numbers for each box. Then, the dimension of the irrep corresponding to $\lambda$ is according to $f^\lambda=n!/ h^\lambda$. 

\section{The Young-Yamanouchi Basis}
\label{YoungYamanouchi}

In order to talk about quantum operations orthonormal basis for the discussed subspaces are needed. It would be nice if we have a lucid description of the basis, in a way that the action of $X$ operators on these subspaces is clear. Moreover, we seek for an inductive structure for the orthonormal basis of the irreps that is adapted to the nested subgroups $S_1\subset S_2 \subset \ldots \subset S_n$. By that we mean states that are marked with quantum numbers like $|j_1, j_2, j_3, \ldots, j_k\rangle$, such that while elements of $S_n$ affect all the quantum numbers, for any $m_1< n$, elements of $S_n$ restricted to the first $m_1$ labels affects the first $k_1$ quantum numbers only, and act trivially on the rest of the labels. Also, for any $m_2< m_1<n$, the elements of $S_n$ restricted to the first $m_2$ labels affect the first $j_2< j_1 < k$ quantum numbers only, and so on.

Fortunately, such a bases exist, and are known as the subgroup adapted Young-Yamanouchi (YY) bases \cite{james1981representation}. These bases are both intuitive and easy to describe: for any partition of shape $\lambda$, mark an orthonormal basis with the standard Young Tableaus of shape. Agree on a lexicographic ordering of the standard tableaus, and denote these basis corresponding to the partition $\lambda$, by a $\{|\lambda_j\rangle\}_{j=1}^{f^\lambda}$. Denote the action of a swap $(i, j)$ on $|\lambda_l\rangle$ by $|(i,j). \lambda_l\rangle$, to be the basis of a tableau that is resulted by exchanging location of $i$ and $j$ in the boxes. Suppose that for such tableau $t$, the number $j$ ($i$) is located at the $r_j$ and $c_j$ ($r_i$ and $c_i$) row and column of $t$, respectively. Then, define the axial distance $d_{ij}$ of the label $i$ from label $j$ of on each tableau to be $(c_j-c_i)-(r_j-r_i)$. Or in other words, starting with the box containing $i$ walk on the boxes to get to the box $j$. Whenever step up or right is taken add a $-1$, and whenever for a step down or left add a $1$. Starting with the number $0$, the resulting number in the end of the walk is the desired distance. Given this background, the action of $L_{(k,k+1)}$ on the state $|\lambda_i\rangle$, is according to:

$$
L_{(k,k+1)} |\lambda_i\rangle = \dfrac{1}{d_{k+1, k}} |\lambda_i\rangle + \sqrt{1-\dfrac{1}{d^2_{k+1, k}}}|(k,k+1).\lambda_i\rangle
$$

\noindent Three situations can occur: either $k$ and $k+1$ are in the same column or row, or they are not. If they are in the same row, since the tableau is standard, $k$ must come before $k+1$, then the axial distance is $d_{k+1, k}=1$, and the action of $L_{(k,k+1)}$ is merely:

$$
L_{(k,k+1)} |\lambda_i\rangle = |\lambda_i\rangle.
$$

\noindent If the numbers are not in the same column, $k$ must appear right at the top of $k+1$, and the action is:

$$
L_{(k,k+1)} |\lambda_i\rangle = -|\lambda_i\rangle.
$$

\noindent Finally, if neither of these happen, and the two labels are not in the same row or column, then the tableau is placed in the superposition of itself, and the tableau wherein $k$ and $k+1$ are exchanged. Notice that if the tableau $|\lambda_i\rangle$ is standard the exchanged tableau $|\lambda_i\rangle$ is also standard. This can be verified by checking the columns and rows containing $k$ and $k+1$. For example, in the row containing $k$, all the numbers at the left of $k$ are less than $k$, then if we replace $k$ with $k+1$, again all the numbers on the left of $k+1$ are still less than $k+1$. Similar tests for the different parts in the two rows and columns will verify $(k,k+1)\lambda_i$, as a standard tableau. The action of $L_{k,k+1}$ in this case is also an involution. This is obvious for the two cases where $k$ and $k+1$ are in the same row or column. Also, in the third case if the action of $L_{(k,k+1)}$ maps $|\lambda\rangle$ to $\dfrac{1}{d} |\lambda\rangle + \sqrt{1-\dfrac{1}{d^2}}|t\circ \lambda\rangle$ then a second action maps $|t\circ \lambda\rangle$ to $\dfrac{-1}{d} |t\circ \lambda\rangle + \sqrt{1-\dfrac{1}{d^2}}|\lambda\rangle$, and therefore:

$$
L^2_{(k,k+1)}|\lambda\rangle = \dfrac{1}{d}(\dfrac{1}{d} |\lambda\rangle + \sqrt{1-\dfrac{1}{d^2}}|t\circ \lambda\rangle)+\sqrt{1-\dfrac{1}{d^2}}(\dfrac{-1}{d} |t\circ \lambda\rangle + \sqrt{1-\dfrac{1}{d^2}}|\lambda\rangle)=|\lambda\rangle.
$$

Given this description of the invariant subspaces, we wish to provide a partial classification of the image of the ball permuting gates on each of these irreps. The hope is to find denseness in $\prod_\lambda SU(V_\lambda)$, on each of the irreps $V_\lambda$, with an independent action on each block. In this setting, two blocks $\lambda$ and $\mu$ are called dependent, if the action on $\lambda$ is a function of the action on $\mu$, \i.e., the action on the joint block $V_\lambda\oplus V_\mu$ resembles $U\times f(U)$, for some function $f$. Then, independence is translated to decoupled actions like $I \times U$ and $U \times I$.

Throughout, the $\lambda \vdash n$, means that $\lambda$ is a partition of $n$. We say $\mu\vdash n+1$ is constructible by $\lambda \vdash n$, if there is a way of adding a box to $\lambda$ to get $\mu$. We say a partition $\mu \vdash m$ is contained in $\lambda \vdash n$, for $m<n$, if there is a sequence of partitions $\mu_1 \vdash m+1$, $\mu_2 \vdash m+2, \ldots, \mu_{n-m-1} \vdash n-1$, such that $\mu_1$ is constructible by $\mu$, $\lambda$ is constructible by $\mu_{n-m-1}$, and finally for each $j \in [n-m-2]$, $\mu_{j+1}$ is constructible by $\mu_j$. We also call $\mu$ a sub-partition of $\lambda$. A box in a partition $\lambda$ is called removable, if by removing the box the resulting structure is still a partition. Also, define a box to be addable if by adding the box the resulting structure is a partition.

\begin{theorem}
The Young-Yamanouchi bases for partitions of $n$ are adapted to the chain of subgroups $\{e\}=S_1 \subset S_2 \subset \ldots \subset S_n$. 
\end{theorem}

\begin{proof}
Let $\lambda\vdash n$, and $t$ be any standard tableau of shape $\lambda$.  We construct some enumeration of states in the Young-Yamanouchi basis of $\lambda$ which is adapted to the action of subgroups. For any $m<n$, since $t$ is a standard tableau, the numbers $1,2,3,\ldots, m$, are all contained in a sub-partition $\mu \vdash m$ of $\lambda$. This must be true, since otherwise the locus of numbers $1, 2, 3, \ldots, m$ do not shape as a sub-partition of $\lambda$. Let $\nu$ be the smallest sub-partition of $n$ that contains these numbers. Clearly, $|\nu|>m$. The pigeonhole principle implies that, there is a number $k>m$ contained somewhere in $\nu$. The box containing $k$ is not removable from $\nu$, since otherwise you can just remove it to obtain a sub-partition smaller than $\nu$ that contains all of the numbers in $[m]$. Therefore, if $k$ is in the bulk of $\nu$, then both the row and column containing $k$ are not in the standard order. If $k$ is on a vertical (horizontal) boundary, then the column (row) of the box containing $k$ is not standard.

Let $\lambda_k$ be the smallest sub-partition of $\lambda$ that contains $[k]$. Then the enumeration of the basis is according to $|\lambda_1, \lambda_2, \ldots, \lambda_n \rangle$. Here, $\lambda_n=\lambda$, and $\lambda_1$ is a single box. From before, for any $j<n$, $\lambda_{j+1}$ is constructible by $\lambda_j$. For $m<n$, let $S_m$ be the subgroup of $S_n$, that stabilizes the numbers $m+1, m+2, \ldots, n$. For any $k\leq m$, $L_{(k,k+1)}$ just exchanges the content of boxes withing $\lambda_m$, and therefore leaves the quantum numbers $\lambda_{m+1}, \lambda_{m+2}, \ldots, \lambda_n$ invariant. Moreover, the box containing $m$ is somewhere among the removable boxes of $\lambda_m$, since otherwise, as described in the last paragraph, the tableau $\lambda_m$ is not standard. The box containing $m-1$ is either right above or on the left side of $m$, or it is also a removable box. In the first two cases, the action of $L_{(m-1,m)}$ is diagonal, and the quantum numbers are intact. In the third case, the only quantum numbers that are changed are $\lambda_{m-1}$ and $\lambda_m$.
\end{proof}

Consider now the action of $S_{n-1}$ on an element $|\lambda_1, \lambda_2, \ldots , \lambda_n=\lambda\rangle$. In any case $\lambda$ is constructible by $\lambda_{n-1}$, and the construction is by adding an addable box to $\lambda_{n-1}$. In other words, $\lambda_{n-1}$ can be any partition $\vdash n-1$, that is obtained by removing a removable box from $\lambda$. These observations, all together, lead to a neat tool:

\begin{lemma}
(Branching.) Under the action of $S_{n-1}$, $V_\lambda \cong \bigoplus_{\substack{\mu\vdash n-1 \\ \mu\subset \lambda}}V_{\mu}$.
\end{lemma}

\begin{proof}
The proof is directly based on the structure of the YY bases. What we would like to emphasize here is that the multiplicity free branching rule of the symmetric group is manifest in the structure of the YY bases. For other formal proofs see \cite{james1981representation}.

Choose an orthonormal basis according to YY. Enumerate the removable boxes of $\lambda$ by $1,2,\ldots, p$. Clearly, in any standard tableau of $\lambda$, the box containing $n$ is a removable one. Group the tableaus according to the location of $n$. Clearly, each subspace corresponds to a partition $\mu \vdash n-1 \subset \lambda$. Call these partitions $\mu_1, \mu_2, \ldots, \mu_p$, according to the enumeration of removable boxes. Also denote the space $V_{\mu_j}$ correspondingly. For any $\mu_j$, any element of $S_{n-1}$, acted on $V_{\mu_j}$, generates a vector within $V_{\mu_j}$. In other words, these subspaces are stable under $S_{n-1}$.
\end{proof}

\section{Detailed proofs for Section \ref{classical}}
\label{ballpermutingoracles}

\begin{theorem}
$\DBall=\DBall_{adj}=\L=\Rev\L$ 
\end{theorem}

\begin{proof}
We first prove the direction $\DBALL=\DBALL_{adj}$. Clearly, $\DBALL_{adj}\subseteq \DBALL$ as a special case. Any nonadjacent swap $(i,j)$ for $i<j$ can be obtained by a sequence of adjacent transpositions $(i,j)= (i) \circ (i+1) \circ \ldots \circ (j-2) \circ (j-1) \circ (j-2) \circ \ldots \circ (i+1) \circ (i)$. Therefore, any sequence of $m$ nonadjacent swaps on $n$ labels can be simulated by $O(m.n)$ number of adjacent transpositions. Thereby $\DBALL \subseteq \DBALL_{adj}$.

In order to prove the direction $\Rev \L \subseteq \DBALL$ we observe from before that the evolution of a reversible computation is according to a configuration space wherein all configuration nodes have in-degree and out-degree at most $1$, and thereby the map which evolves the current configuration of a Turing machine to the next one is a bijection between configurations. A $\LOGSPACE$ (reversible) Turing Machine that uses $c. \log(n)$ space has a configuration space of size $c. n \log (n) n^c=n^{O(1)}$. Notice that such Turing Machine runs for at most polynomial amount of time before looping. Now given any $\LOGSPACE$ machine, consider a $\DBall$ oracle of size $N=n^{O(1)}$. Also without loss of generality we can assume that the Turing Machine on any input runs in a fixed time $T=n^{O(1)}$ for all of its inputs.  Given the description of the Turing machine and its input, we can encode the description of each configuration of the Turing machine with numbers $1, 2, \ldots, N$. Now each step of the computation corresponds to a permutation $:[N]\rightarrow [N]$, and each permutation can be decomposed into $N^{O(1)}$ pairwise permutations (swaps). Therefore an $AC^0$ machine encodes the evolution of the Turing Machine as sequence of swaps. And the evolution of the machine for $T$ steps as the list of swaps repeated for $T$ times. Let $1$ be the encoded initial state. The oracle then applies these swaps in order and in the end we look at the location of the symbol $1$ in the final permutation. If the final location corresponds to an accepting state the $\DBALL$ computer accepts, and otherwise rejects. Now from theorem \ref{L=RevL} we observe that $\Rev \L=\L$ and thereby the direction $\Rev \L=\L\subseteq \DBALL$ is derived.

In order to prove the direction $\DBALL\subseteq \L$, consider any $\DBall$ oracle queried with a list of swaps as the input. We can devise a $\LOGSPACE$ computation in which the final permutation is computed. Notice that the logarithmic space is unable to store the full description of the final permutation. Therefore, instead we design the computation in a way that the location of each symbol $s$ in the final permutation appears on the read/write tape. For this purpose, the machine just keeps track of the current location of $s$ on the read/write tape. If at some step a swap $(i,j)$ is queried, the machine updates the tape if and only if the location of $s$ is either $i$ or $j$. Consider this program as a subroutine. Given the list of swaps and a target permutation as the input, a $\LOGSPACE$ machine implements this subroutine for each symbol and compares the location of the symbol in the final permutation with the input and rejects if they do not match. Then the whole computation accepts if all the tests succeed.
\end {proof}

\begin{theorem}
$\L \subseteq \BPL \subseteq \RBALL=\RBALL_{adj}\subseteq \Almost \L\subseteq \BPP$. However, if we let $\RBALL (2)$ to the class where the $\AC$ machine is allowed to make two adaptive queries, then $\RBALL(2)=\Almost\L$.
\end{theorem}

\begin{proof}
In order to observe $\RBALL=\RBALL_{adj}$, we use the fact that any nonadjacent swap can be produced as application of polynomially many transpositions, thereby $\RBALL_{adj}$ simulates $\RBALL$ by simulating each swap with a sequence of adjacent swaps. More precisely, suppose that $\RBALL$ queries the swap $(i,j)$ with probability $p$. Then $\RBALL_{adj}$ computer first queries $(i) , (i+1)  ,\ldots , (j-2)$ each with probability $1$, then queries $(j-1)$ with probability $p$, and finally queries $ (j-2), \ldots, (i+1), (i)$ each with probability $1$. This proof works independent of the number of queries.

Now suppose that single queries are allowed. As the simulation of $\DBALL$ in $\L$ requires repeated use of a subroutine for each label over and over, any machine to simulate $\RBALL$ with this scheme needs consistent access to the random bits for each subroutine. Given this, we observe that the machines of class $\Almost\L$ provide such consistent access to these random bits. Therefore, an $\Almost\L$ machine runs $\RBALL$ in the $\L$ simulation for each label, and whenever a probabilistic swap is queried, the machine uses the random oracle to decide whether to make it or not. Notice that for this simulation a random oracle is required rather than an ephemeral stream of random bits. That is because the simulation needs to use the same random bits over and over. Suppose that the $\Almost\L$ machine requires $N=n^{O(1)}$ random bits in each $\L$ simulation. Then it picks a lexicographic convention (hardwired to the machine's transition function) on finite strings, and in order to obtain the $j \leq N$'th random bit $b_j$, it queries the $j$'th lexicographic string to the oracle and lets $b_j=1$ if the oracle accepts, and $0$ otherwise.

In order to see the direction $\BPL \subseteq \RBALL$, we amalgamate the configuration space with a single tape cell which is a random bit. Thereby we double the configuration space, by adding $C0$ and $C1$ for each configuration $C$ of the original $\L$ machine. Also we map all configurations of the form $C0$ to odd numbers $1, 3, 5, \ldots$ and all of those with the form $C1$ to even numbers $2, 4, 6,\ldots$. Thereby we just run $\L$ in the $\DBALL$ simulation and whenever the machine needs a random bit on the random bit cell the $\AC^0$ machine queries the swaps $(1,2), (3,4), \ldots$ each with probability $1/2$ to the $\RBall$ oracle.

Now suppose that $\RBALL$ computer is allowed to adaptively query the Ball-Permuting oracle twice. First of all, in this case also $\RBALL\subseteq \Almost\L$. In order to see this suppose that the simulation of the first query to $\RBall$ needs $N_1$ random bits and the second one requires $N_2$ bits. Then the $\Almost \L$ machine first queries the first $N_1$ strings to its random oracle, and when it is done with the first round of $\RBall$ queries, it uses the result to simulate the $AC^0$ machine to design the second query to the oracle. Notice that the $\log n$ read/write tape might not be sufficient to store the whole content of the second query, and thereby it suffices for the machine to just store one bit of the second query at a time and repeat the whole computation to obtain the next bit.

In order to see the $\Almost\L\subseteq \RBALL$ direction, suppose that an $\Almost\L$ machine queries $N$ distinct strings to its oracle. Then an $\RBALL$ computer first queries $N$ transpositions $(1,2), (3,4),\ldots,(2N-1,2N)$ each with probability $1/2$, and uses the result to design the second query to simulate the running of $\Almost\L$ on the queried strings.
\end{proof}

\begin {theorem}
$\Ball_{adj}$ and $\Ball$ are polynomial-time reducible to each other, and furthermore they are complete for the class $\NP$. 
\end{theorem}

\begin {proof}
We introduce a polynomial time reduction from the word problem of permutations $\WPPP$ which is known to be $\NP$-complete. Any instance of $\WPPP$ is given by an ordered list of subsets of $[n]$, thereby for each subset $S =\{i_1, i_2, \ldots, i\}$ we add $ O(k^3)$ swaps. We need to choose the list in a way that for each permutation on $S$, there is a nondeterministic choice of swaps which produces the permutation. Each permutation on $k$ elements can be produced by $O(k)$ swaps. Therefore we list all the swaps on the elements of $S$ in some list $L$ and repeat the list for $k$ times. 
\end{proof}

\begin{theorem}
$\BALL^\star_{adj}\subseteq\P$
\label{nballstar}
\end{theorem}

\begin{proof}
This can be done by a reduction to the problem of edge disjoint path for directed planar graphs, which is contained in linear time. Given the list of $m$ swaps, probabilities, and the target permutation $\sigma \in S_n$ as the input of the $\Ball$ problem, we construct a directed planar graph with $m$ nodes, $n$ source edges (nodes) and $n$ sink edges (nodes), according to the following: Initially add $n$ source nodes with $n$ outgoing edges, one for each. Number these edges and nodes with $1,2,3, \ldots, n$. For each transposition $(i)$, merge the edge $i$ and $i+1$ in a vertex with two outgoing edges, update the numbering of the edges accordingly, \i.e. name one of the edges to be $i$ and another $i+1$. Continue this for all transpositions. At the end, add $n$ sink nodes each taking one of the edges as an input edge, and number them according to the target permutation. Clearly, the resulting graph is planar. This can be seen by induction on the steps of the construction algorithm. We need to prove that there is a nondeterministic generation of the target permutation in the $\Ball$ computation if and only if there is an edge disjoint path between the source nodes and sink nodes. Suppose that there is an edge disjoint path in the graph, we construct the list of swaps that create the target permutation. Sort the vertices in an ascending order by the distance from the source nodes. Each vertex is mapped to a transposition in the list with the ascending order, \i.e., if some vertex inputs the edges $i$ and $i+1$, then the corresponding transposition is $i$. Two paths are incident to each vertex. Suppose that the input edge of the first path takes the label $i$ and the second one $i+1$. If the output edge of the first path is $i+1$, we include the swap $(i)$ in the instruction list, and otherwise we don't. Each path thereby maps the initial label of its corresponding source node to the desired permuted label in the sink node. Now if there is a nondeterministic computation of the target permutation, then there is a way of choosing the transpositions to construct the target permutation. We then construct the resulting planar graph according to the algorithm, and the edge disjoint paths by the following: for each source node pick the first outgoing edge. For each edge in this path then its endpoint corresponds to a vertex which corresponds to a swap in the list. If the edge is labeled by $i$ choose the edge $i+1$ as the next edge if the corresponding swap is active, and otherwise pick the edge $i$.
\end{proof}

\begin{corollary}
The edge disjoint path problem in the non-planar case is $\NP$-complete.
\end{corollary} 

\begin{proof}
There is a reduction from $\BALL$ to the edge disjoint path. The reduction is similar to the one given in the proof of lemma ~\ref{nballstar}. The graph instance of $\EDP$ is non-planar if and only if the swaps of $\BALL$ are adjacent. In short we have:

$$
\VDP \leq_\P \WPPP \leq_\P \BALL \leq_\P  \EDP
$$

\noindent That means that a polynomial time algorithm for the general $\EDP$ implies $\P=\NP$.
\end{proof}

\section{Detailed proofs for Section \ref{sep}}
\label{trace}

\begin{theorem}
There is an efficient $\DQC 1$ algorithm which takes the description of a $\Poly(n)$ size ball permuting circuit $C$ over $\C S_n$ as its input, and outputs a complex number $\alpha$ such that \[|\alpha - \langle 123 \ldots n | C | 123 \ldots n\rangle| \leq \dfrac{1}{\Poly(n)}\], with high probability.
\label{mainDQC1}
\end{theorem}

The theorem is proved in three steps. First, in lemma \ref{lem1} we observe that for ball permuting circuits the computation of single amplitudes can be reduced to the computation of (normalized) traces. Next, we borrow a result of \cite{shor2008estimating} which provides a reduction from additive approximation of traces for unitary matrices to $\DQC 1$ computations. Finally, in the third step, by some careful analysis it is shown that the $\DQC 1$ reduction of the second step is an efficient one. The main idea for this step is to use a compressed encoding permutations with binary bits.

The amplitudes in ball permuting circuits are related to traces according to:

\begin{lemma}
For any ball permuting quantum circuit $C$, the trace $Tr(C)= n! \langle 123\ldots n | C | 123\ldots n \rangle$.
\label{lem1}
\end{lemma}

\begin{proof}
A quantum ball permuting circuit, by definition, consists of left permuting actions only which commute with right actions $R(\sigma)$ (relabeling) for any $\sigma \in S_n$. Thereby, $\langle 123\ldots n | C | 123\ldots n \rangle$ $ =$

\noindent $ \langle 123\ldots n | R^{-1}(\sigma) C R(\sigma)| 123\ldots n \rangle$$= \langle \sigma | C | \sigma \rangle$. From this, $Tr(C)= \sum_{\sigma \in S_n} \langle \sigma | C | \sigma \rangle = n! \langle 123\ldots n | C | 123\ldots n \rangle$.
\end{proof}

Next, we formally mention the problem of trace approximation:

\begin{definition}
($\Trace$) given as input the $\Poly(n)$ size description of a unitary circuit $U$ as a composition of gates from a universal gate set over $n$ qubits, compute a complex number $t$ such that $|t-\dfrac{1}{2^n} Tr(U)| \leq \dfrac{1}{\Poly(n)}$, with high probability.
\end{definition}

The following theorem provides an efficient $\DQC 1$ algorithm for $\Trace$:

\begin{theorem}
(Jordan-Shor \cite{shor2008estimating}) $\Trace$ is a complete problem for $\DQC1$. \footnote{Moreover, the authors show that $\Trace$ is a complete problem for this class, with polynomial time pre-processing.}
\label{JordanShor}
\end{theorem}

Indeed, this theorem can be reformulated as: given an $n$ qubit unitary $U$, there is a round of $\DQC 1$ computation which reveals a coin which gives heads with probability $\dfrac{1}{2} + \dfrac{1}{2}\dfrac{\Re Tr (U)}{2^n}$. Also, there is another similar computation which gives a coin with bias according to the imaginary part of the normalized trace.

Using these observations, we are ready to present the proof of the main theorem:

\begin{proof}
(of theorem \ref {mainDQC1}) The objective is find an efficient algorithm which given a ball permuting circuit $C$ over $n$ labels, outputs the description of a unitary $U$ over $m=\Poly(n)$ qubits such that $\dfrac{1}{2^m}Tr(U) =\dfrac{1}{\Delta(n)} \langle 123\ldots n | C | 123 \ldots n \rangle$, with $\Delta(n)=\Poly(n)$. Given this reduction using theorem \ref{JordanShor} we deduce that the additive approximation of the amplitude can be obtained by rounds of $\DQC 1$ computation. 

The basic idea is to encode permutations with strings of bits, perform the circuit $C$ on this encoded space, and take the trace of $C$ using a $\DQC1$ circuit. For ease of presentation, we will first present the proof using a simple encoding of permutations which turns out not to work, and later describe the more complex encoding which suffices for the proof.

Suppose we represent a permutation $\sigma \in S_n$ using $n\lceil\log n \rceil$ bits, i.e. each particle label in $[n]$ is represented using $\lceil \log n \rceil$ bits. 
Simulate each $X$ gate in $C$ with a quantum circuit which swaps the encoded numbers in a superposition. Since each gate only acts on $O(\log n)$ qubits, such a quantum circuit can be efficiently obtained from a universal gate set by the Solovay-Kitaev Theorem \cite{DawsonSolovayKitaev}. Let $U$ be the composition of these unitary circuits. The objective is to perform a $\DQC 1$ computation to obtain an approximation to $Tr(U)/D$, where $D$ is the dimension of the Hilbert space that $U$ is acting on. However one can easily see that $\operatorname*{Tr}(U)$ is not in general equal to $\operatorname*{Tr}(C)$ - because among the summands of $\operatorname*{Tr}(U)$ there are terms like $\langle b | U | b\rangle$, where $b$ is a string of bits with repeated labels (for example $|1 1 2 3 4\rangle$). Such terms do not appear in $\operatorname*{Tr}(C)$ because they are not valid encodings of permutations. In order to avoid the contribution of these terms, use $\dfrac{n(n-1)}{2}$ more (flag register) qubits, $f_{ij}, i< j \in [n]$. Then we add another term $T$ to the quantum circuit to obtain $U T$. The role of $T$ is simply to modify the flag registers in a way that the contribution of unwanted terms in the trace becomes zero: for each $i<j \in [n]$, using sequences of $CNOT$ gates, $T$ compares the qubits $(i-1) \lceil \log n\rceil +1$ to $i.\lceil \log n\rceil$ with the qubits $(i-1) \lceil \log n\rceil +1$ to $i.\lceil \log n\rceil$, bit by bit, and applies $NOT$ to the register $f_{i,j}$ if the corresponding bits are all equal to each other. Then $U T$ is fed into the $\Trace$ computation. Let's see what approximation to $ \langle 123\ldots n | C| 123 \ldots n \rangle$ we get in this case. Let $N:=n\lceil \log n\rceil+n(n-1)/2$. The trace $Tr(U)=\sum_{x\in \{0,1\}^N} \langle x|U|x \rangle$. Given the described construction, the term $\langle x | U | x\rangle= \langle \sigma | C |\sigma \rangle$, if and only if the label part of $x$ is the correct encoding of the permutation $\sigma$, and if $x$ is not a correct encoding of a permutation it gives $0$. There are $2^{n(n-1)/2}$ strings like $x$ which encode $\sigma$ correctly, therefore:

$$
\dfrac{Tr(U)}{2^N}= \dfrac{2^{n(n-1)/2}}{2^N} Tr (C)= \dfrac{n!}{2^{n \lceil \log n \rceil}} \langle 123\ldots n | C | 123 \ldots n \rangle.
$$

Using $\DQC1$ computations we can estimate the value of $Tr(U)/2^N$ to 1/poly additive error. This is almost what we want, but the problem is that the coefficient $\dfrac{n!}{2^{n \lceil \log n \rceil}}$ can be exponentially small, because $\log(n!)\approx n\log n- \Theta(n)$ by Stirling's approximation. Therefore the amplitude we are trying to compute ($\langle 123\ldots n | C | 123 \ldots n \rangle$) is exponentially suppressed in this model, so a 1/poly approximation to $1/2^N Tr(U)$ does not yield a 1/poly approximation to $\langle 123\ldots n | C | 123 \ldots n \rangle$. 

Taking a close look at the this coefficient, one can see that for any encoding of permutations with bit-strings, the proportionality constant appears as:

$$
\dfrac{n!}{dim{V}}
$$

\noindent where $V$ is the dimension of the Hilbert space that is used to encode permutations in it. In the latter example, we used $O (n \log n )$ bits to encode permutations of $n$ labels. Our problem arose because $2^{n\log n}$ is exponentially larger than $n!\approx (n/e)^n$. 

To fix this issue, we will need to use a more compressed encoding. More precisely, we need an encoding that uses $O(\log (n! \Poly(n))$ bits.
Moreover, in order to provide efficient quantum circuits, the code needs to be local, in the sense that in order to apply a swap, we just need to alter only $O(\log n)$ bits. Otherwise, it is not clear if it is possible to implement the encoded swaps efficiently with qubit quantum circuits.

To do this, we consider consider encoding permutations using $\lceil \log n! \rceil$ bits. Specifically, we consider an ordering of permutations (called the factorial number system, reviewed in Section \ref{factorial}) and represent each permutation using $\lceil \log n! \rceil$ bits. This encoding is extremely efficient, but it is not local, because one may need to rewrite all of the bits to perform an encoded swap. To overcome this issue, to perform an encoded swap on indices $i$ and $i+1$, we first extract $O(\log n)$ bits of information which encodes the values of the permutation at those locations.  We then apply the partial swap on the extracted entries (which is now manifestly local) and then convert the inefficient codes back to the compressed ones. In this manner we can approximate $\langle 123\ldots n | C | 123 \ldots n \rangle$ to 1/poly accuracy in the manner described above.

More precisely, we encode permutations using the factorial number system, described in the next Subsection \ref{factorial}. The basic idea is that once one has specified the first $k$ entries of the permutation, there are only $n-k$ choices for the next entry. Therefore, one can specify the next entry of the permutation by indicating which of these remaining $n-k$ elements to choose. In particular, we can represent a permutation $\sigma$ by a series of numbers $a_n,a_{n-1},\ldots a_2$, where each $a_i$ is a number from 0 to $i-1$ indicating which of the remaining items appears next in the permutation. (Note $a_1$ need not be included, since once you have specified the first $n-1$ entries of the permutation, you need not specify the last entry.)
The permutation is then represented by the number $N_{\sigma}=\sum_i a_i (i-1)!$, which ranges from $0$ to $n!-1$. 

This representation of a permutation is manifestly local - in order to swap two entries $i$ and $i+1$, one merely needs to perform some operation on $a_i$ and $a_{i+1}$. (This operation is slightly more complicated than just switching $a_i$ and $a_{i+1}$ due to an edge case where $a_i$ and $a_{i+1}$ have adjacent labels amongst the remaining labels, but as explained in the next subsection this is still local). Furthermore, one can easily extract $a_i$ in polynomial time from the number $N_{\sigma}$. To see this, if we let $r_i$ be the remainder of $N_\sigma / i!$, then $a_i$ is simply equal to quotient of $r_i/(i-1)!$. This is analogous to the fact that one can efficiently extract the $i$th digit base 10 of a number encoded in binary. Likewise, given new values of $a_i$ and $a_{i+1}$ one can easily update the value of $N_\sigma$ to its new value. 

One subtlety in this approach is that when extracting $a_i$ and $a_{i+1}$, we are very restricted in our use of ancillas. In particular, since $\DQC1$ circuits only have access to maximally mixed ancillas, we can only ever simulate the use of $O(\log n)$ pure ancillas. This is because, if we ensure the all zero string in the ancillas goes back to the all zero string at the end of the computation, then we can postselect them to be all 0's at the end of the $\DQC1$ computation. Since this occurs with 1/poly probability this is within our abilities. 

Therefore, in order to complete this argument, we will need to show that $a_i$ and $a_{i+1}$ can be extracted, and $N_{\sigma}$ be updated after the swap, using only $O(\log n)$ pure ancillas. In fact, we will show one can get away wth only $O(1)$ ancillas. This is because
given an integer $N$, for a fixed\footnote{Since our swaps are specified ahead of time in the description of $C$, we can hard-code the numbers we divide by into the circuit.} integer $k$, it is possible to compute the quotient $q$ and remainder $r$ of $N/k$ using $O(1)$ ancillas. Indeed the grade-school long division algorithm suffices for this task, and uses only $2$ ancillas (we thank Luke Schaeffer for pointing this out to us). Suppose one wishes to compute the quotient $q$ and remainder $r$ of $N/k = qk + r$. To do so, simply compute how many times $k$ divides the first $\lceil \log k \rceil$ bits of $N$, store this as the first bit $q_i$ of the quotient, and subtract $q_i 2^{\lceil \log N \rceil - \lceil \log k\rceil}$ from $N$. Repeat. One can easily see that since we're dividing in binary, for every bit we compute of the quotient (with the possible exception of the first bit), the leading bit of $N$ is set to zero.  Therefore we can reuse this space to store an additional bit of $q$. At the end of the computation one has $q$ and $r$ stored in $\lceil \log q\rceil + \lceil \log k\rceil  \leq \log N +2$ bits. We have therefore computed the quotient and remainder reversibly in place with only two ancillas. This suffices to prove one can extract $a_i$ and $a_{i+1}$ reversibly from $N_\sigma$ using only $O(1)$ ancillas. The proof that one can update $N_\sigma$ in-place reversibly after altering $a_i$ and $a_{i+1}$ follows analogously by running the above operation in reverse. Therefore this encoding of permutations can be used to estimate $\langle 123\ldots n | C | 123 \ldots n \rangle$ to 1/poly error in $\DQC1$.

Note that the construction is based on adjacent swaps only. If in the description of $C$ nonadjacent swaps are implemented, we can simulate these swaps by adjacent ones. We construct $U$ by approximating each adjacent $X$ gate in $C$.  Each such gate alters $O(\log n)$ bits and because of the Solovay-Kitaev theorem \cite{DawsonSolovayKitaev}, there exists a $\Poly(n, \log 1/\epsilon)$ size circuit that approximates each $X$ gate within error $\epsilon$.
\end{proof}

\subsection{Factorial number system}
\label{factorial}

The encoding of each permutation, $\sigma(1), \sigma(2),\ldots, \sigma(n)$ ($\sigma \in S_n$), is accomplished by a walk from root to each leaf of the following tree, $T_n$: consider a tree with its root located at node $0$, as we mark it to be distinct. Let node $0$ have degree $n$, with its children marked with numbers $1,2,3,\ldots, n$, from left to right. Denote these nodes by layer $1$. Let each node of layer $1$ have $n-1$ children, and label each child of node $i$ in layer $1$, by numbers $[n]-\{i\}$, in an increasing order from left to right. Construct the tree inductively, layer by layer: each node $k$ in layer $j$ have $n-j$ children, and the children labeled with numbers $[n]-L_k$. Where $L_k$ is the set of labels located on the path from node $0$ to node $k$. Therefore, nodes of layer $n$ have no children. The number of leaves of the tree is $n!$. For each leaf there is a unique path from root down to the leaf, and the indexes from top to down represent a permutation. This is because the indexes of each path are different from each other. Also each permutation $\sigma$ is mapped to a unique path in this tree: start from node $0$, pick the child with index $\sigma(1)$, then among the children of $\sigma(1)$, pick the child with index $\sigma(2)$ and so on. Therefore, this establishes a one-to-one map between the paths on $T_n$ and permutations of labels in $[n]$.

The next step is to provide a one-to-one mapping from the paths on the graph to bit strings of length $\log n! + O(n)$. First, label the edges of $T_n$ by the following. For each node of degree $p$, with children labeled with $x_0< x_1< \ldots< x_{p-1}$, label the edge incident to $x_0$ by $0$, the edge incident to $x_1$ by $1$, and so on. Given these edge labels, The construction is simple: represent each path with the bit string $a_n a_{n-1} \ldots a_0$, where $a_j$ is a bit string of length $\lceil \log j\rceil =\log j + O(1)$, is the binary representation of the label of the edge used in the $j$'th walk.

Note that the factorial numbers have local properties under swap. In order to apply a swap on this encoding one needs to alter only $O(\log n)$ bits. Suppose that the permutations $\sigma=\sigma(1), \sigma(2), \ldots, \sigma(k), \sigma(k+1), \ldots, \sigma(n)$ and $\pi=\sigma(1), \sigma(2), \ldots, \sigma(k+1), \sigma(k), \ldots, \sigma(n)$ are represented by the binary encoding $X=a_1, a_2, \ldots a_k, a_{k+1}\ldots, a_n$ and $Y=b_1, b_2, \ldots b_k, b_{k+1}\ldots, b_n$, respectively. Clearly, $\pi$ can be obtained from $\sigma$ by swapping the element $k$ and $k+1$. Notice that $a_1=b_1, a_2=b_2, \ldots , a_{k-1}=b_{k-1}$. This is because the corresponding path representations of the two permutations on $T_n$ walk through the same node at the $k-1$'th walk. Also $a_{k+2}=b_{k+2}, \ldots , a_n=b_n$. This is because the subtrees behind the $k+2$'th layer nodes in the two paths are two copies of the same tree, since their nodes consist of same index sets. Therefore, $X$ and $Y$ differ only at $a_k, a_{k+1}$ and $b_k, b_{k+1}$ substrings. As a consequence of these observations, the bit-string codes for two permutations that differ in adjacent labels only, are different in $O (\log n)$ bits.

\section{Review of Exchange Interactions}
\label{ExchangeReview}

We show how to use arbitrary initial states to obtain a programmable $\BQP$ universal model. This is done by demonstrating a reduction from the exchange interaction model of quantum computation which is already known to be $\BQP$ universal.

Here, we first briefly review the exchange interaction model \cite{kempe2001theory, bacon2001encoded,kempe2001encoded}, and then describe how to do a reduction from the computation in this model to the ball permuting model of computing on arbitrary initial states. Next, we sketch the proof of universality for the exchange interaction model, which in turn results in $\BQP$ universality of ball permuting model on arbitrary initial states. 

Consider the Hilbert space $(\C^2)^{\tensor n}=:\C \{0,1\}^n$, with binary strings of length $n$, $\mathcal{X}_n:=\{|x_1\rangle \tensor |x_2\rangle \tensor \ldots \tensor |x_n\rangle: x_j \in \{0,1\} \}$, as the orthonormal computational basis. Exchange interactions correspond to the unitary gates $T (\theta, i, j)=\exp (i \theta E_{(i,j)})=\cos\theta I + i \sin \theta E_{(i, j)}$, where the operator, $E_{(i,j)}$, called the exchange operator, acts as:

$$
E= \dfrac{1}{2}(I + \sigma_x \tensor \sigma_x + \sigma_y \tensor \sigma_y + \sigma_z \tensor \sigma_z)
$$

\noindent on the $i, j$ slots of the tensor product, and acts as identity on the other parts. More specifically, $E$ is the map:

\begin{eqnarray*}
&|00\rangle&\rightarrow \hspace{3mm}|00\rangle,\\
&|01\rangle&\rightarrow \hspace{3mm}|10\rangle,\\
&|10\rangle&\rightarrow \hspace{3mm}|01\rangle,\\
&|11\rangle&\rightarrow \hspace{3mm}|11\rangle.
\end{eqnarray*}

The action of $E_{i j}$ is very similar to the permuting operator $L_{(i, j)}$, except that $E$ operates on bits rather than the arbitrary labels of $[n]$. These operators are also known as the Heisenberg couplings, related to the Heisenberg Hamiltonian for spin-spin interactions.

\section{Reduction from Exchange Interactions}
\label{bqpuniversality}

Define $\mathcal{X}^k_n :=\{|x\rangle: x\in \{0,1\}^n, |x|_H = k\}$ to be the subset of $\mathcal{X}_n$, containing strings of Hamming distance $k\leq n$. Here, $|.|_H$ is the Hamming distance. Also, let $\C \mathcal{X}^k_n$ be the corresponding Hilbert space spanned by these basis.

\begin{theorem}
Given a description of $U=T (\theta_m, i_m, j_m)\ldots T (\theta_2, i_2, j_2) T (\theta_1, i_1, j_1)$, and an initial state $|\psi\rangle \in \C \mathcal{X}^k_n$, there exists an initial $|\psi'\rangle\in \C S_n$, and a ball permuting circuit, with $X$ operators, that can sample from the output of $U|\psi'\rangle$, exactly.
\label{hrd}
\end{theorem}

\begin{proof}
We show how to encode any state of $\C \mathcal{X}^k_n$ with states of $\C S_n$. Let $S_{k, n-k}$ be the subgroup of $S_n$ according to the cycles $\{1,2,\ldots, k\}$ and $\{k+1, k+2, \ldots, n\}$, and denote $|\phi_0\rangle = \dfrac{1}{\sqrt{k! (n-k)!}}\sum_{\sigma \in S_{k, n-k}} R(\sigma) |123\ldots n\rangle$ be an encoding of the state $|1^k 0^{n-k}\rangle$. Here, $1^k$ means $1$'s repeated for $k$ times. This is indeed a quantum state that is symmetric on each the labels of $\{1,2,\ldots, k\}$ and $\{k+1,k+2,\ldots, n\}$, separately.  Any string of Hamming distance $k$ can be obtained by permuting the string $0^k 1^{n-k}$. For any such string $x$ let $\pi_x$ be such a permutation, and encode $|x\rangle $ with $|\phi(x)\rangle := L_{\pi_x} |\phi_0\rangle$. Therefore, given any initial state $|\psi\rangle := \sum_{x \in \mathcal{X}^k_n} \alpha_x |x\rangle$, pick an initial state $|\psi'\rangle := \sum_{x \in \mathcal{X}^k_n} \alpha_x |\phi(x)\rangle$ in $\C S_n$. Now, given any unitary $U=T (\theta_m, i_m, j_m)\ldots T (\theta_2, i_2, j_2) T (\theta_1, i_1, j_1)$ with $T$ operators, pick a corresponding ball permuting circuit $U'=X (\theta_m, i_m, j_m)\ldots X (\theta_2, i_2, j_2) X (\theta_1, i_1, j_1)$. It can be confirmed that for any $i< j \in [n]$ if $E_{(i, j)} |x\rangle = |x'\rangle$, then $E_{(i, j)} |\phi(x)\rangle = |\phi(x')\rangle$. From this, if $U|\psi\rangle=\sum_{x \in \mathcal{X}^k_n} \beta_x |x\rangle$, then $U'|\psi'\rangle=\sum_{x \in \mathcal{X}^k_n} \beta_x |\phi(x)\rangle$.

 It remains to show that given access to the output of $U'|\psi'\rangle$, one can efficiently sample from $U|\psi\rangle$. Suppose that $U'|\psi'\rangle$ is measured in the end, and one obtains the permutation $\sigma= (\sigma(1),\sigma(2),\ldots, \sigma (n))$. Then, by outputting a string $x$ by replacing all the labels of $\{1,2,\ldots, k\}$ in $\sigma$ with ones and the other labels with zeros the reduction is complete. The probability of obtaining any string $x$ with this protocol is exactly equal to $| \langle x |U|\psi\rangle |^2$.
\end{proof}

Universal quantum computing is possible by encoding a qubit using three spin $1/2$ particles. Suppose that the following initial states are given in $\C \mathcal{X}^1_3$:

$$
|0_L\rangle := \dfrac{|010\rangle-|100\rangle}{\sqrt{2}}
$$

and,

$$
|1_L\rangle :=\dfrac{2|001\rangle-|010\rangle-|100\rangle}{\sqrt{6}},
$$

\noindent
as some logical encoding of a qubit using three quantum digits. We claim that there is a way to distinguish $|0_L\rangle$ from $|1_L\rangle$ with perfect soundness. These mark the multiplicity space of the space with half $Z$ direction angular momentum and half total angular momentum. First, we should find a way to distinguish between these two states using measurement in the computational basis. Suppose that we have access to $k$ copies of an unknown quantum state, and we have the promise that it is either $|0_L\rangle$ or $|1_L\rangle$, and we want to see which one is the case. The idea is to simply measure the third bit of each copy, and announce it to be $0_L$ if the results of the $k$ measurements are all $0$ bits. If the state has been $|0_L\rangle$, the probability of error in this decision is zero, because $|0_L\rangle= \dfrac{|01\rangle-|10\rangle}{\sqrt{2}}\tensor |0\rangle$. Otherwise, we will make a wrong decision with probability at most $(1/3)^k$, which is exponentially small. This is because the probability of reading a $0$ in the third bit of $|1_L\rangle$ is $1/3$. 

\begin{theorem}
There is a way of acting as encoded $SU(2)$ on the span of $\{|0_L\rangle,|1_L\rangle \}$, and also $SU(4)$ on the concatenation of two encoded qubits. 
\end{theorem}

\begin{proof}(Sketch) according to the analysis of \cite{divincenzo2000universal, kempe2001theory}, one can look at the Lie algebra of the exchange operators to find encoded $\su(2)$ algebra on the encoded qubit. Also, we need to take enough commutations such that the action of the designed operators annihilates the two one dimensional spaces spanned by $|000\rangle$ and $|111\rangle$. The authors of \cite{kempe2001theory} prove that there is a way to act as $SU(V(s,m))$ on each invariant subspace $V(s,m)$. Here $V(s,m)$ is the subspace corresponding to total angular momentum $s$ and $z$ direction angular momentum $m$. Moreover, they prove that the action on two subspaces $V(s_1, m_1)$ and $V(s_2, m_2)$ can be decoupled, unless $s_1=s_2$, and $m_2=-m_1$, where the two subspaces are isomorphic. It is almost enough to prove that the state $|0_L\rangle \otimes |0_L\rangle$ is contained in non-isomorphic invariant subspaces. However, this is also true, since $|0_L\rangle \otimes |0_L\rangle$ is completely contained in subspaces with $m=2$.
\end{proof}

See \cite{bauer2014universality,kempe2002exact,wu2002power} for similar models with encoded universality. Therefore, this is a nonconstructive proof for the existence of an encoded entangling quantum gate; CNOT for example. Indeed, the actual construction of a CNOT is given in \cite{divincenzo2000universal} . Notice that for a decision problem, one can formulate quantum computation in such a way that only one qubit needs to be measured in the end, and this can be done by distinguishing $|0_L\rangle$ and $|1_L\rangle$ using measurement in the computational basis. The probability of success in distinguishing between the two bits can also be amplified by just repeating the computation for polynomial number of times, and taking the majority of votes. Also, taking the majority of votes can be done with encoded CNOTs and single qubits gates on a larger circuit, and without loss of generality we can assume that one single measurement on one single qubit is sufficient.

\section{Partial Classification of Quantum Computation on Different Initial States}
\label{classification}

In the following, it is proved that the ball permuting gates act densely on invariant subspaces corresponding to Young tableaus with two rows or two columns. The proof is based on the bridge lemma and decoupling lemma of reference \cite{aharonov2011bqp}. As we discuss, conditioned on the existence of a bridge operator, and decoupled dense action on two orthogonal subspace of different dimensionality, the bridge lemma glues the two subspaces into a larger subspace with dense action on it. Also, the decoupling lemma decouples action on two orthogonal subspaces of different dimensionality, given dense action on each of them. Consulting with \cite{kuperberg2011denseness}, it is conceivable that these two lemmas have natural generalizations to more than two subspaces and subspaces that have equal dimensionality. We conjecture that using these tools one can prove that the action of ball permuting gates is dense on all invariant subspaces of the symmetric group, even for those which correspond to Young diagrams of more than two rows/columns. We leave this investigation to further work. 

In this section, the Lie algebra and the unitary Lie group generated by $X$ operators are used interchangeably. The Hilbert space $\C S_n$ has the decomposition:

$$
\C S_n \cong \bigoplus_{\lambda \vdash n} V_\lambda \otimes X_\lambda
$$

Let $G$ be the unitary group generated by these $X(\theta,k)=\exp (i\theta L_{(k,k+1)}).$ operators. As described earlier, the space tangent to the identity element of $G$ is a Lie algebra, $g$, which contains $L_{(k,k+1)}$ for all $k\in [n-1]$, and is close under linear combination over $\mathbb{R}$, and the Lie commutator $i [\cdot, \cdot]$. The objective is to show that for any $\lambda\vdash n$ with two rows or two columns, and any element $U$ of $SU(V_\lambda)$, there is an element of $G$ that is arbitrarily close to $U$.

The proof is presented inductively. First of all, for any $n$, the irreps $V_n$ and $V_{1,1,1,\ldots, n}$ are one dimensional, and the action of $x \in G$ is to add an overall phase. However, observing the structure of YY basis for these irreps, the action of $G$ on the joint blocks $V_n \oplus V_{(1,1,1, \ldots, 1)}$ cannot be decoupled, and the projection of $G$ onto these subspaces is diagonal, and moreover isomorphic to the group $e^{i\theta}\times e^{-i \theta}: \theta \in \mathbb{R}$. Intuitively, these are Bosonic  and Fermionic subspaces, where an exchange $L_{(k,k+1)}$ of particles results in a $+1$ and $-1$ overall phase, respectively. 

For $n=2$, the only invariant subspaces are $V_2$ and $V_{(1,1)}$, and we know the structure of these irreps from the last paragraph:

$$
\C S_2\cong V_2 \oplus V_{(1,1)}, \hspace{1cm} G \twoheadrightarrow e^{i\theta}\times e^{-i \theta}: \theta \in \mathbb{R}.
$$

For $n=3$, the decomposition is according to:

$$
\C S_3 \cong V_{3} \oplus V_{(1,1,1)} \oplus V_{(2,1)}\otimes X(2).
$$

Here, $X(2)$ is a two dimensional multiplicity space. There are two standard $(2,1)$ tableaus and therefore $V_{(2,1)}$ is also two dimensional. Observing the YY basis the two generators $L_{(1,2)}$ and $L_{(2,3)}$ take the matrix forms:

$$
L_{(1,2)} = 
 \begin{pmatrix}
  1 &0 \\
  0&-1
 \end{pmatrix},
$$

\noindent and,

$$
L_{(2,3)} = 
 \begin{pmatrix}
  -1/2 &\sqrt{3}/2 \\
  \sqrt{3}/2&1/2
 \end{pmatrix}.
$$

\noindent The basis of the matrix are marked with the two standard Young tableaus of shape $(2,1)$. The first basis corresponds to the numbering $(1,2; 3)$ and the second one corresponds to $(1,3;2)$. Here, the rows are separated by semicolons. The following elements of the Lie algebra $g$ generate $\su(V_{(2,1)})$ and annihilate the two Bosonic and Fermionic subspaces:

$$
\dfrac{1}{2\sqrt{3}}[L_{(1,2)}, [L_{(1,2)},L_{(2,3)}]]= 0 \oplus 0 \oplus \sigma_x\otimes I,
$$

$$
\dfrac{i}{\sqrt{3}}[L_{(1,2)}, L_{(2,3)}]= 0 \oplus 0 \oplus \sigma_y\otimes I,
$$

\noindent and,

$$
\dfrac{1}{6}[[L_{(1,2)},[L_{(1,2)}, L_{(2,3)}]],[L_{(1,2)}, L_{(2,3)}]]= 0 \oplus 0 \oplus \sigma_z \otimes I.
$$

This implies the denseness of $G$ in $1 \times 1 \times SU(V_{(2,1)})$. Therefore, we obtain a qubit coupled to the multiplicity space, placed in a superposition of the one dimensional Bosonic and Fermionic subspaces. So, projecting onto a subspace like $V_{(2,1)}\otimes |\psi\rangle$, for $|\psi\rangle \in X(2)$, we obtain a qubit.

We use this result as the seed of an induction. The upshot is to add boxes to $(2,1)$ one by one, in a way that the partitions remain with two rows or two columns. At each step, we use the branching rule to combine the blocks together to larger and larger special unitary groups. In the course of this process, we use two important tools, called the bridge lemma, and decoupling lemma:

\begin{lemma}
(Aharonov-Arad \cite{aharonov2011bqp}) let $A$ and $B$ be two orthogonal subspaces, with \textit{non-equal} dimensions, $\dim A < \dim B$:

\begin{itemize}
\item (Bridge) if there is some state $|\psi\rangle \in A$, and a (bridge) operator $V \in SU (A \oplus B)$, such that the projection of $V |\psi\rangle$ on $B$ is nonzero, then the combination of $SU(A)$, $SU(B)$, and $V$ is dense in $SU(A \oplus B)$.

\item (Decoupling) suppose for any elements $x\in SU(A)$ and $y \in SU(B)$, there are two corresponding sequences $I_x$ and $I_y$ in $G$, arbitrarily close to $x$ and $y$, respectively, then the action of $G$ on $A \oplus B$ is decoupled, \i.e., $SU(A)\times SU(B)\subseteq G$. 

\end{itemize}
\label{bridge}
\end{lemma}

See \cite{aharonov2009polynomial,aharonov2007polynomial} for more similar results. Intuitively, what bridge lemma says is that given two subspaces, with one of them larger than the other, dense action each, along with a bridge between them, implies denseness on the combined subspace. That is a bridge glues them to a larger special group. The condition of different dimensions is a crucial requirement for the application of this lemma. The decoupling lemma, on the other hand, states that given dense action on two subspaces, as long as they have different dimensionality, there is way of acting on the two subspaces independently. Again, in this case non-equal dimensionality is important. For example, suppose that $\dim A=\dim B$, then the action $x \times \bar{x}: x \in SU(A)$, cannot be decoupled. Here $\bar{x}$ is the complex conjugate of $x$, \i.e., entries $\bar{x}$ as a matrix are complex conjugates of corresponding entries of matrix $x$. In order to see this, just notice that after finite compositions, the general form of elements generated in this way is $(x_1 x_2 \ldots x_n) \times \overline{(x_1 x_2 \ldots x_n)}$, and an identity action on the left part forces identity action on the right part of the Cartesian product.

Next, we show that the lemma along with the branching rule, force denseness on all irreps corresponding to partitions of two rows or two columns. We will take care of the case with two rows. The situation with two columns is similar. As a way of induction, suppose that, for any $m<n$, for any  $\lambda = (\lambda_1\geq \lambda_2) \vdash m$, the projection of $G$ on $\lambda$ is dense in $SU(V_\lambda)$. The objective is to prove denseness for any partition $\mu \vdash n$.

This is true for $(2,1)$, as showed above. For the sake of illustration, we prove this for $n=4$. The partitions $(4)$ is immediate, because this is one dimensional. Also, the partition $(2,2)$ is immediate, since the branching rule, under the action of $S_3$ is:

$$
V_{(2,2)} \cong V_{(2,1)},
$$

\noindent That is the only removable box from $(2,2)$ is the last box, and in the YY basis for $(2,2)$, this last box can contain the symbol $4$ only. So, the same operators of $S_3$ act densely on this subspace.

The situation with the partition $(3,1)$ is a little different. Analyzing the hook lengths, $V_{(3,1)}$ has dimension $3$, and the branching rule involves the direct sum of partitions $(2,1)$ and $(3)$:

$$
V_{(3,1)} \cong V_{(2,1)} \oplus V_{(3)}.
$$

\noindent Where, $V_{(2,1)}$ is two dimensional, and $V_{(3)}$ is one dimensional, and therefore, they have non-equal dimensions, and also their direct sum adds up to dimension $3$. From, the analysis of $S_3$ we know that independent $SU(2)$, and $SU(1)=\{1\}$ is possible on these irreps. It suffices to find a bridge operator in $SU(V_{(2,1)} \oplus V_{(3)})$. In the first glance, the operator $L_{3,4} \in g$ sounds like a suitable choice. However, there is a problem with this: the restriction of $L_{3,4}$ on $V_{(3,1)}$ is not traceless, and therefore the image under exponentiation does not have unit determinant. Therefore, a wise choice for a bridge operator is $i [L_{(2,3)}, L_{(3,4)}]$. Looking at the actual matrices, restricted to the YY basis of $(3,1)$, one finds $i [L_{(2,3)}, L_{(3,4)}]$, as a suitable bridge, that is nice and traceless:

$$
i 
\begin{pmatrix}
0&\sqrt{2}&-\sqrt{\dfrac{2}{3}}\\
-\sqrt{2}&0 &\sqrt{\dfrac{1}{3}}\\
\sqrt{\dfrac{2}{3}} &-\sqrt{\dfrac{1}{3}}&0\\
\end{pmatrix}.
$$

\noindent Here the matrix is written in the basis corresponding to the tableaus $(1,2,3;4) , (1,2,4; 3)$  and $(1,3,4; 2)$. The bridging is between the $(1,2)$ and $(2,1)$ elements of the matrix. Thereby, the bridge lemma implies the desired denseness.

For general $n$, two situations can happen, either the partition under analysis is of the form $(\nu, \nu)=(n/2, n/2)$ (for even $n$ of course), or not. In the first case, the situation is similar to the partition $(2,2)$ of $n=4$. Thereby, restricted to $S_{n-1}$:

$$
V_{(\nu, \nu)}\cong V_{(\nu, \nu-1)},
$$

\noindent and based on the induction hypothesis the image of $G$ is already dense in the subspace. In the second case, also two cases can happen: either the partition has the form $\mu=(\nu+1, \nu)$, with $2 \nu+1 =n$, or not. In the first case, the branching rule is according to:

$$
V_{(\nu+1, \nu)}\cong V_{(\nu, \nu)}\oplus V_{(\nu+1, \nu-1)}
$$

\noindent The space $V_{(\nu, \nu)}$ corresponds to all YY basis corresponding to tableaus, wherein the index $n$ is located in the last box of the first row. Therefore, the index $n-1$ in all of the tableaus of $(\nu, \nu)$ is located in the last box of the second column, because this is the only removable box available. For simplicity, let's call this space $V_1$. The YY bases of $V_{(\nu+1, n-1)}$ correspond to all the tableaus of $(\nu+1, \nu)$, where the index $n$ is located in the last box of the second row. In this space, the location of the index $n-1$ is either in the last box in the first row or in the box right at the left of the last box in the second row. A coarser stratification of the states in $V_{(\nu+1, \nu-1)}$ is by grouping the YY basis according to the location of $n-1$. Let $V_2$ be the first one, and $V_3$ the second one. Therefore, $YY$ bases of $V_{(\nu+1, \nu)}$ can be grouped in three ways, $V_1, V_2, V_3$, corresponding to all the ways that one can remove two boxes from the original $V_{(\nu+1, \nu)}$. Again, a neat candidate for a bridge is $L_{(n-1,n)}$. Taking a closer look at the operator $L_{(n-1,n)}$, it can be decomposed according to:

$$
L_{(n-1,n)}=\sum_{|j\rangle \in V_3} |j\rangle \langle j|+ \dfrac{1}{2} \sum_{\substack{k' : k\\|k\rangle \in V_1\\|k'\rangle \in V_2}} |k\rangle\langle k| - |k' \rangle \langle k'|+\sqrt{\dfrac{3}{2}} \sum_{\substack{k' : k\\|k\rangle \in V_1\\|k'\rangle \in V_2}} |k\rangle\langle k'| + |k' \rangle \langle k|
$$

\noindent $|j\rangle$, $|k\rangle$, and $|k'\rangle$, of $V_1$, $V_2$, and $V_3$ are the corresponding orthonormal basis in the spaces. Notice that the space $V_1$ is isomorphic to $V_2$, and $k : k'$, refers to this isomorphism. Clearly, the restriction of $L_{(n-1,n)}$ to this block is not traceless, and indeed $tr_{V_{(\nu+1,\nu)}}=\dim V_3 = \dim V_{(\nu+1, \nu-2)}$.

Now, we use the decoupling lemma of Aharonov-Arad. $V_{(\nu+1, \nu)}$ and $V_{(\nu,\nu)}$ have different dimensionality, and also, due to the induction hypothesis the operators can act as the special unitary group on each of them. Thereby, there is a way to act as $x\oplus 0$ on the joint space $V_{(\nu,\nu)}\oplus V_{(\nu+1, \nu-1)}$, for some traceless element $x \in \su(V_{(\nu, \nu)})$. Therefore, $x |j\rangle =0$ and $x |k'\rangle$, for all $|j\rangle \in V_3$, $|k'\rangle \in V_2$. And denote $|x k\rangle := x |k\rangle$, for $|k\rangle \in V_1$. Taking the commutator $i [x, L_{(n-1,n)}]$:

$$
i [x, L_{(n-1,n)}]= \dfrac{i}{2} \sum_{\substack{k' : k\\|k\rangle \in V_1\\|k'\rangle \in V_2}} |x k\rangle\langle k|-|k\rangle\langle x k|+ i \sqrt{\dfrac{3}{2}} \sum_{\substack{k' : k\\|k\rangle \in V_1\\|k'\rangle \in V_2}} |x k\rangle\langle k'| - |k' \rangle \langle x k|.
$$

\noindent Clearly, this operator is traceless, Hermitian, and also one can choose $x$ in such a way that the bridging term in the second sum is nonzero.

Given the above proof for the case $V_{(\nu+1, \nu)}$, we will use a similar technique to take care of the situation $V_{(p, q)}$, where $p> q+1$, and $p+q=n$. Again, the branching rule is:

$$
V_{(p, q)} = V_{(p, q-1)}\oplus V_{(p-1, q)}.
$$

\noindent The space $V_{(p,q-1)}$ corresponds to all YY bases that correspond to the tableaus where the index $n$ is located at the last box of the first row. In this space, the index $n-1$ is either located at the left side of the box containing $n$, or it is located in the last box of the second row. Call the space corresponding to the first (second) one $V_1$ ($V_3$).  $V_{(p-1,q)}$ corresponds to all YY bases of tableaus with index $n$ is located at the last box of the second row. In this space, the index $n-1$ is either located at the left side of the box containing $n$, or it is located in the last box of the first row. Call the first space $V_2$ and the second one $V_4$. Again, write the decomposition of $L_{(n-1, n)}$, accordingly:

$$
L_{(n-1,n)}=\sum_{|j\rangle \in V_1} |j\rangle \langle j|+\sum_{|j\rangle \in V_2} |j\rangle \langle j|+ \alpha(p,q) \sum_{\substack{k' : k\\|k\rangle \in V_3\\|k'\rangle \in V_4}} |k\rangle\langle k| - |k' \rangle \langle k'|+\beta(p,q) \sum_{\substack{k' : k\\|k\rangle \in V_3\\|k'\rangle \in V_4}} |k\rangle\langle k'| + |k' \rangle \langle k|
$$

\noindent Here:

$$
\alpha(p,q) = \dfrac{1}{p-q+1}
$$

\noindent and,

$$
\beta(p,q)=\sqrt{1-\dfrac{1}{(p-q+1)^2}}.
$$

Once again, $V_2$ is isomorphic to $V_3$, and $k: k'$ denotes the correspondence between elements of the two spaces. Once again, we use the decoupling lemma, which asserts the existence of elements like $X:=x\oplus 0$, and $Y:=0 \oplus y$, on $V_{(p,q-1)}\oplus V_{(p-1,q)}$, for every $x\in \su(V_{(p,q-1)})$ and $y \in \su (V_{(p,q-1)})$. A bridge between $V_3$ and $V_4$ is needed, in such a way that the bridge annihilates both $V_1$ and $V_2$. A candidate for a bridge is $[Y,[X, L_{(n-1,n)}]]$. However, it can be easily shown that the element $i[X, L_{(n-1,n)}]$ will also work. The operator $X$ annihilates everything in $V_2$ and $V_4$. Therefor, taking the commutator, the second sum is annihilated, and also, all the remaining terms are traceless and one can find $x$ in such a way that the bridge part is nonzero. All the above results also apply to the tableaus with two columns.

\section{Remarks on the Exchange interactions}
\label{remarksexchange}
In this section we emphasize on some connections between the Hilbert space of permutations and the Hilbert space of qubits in the exchange interaction model. 

For any $k\in [n]$, $\C \mathcal{X}^k_n$ is an invariant subspace of the group, $G_T$, generated by $T$ operators. This is because the exchange operators do not change the Hamming distance of the computational basis. So the decomposition $\C \{0,1\}^n \cong \bigoplus_{k} \C \mathcal{X}^k_n$ is immediate. Consider the standard total $Z$ direction angular momentum operator:

$$
J_Z:= \dfrac{1}{2} (\sigma^1_z+\sigma^2_z+\ldots+\sigma^n_z).
$$

\noindent
Then $[J_Z, E_{(i,j)}]=0$ for all $i$ and $j$. Here, the superscript $j$ in $A ^j$ for operator $A$ means $I\tensor I \tensor \ldots \tensor \overset{\overset{j}{\downarrow}}{A} \tensor \ldots \tensor I$, the action of the operator on the $j$'th slot of the tensor product. $J_Z$ indeed counts the Hamming distance of a string, and more precisely, for any $|\psi\rangle \in \C \mathcal{X}^k_n$, $J_Z|\psi\rangle = (\dfrac{n}{2}-k)|\psi\rangle$. Therefore, the eigenspace corresponding to each eigenvalue of $J_Z$ is an invariant subspace of $G_T$. For each eigenvalue $n/2 - k$ the multiplicity of this space is ${n}\choose {k}$, the number of $n$ bit strings of Hamming distance $k$. One can also define the $X$ and $Y$ direction total angular momentum operators in the same way:

$$
J_X:= \dfrac{1}{2} (\sigma^1_x+\sigma^2_x+\ldots+\sigma^n_x),
$$

and,

$$
J_Y:= \dfrac{1}{2} (\sigma^1_y+\sigma^2_y+\ldots+\sigma^n_y).
$$

\noindent Indeed, consulting the decoherence free subspaces theory of the exchange operators, the algebra generated by the operators $J_X, J_Y$ and $J_Z$, is the unique commutant of the exchange operators, and vice versa. Indeed, for any positive algebra that is closed under the conjugation map, the commutant relation is an involution \cite{james1981representation}, i.e., the commutant of the commutant of any such algebra is the algebra itself.

The decomposition of $\C \{0,1\}^n$, of $n$ spin $\dfrac{1}{2}$ particles, is well known, and can be characterized by total angular momentum, and the $Z$ direction of the total angular momentum. The total angular momentum operator is:

$$
J^2 = (\sum_{j\in [n]}\dfrac{1}{2}\sigma^j_x)^2+(\sum_{j\in [n]} \dfrac{1}{2}\sigma^j_y)^2+(\sum_{j\in [n]}\dfrac{1}{2}\sigma^j_z)^2.
$$

\noindent Indeed, using a minimal calculation one can rewrite $J^2$ as:

$$
J^2= n(n-1/4) + \sum_{i<j \in [n]} E_{(i,j)},
$$

\noindent and it can be confirmed that for all $k<l \in [n]$, $[E_{(k,l)}, J^2]=0$. In other words, the exchange operators do not change the total and $Z$ direction angular momentum of the a system of spin $1/2$ particles. The decomposition of $\C \{0,1\}^n$ can be written down according to these quantum numbers. Let $V(s)\subset \C \{0,1\}^n$, be the set of states $|\psi\rangle$ in $\C \{0,1\}^n$ such that $J^2|\psi\rangle = s(s+1/2)|\psi\rangle$, and $V(s,m)\subset V(s)\subset \C \{0,1\}^n$, as the subspace with states $|\phi\rangle$ such that $J_Z |\phi\rangle = m/2 |\phi\rangle$. If $n$ is even, $\C \{0,1\}^n$ decomposes according to:

$$
\C \{0,1\}^n \cong V(0) \oplus V(1) \oplus \ldots \oplus V(n/2),
$$

\noindent
and each of these subspaces further decomposes to:

$$
V(s)\cong V(s,-s)\oplus V(s,-s+1) \oplus \ldots \oplus V( s, s).
$$

\noindent For odd $n$ the only difference is in the decomposition $\C \{0,1\}^n \cong V(\dfrac{1}{2}) \oplus V(\dfrac{3}{2}) \oplus \ldots \oplus V(n/2)$. From what is described in the context of decoherence free subspaces theory, the exchange interaction can affect the multiplicity space of each subspace $V(s,m)$. We are interested in the subspaces of the form $V(s,0)$ for even $n$, and $V(s, \pm \dfrac{1}{2})$, for odd $n$, which correspond to the decomposition of $\mathcal{X}^{n/2}_n$, and $\mathcal{X}^{(n \pm 1)/{2}}_n$, based on the total angular momentum, respectively.

There is a neat connection between the multiplicity space of these subspaces, and the subgroup adapted YY bases. For $k\in [n]$, define the following series of operators:

$$
J_k^2 = k(k-1/4) + \sum_{i<j \in [k]} E_{(i,j)}.
$$

\noindent Clearly, $J_k^2=J^2$. These are indeed the total angular momentum measured by just looking at the first $k$ particles. Using a minimal calculation one gets $[J^2_k , J^2_l]=0$ for all $k,l$. That is they are all commuting, and they can be mutually diagonalized. For $x_j \in [n]$, let $|x_1, x_2, \ldots, x_n\rangle$, be such basis with $J^2_j |x_1, x_2, \ldots, x_n\rangle=x_j (x_j + \dfrac{1}{2})|x_1, x_2, \ldots, x_n\rangle$. These are appropriate candidates as a basis for the multiplicity space of $V(s,0)$($V(s,1/2)$ for odd $n$). Then, $x_n=s$. Analyzing these operators more carefully, it is realized that for each $l<n$, either $x_{l+1}= x_l + 1/2$ or $x_{l+1}= x_l - 1/2$. Intuitively, this is saying that adding a new spin $1/2$ particle $Q=\C^2$ to $V(x_j)$:

$$
V(x_j)\tensor Q \cong V(x_j+\dfrac{1}{2})\oplus V(x_j-\dfrac{1}{2}),
$$

\noindent for $x_j>0$, and otherwise:

$$
V(0)\tensor Q \cong V(\dfrac{1}{2}),
$$

\noindent This is similar to the branching rule of the symmetric group representation theory. The second form is directly related to the branching rule of $V_{(n,n)}\cong V_{(n,n-1)}$. For simplicity, here we consider the twice of the $J$ operators instead, so that the branching rule takes the form:

$$
V(x_j)\tensor Q \cong V(x_j+1)\oplus V(x_j-1),
$$

\noindent for $x_j>0$ and,

$$
V(0)\tensor Q \cong V(1),
$$


\begin{figure}[tp]
\centering
\includegraphics[height=3.0in]{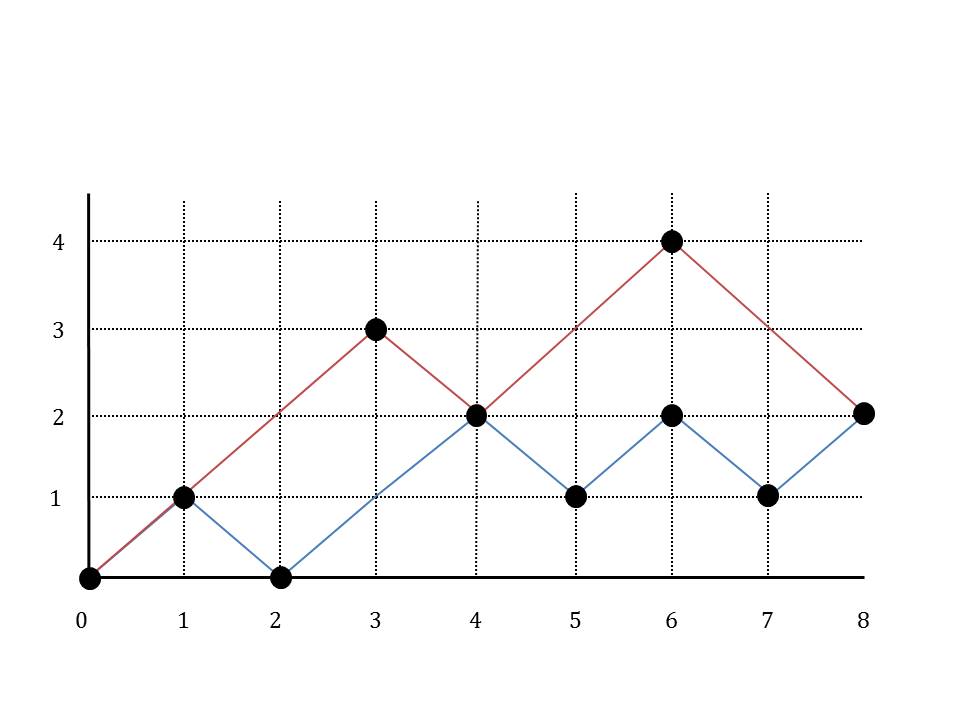}
\caption[Example of a path model]{An example of a path model. The blue and red paths start out of $(0, 0)$ and end up with the point $(8,2)$. YY basis corresponding to tableaus that two rows are closely related to the path model.
}
\label{Path}
\end{figure}

\noindent otherwise. Applying this rule recursively, the path model is obtained. See Figure \ref{Path} for an example. A path model $P_s$ is the set of paths between two points $(0,0)$ and $(n, s)$ in a two dimensional discrete Cartesian plane $\{0,1,2,3,\ldots, n\}^2$, where no path is allowed to cross the $(x,0)$ line, and at each step the path will move either one step up or one step down. Path up/down from the point $(x,y)$ is the connection from this point to $(x+1, y+1)/(x+1, y-1)$. See Figure \ref{Path} for an example of a path model. Therefore, the Hilbert space $V(s)$ corresponds to the orthonormal basis labeled by the Paths to the point $(n,2s)$. 

Indeed, we could agree on a path model for the YY basis of the tableaus with two rows. Let $\lambda$ be the Young diagram of shape $\lambda= (n,m)$. The path model is constructed in the following way: map each tableau $t$ to a symbol, $M_t= y_1 y_2 \ldots y_{n+m}$, where $y_j$ is the row index of the box containing the number $j$. Starting at the point $(0, 0)$ then the path corresponding to $t$ is constructed by taking a step up, whenever a $1$ is read in $M$, and a step down otherwise. Thereby, $P_0$ corresponds to $(n,n)$, and $P_{2n}$ corresponds to the partition $(2n)$.

\section{Programming $\RQBall$}
\label{PostProgramming}

The goal is to come up with a quantum algorithm based on particle scattering in $1+1$ dimension, which takes the description $\langle C \rangle$ of a general $X$ quantum ball permuting model as an input, and outputs the description of a sequence of particle scatterings and a sequence of intermediate non-adaptive demolition particle measurements, in a way that the overall process efficiently samples from the output of $C$. The construction of this section is very similar to the nondeterministic gates of \cite{knill2001scheme,knill2000efficient}. For a review of quantum computing with intermediate measurements see \cite{leung2004quantum, terhal2002adaptive,briegel2009measurement}.

Consider the $X$ ball-permuting gate of Figure \ref{fig1}, where we let the two input  wires interact with arbitrary amplitudes, and in the end we measure the label locations of $A$ and $B$. The objective is to have a particle scattering gadget that can simulate the output distribution of this circuit. Therefore, we can use the four particle gadget of Figure \ref{fig2}. The left and right rectangles demonstrate demolition measurements and the final superposition is created at the locations $A$ and $B$. The overall scattering process acts as a nondeterministic gate, in the sense that the gadget succeeds its simulation, only if the left detector measures label $a$ and the right detector measures label $b$, and an experimenter can verify this in the end. The velocities $v_1, v_2, v_a$ and $v_b$ can be tuned in such a way that the desired swap is obtained. The probability of success, thereby, depends on these velocity parameters. More precisely, conditioned on a successful simulation, the overall action of the scattering gadget is the gate $X(\tan^{-1 }z_{eff}, 1)$, where:

$$
\tan^{-1} z_{eff}= \tan^{-1} z_{1}+\tan^{-1} z_{2}
$$

\noindent with $z_1=v_1-v_2$ and $z_2= v_a-v_b$. As a result of this, the left and right black output particles will have velocities $v_b$ and $v_a$, respectively.   Notice that all of these results still hold if the black particles start out of arbitrary initial superpositions. However, one should make sure that the state of the ancilla particles are separable from the black ones.

Moreover, as described, in this model of scattering the particles move on straight line in time-space place, and they do not naturally change their directions. We thereby can use a two particle gadget of Figure \ref{fig3} to navigate the particles' trajectories. The two particles collide from left to right, and the left particle is measured in the end. Conditioned on the detector measuring the label $a$, the navigation is successful, and the outcome of this process is particle with its original label $1$ moves to the right direction with velocity $v_a$. One can match $v_a=v_1$, so that the overall action of the nondeterministic gadget is a change of direction. The success probability, then depends on $v_1$ and $v_a$.

As another example consider the $X$ quantum ball permuting circuit of Figure \ref{fig4}. This circuit consists of $X$ gates, $1, 2$ and $3$, and they permute labels of the four input wires. In the end we measure the output wires $A, B, C$ and $D$, in the particle label basis. We use the particle scattering sequences of Figure \ref{fig5} to simulate this circuit. Again the blue particles are ancilla, and the black particles correspond to the wires, and the labels $1, 2$ and $3$, correspond to the simulation of gates $1, 2$ and $3$ in Figure \ref{fig4}, respectively. Each of the detectors measure in the particle label basis, and in the end the experimenter measures the particle locations $A, B, C$ and $D$, corresponding to the output wires $A, B, C$ and $D$, in Figure \ref{fig4}, respectively. The overall scattering process succeeds in its simulation only if the detectors measure the ancilla particles with their initial labels. That is, conditioned on all blue particles successfully pass through their intermediate interactions and bouncing off the last interaction, the scattering process simulates the circuit successfully. This is true, because the particles move on straight lines, and the only event corresponding to detecting an ancilla particle with its original label is the one where it never bounces off in its intermediate interactions, and bounces off its final collision before moving to the detector. For an example of a larger simulation see the simulation of the $X$ quantum circuit of Figure \ref{fig6} with the scattering process of Figure \ref{fig7}. This example specially, demonstrates that during the scattering, blue (ancilla) particles can experience many intermediate interactions, and the number of these interactions can scale linearly in the number of particles being used. Therefore, the event corresponding to a successful simulation can have exponentially small probability.

It is important to mention that because of the Yang-Baxter equation, braiding of two particles is impossible. Braiding means that two particles can interact with each other over and over, however, because of the expression of unitarity, $H(u) H(-u) = I$, two successive collisions is equivalent to no collision. The role of the intermediate measurements is to allow two particles interact over and over without ending up with identity. 

\subsection{Stationary Programming}

The simulations of last section are both intuitive and instructive. However, they have a drawback. The slope of the lines corresponding to particle trajectories, depend on the velocities of the particles. So for large simulations, we need to keep the track of the architecture of collisions, and the amplitudes of interactions at the same time, and this can be both messy and difficult. In this section, we try to present a better simulation scheme where one only needs to keep track of amplitudes, and the architecture of collisions can be tuned easily. The philosophy is to have steady particles, in the beginning, and whenever we want a ball permuting gate, a number of ancilla particles are fired to the target steady particles. Then the intermediate detections are used, and then postselections on their outcomes enables the model to simulate an arbitrary $X$ quantum ball permutation. By stationary particle we mean a particle that is not moving. In order to fulfill this purpose, we use the stationary gadget of Figure \ref{fig8}. The objective is to impose a desired permutation on the input black particles. And we want the black particles to stay stationary in the end of the simulation. In order to do this, two other stationary ancillas are put at the left and right of the black particles. Then, two other ancilla particles, the desired velocities, are fired from left and right, and postselection is made on them bouncing off from the black particles. Then the two black particles interact and exchange momenta, and then they collide with the two stationary ancillas. In the end, we measure and postselect on the ancilla particles bouncing off the black particles. Therefore, in the end of the process, the stationary black particles are left stationary, and the desired superposition is obtained. In order to see an example for the implementation of the stationary programming in larger circuits, see the simulation of $\XQBALL$ circuit of Figure \ref{fig9} with stationary particle programming of Figure \ref{fig10}. 

\begin{figure}[tp]
\centering
\begin{subfigure}{.5\textwidth}
  \centering
{\includegraphics[height=2.0in]{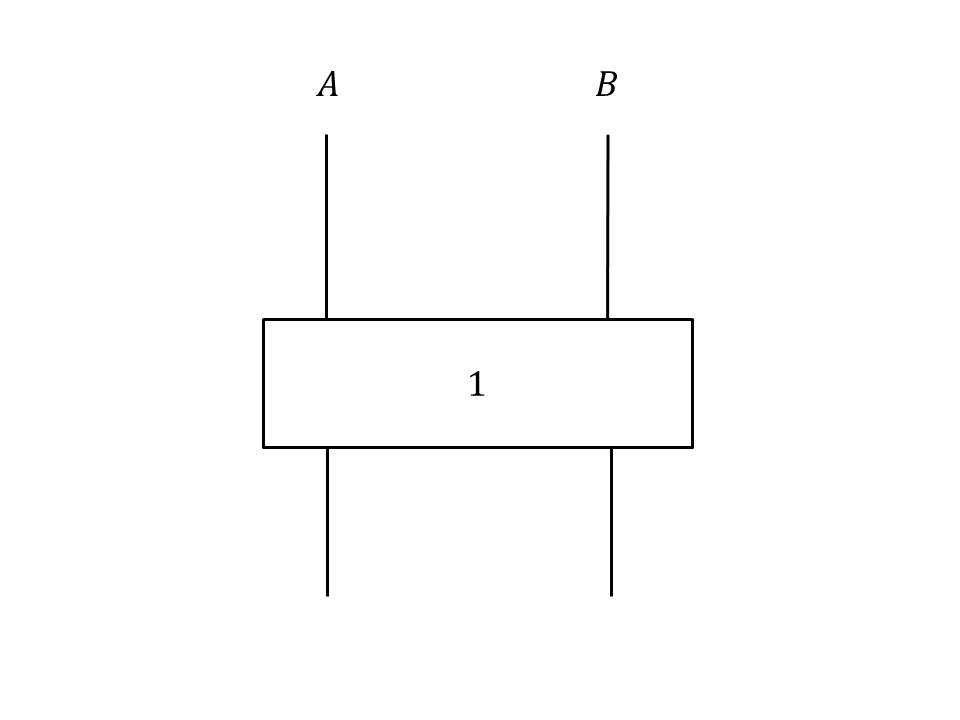}}
  \caption{}
  \label{fig1}
\end{subfigure}%
\hspace{-1cm}
\begin{subfigure}{.5\textwidth}
  \centering
{\includegraphics[height=2.5in]{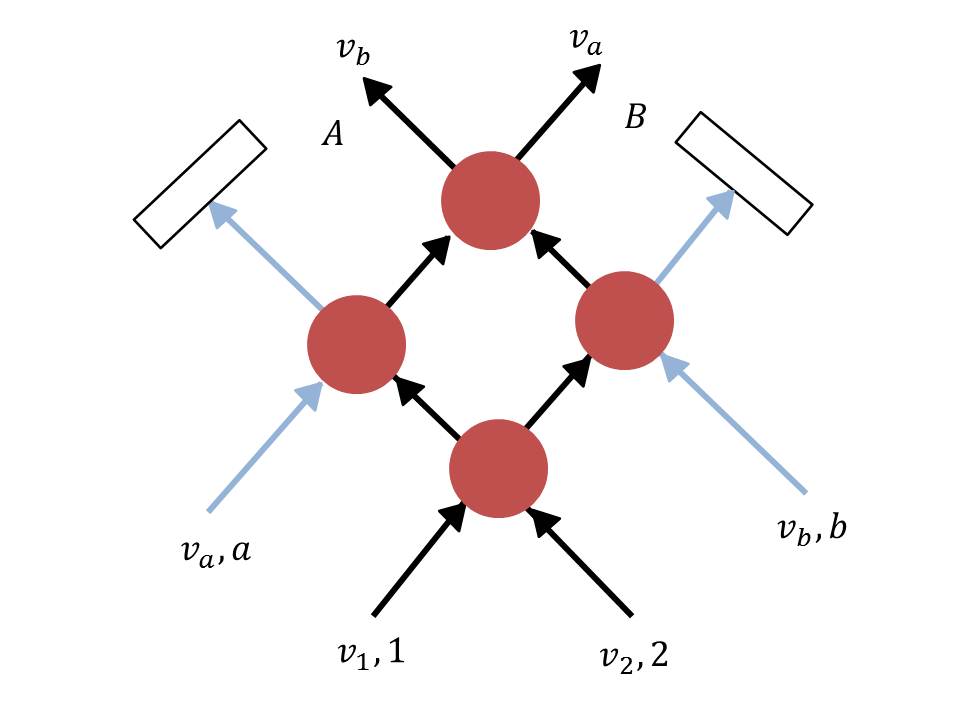}}
  \caption{}
  \label{fig2}
\end{subfigure}
\caption[Nondeterministic four-particle gadget which allows braiding of the balls]{(a) The representation of an $X$ operator. The gate permutes the input labels, and in the end we measure the labels of output wires $A$ and $B$. (b) Four-particle scattering gadget to simulate the $X$ rotation. Lines represent the trajectories of particles, red circles demonstrate interactions, and white rectangles are detectors. Blue particles are ancillas which mediate computation, and black particles are the particles that we wish to implement the actual quantum swap on. The gate is nondeterministic in the sense that it succeeds in producing the desired superposition on labels $|1\rangle$ and $|2\rangle$ only if the left and right detectors detect $|a\rangle$ and $|b\rangle$ labels in the particle label basis, respectively. Conditioned on both ancilla particles bounce off the black particles, the gate operates successfully. The probability of success, thereby, depends on the velocities.
}
\label{fig1-2}
\end{figure}

\begin{figure}[tp]
\centering
\includegraphics[height=3.0in]{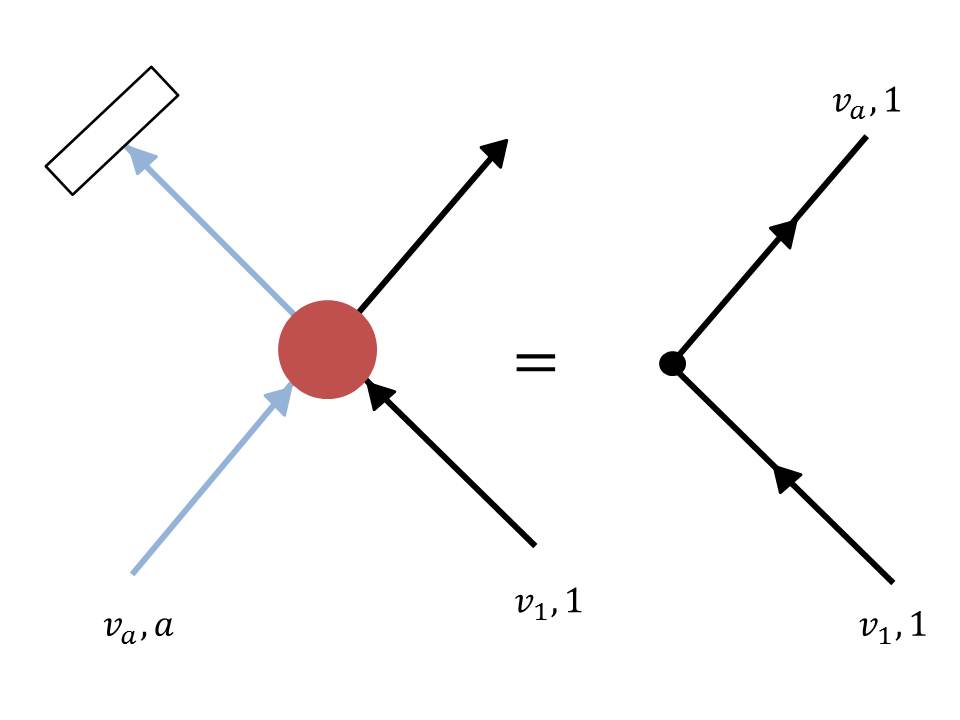}
\caption[Nondeterministic two-particle gadget to navigate the particles]{Two-particle gadget to navigate the trajectory of a single particle. Since in the model we consider the particles move on straight lines, we use this nondeterministic gadget to change the particle's trajectory. The particle that is moving left with velocity $v_1$ non-deterministically changes direction to the right with velocity $v_a$, and this event succeeds only if the detector on the left detects label $|a\rangle$. If the velocities match, $v_a=v_1$, the overall action is a change of direction.
}
\label{fig3}
\end{figure}

\begin{figure}[tp]
\centering
\begin{subfigure}{.5\textwidth}
  \centering
{\includegraphics[height=2.0in]{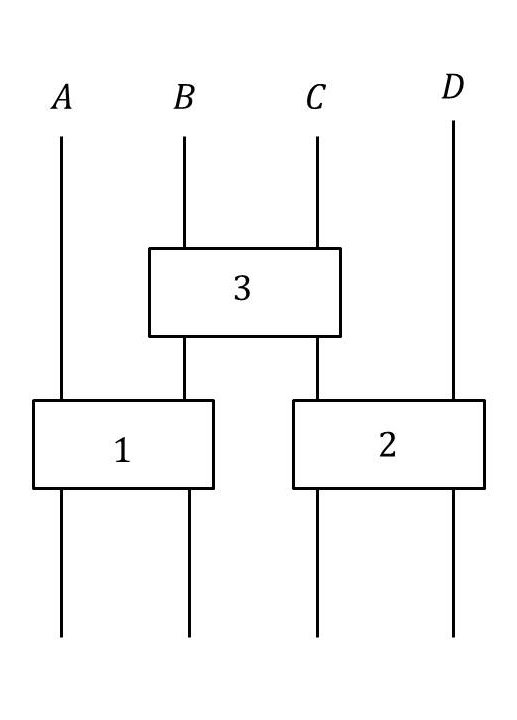}}
  \caption{}
  \label{fig4}
\end{subfigure}%
\hspace{-1cm}
\begin{subfigure}{.5\textwidth}
  \centering
{\includegraphics[height=2.5in]{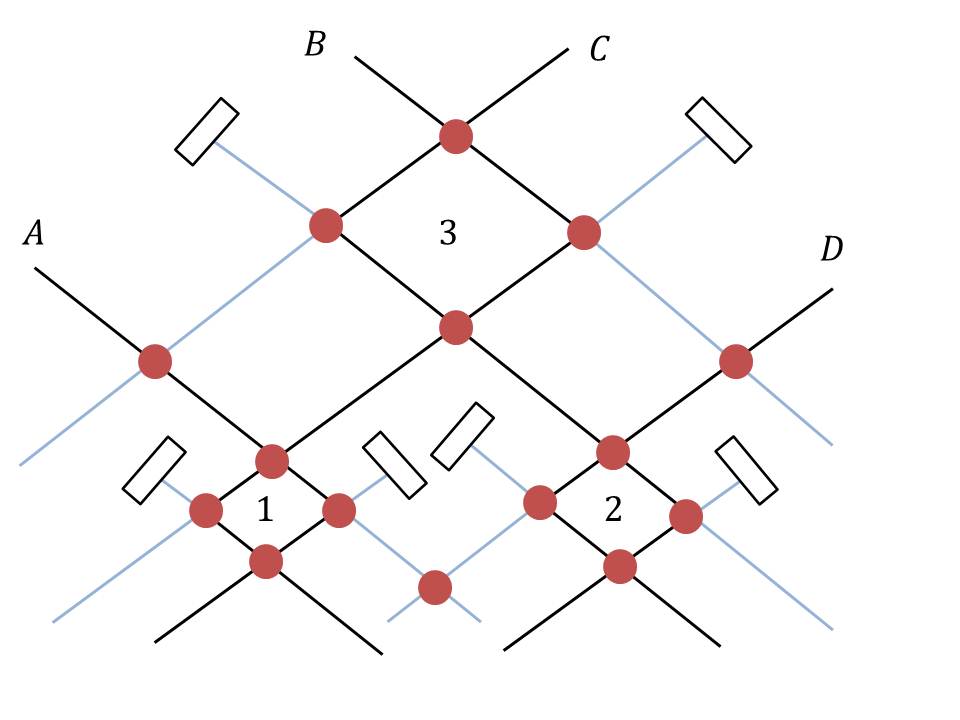}}
  \caption{}
  \label{fig5}
\end{subfigure}
\caption[An example for the simulation of a ball permuting circuit with nondeterministic ball scattering gadgets]{(a) Example of a combination of $X$ operators forming a circuit. The circuit consists of three gates, $1$, $2$, and $3$, and in the end the wires $A$, $B$, $C$, and $D$ are measured in the label basis. (b) An architecture of quantum ball permuting circuit based on particle scattering and intermediate particle measurements to simulate quantum ball permuting circuit of Figure (a). The circuit consists of six ancilla particles which mediate the computation and are detected intermediately with detectors. The labels $1$, $2$, and $3$ demonstrate the simulation of gates $1$, $2$, and $3$, of Figure (a), respectively. In the end we measure the particle locations $A$, $B$, $C$, and $D$. Conditioned on all ancilla particles succeed in passing through all of the intermediate interactions and bouncing off the last interaction, the overall scattering process succeeds in its simulation.}
\label{fig4-5}
\end{figure}


\begin{figure}[tp]
\centering
\begin{subfigure}{.5\textwidth}
  \centering
{\includegraphics[height=2.0in]{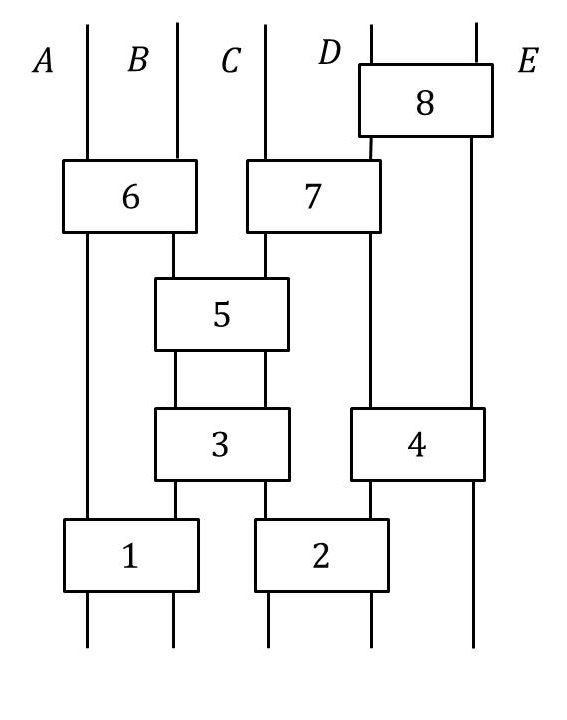}}
  \caption{}
  \label{fig6}
\end{subfigure}%
\hspace{-1cm}
\begin{subfigure}{.5\textwidth}
  \centering
{\includegraphics[height=3.0in]{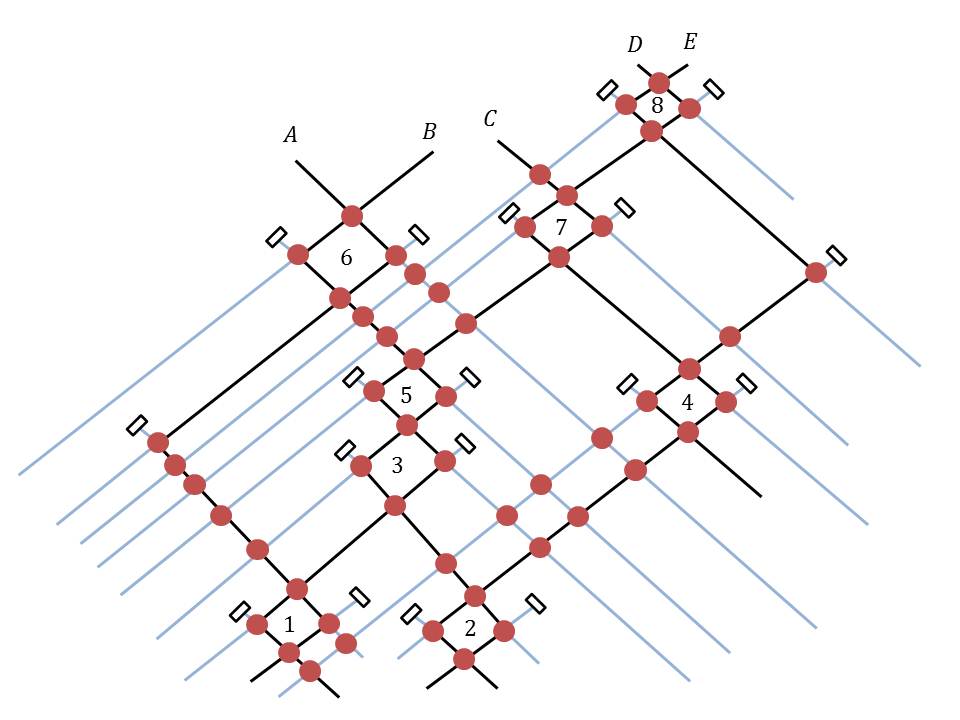}}
  \caption{}
  \label{fig7}
\end{subfigure}
\caption[Another example of ball permuting simulation in nondeterministic ball scattering]{(a) Another example of a quantum circuit with $X$ gates, $1$, $2$, \ldots, $8$, on five labels. In the end we measure the wires $A, B, C, D$ and $E$, in the particle label basis. (b) Programming of particle scattering with intermediate measurements to simulate the $X$ quantum ball permuting circuit of Figure (a) nondeterministically. The labels $1, 2, \ldots, 8$ correspond to the simulation of gates $1, 2, \ldots, 8$ in Figure (a), respectively. Notice that in this example the ancilla particles can experience many intermediate interactions. This example demonstrates that the overall process succeeds in successful simulation, only with small probability, and in general simulations, the probability of success can be exponentially small in the number of particles being used. Therefore, postselecting on the measurement outcomes, one can successfully simulate any $X$ ball permuting quantum circuit. In the end all the particle locations $A, B, C, D$ and $E$ are measured. A drawback in this model of simulation is that it is hard to set the velocities as we proceed to higher layers of the quantum circuit, and we might need to use particles with higher and higher velocity, as we proceed to the top of the circuit.}
\label{fig6-7}
\end{figure}

\begin{figure}[tp]
\centering
\includegraphics[height=3.0in]{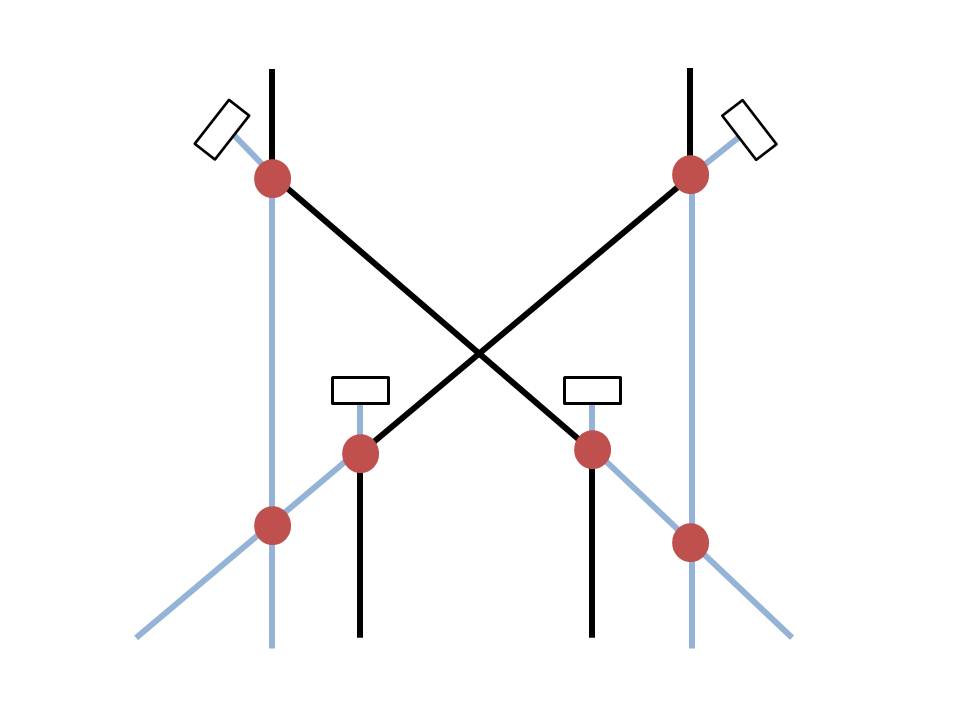}
\caption[Nondeterministic stationary ball scattering gadget to simulate an $X$ operator]{Nondeterministic four-particle gadgets for stationary programming of particle scattering with particle collisions and intermediate measurements. The overall gadget simulates the two label permutation of Figure \ref{fig1}. The objective is to produce superpositions on stationary black particles. Here a stationary particle means a particle that does not move. Initially, two black particles are stationary in the beginning, and we put two more stationary ancilla particles next to them. Then we shoot two ancilla particles from left and right and measure and postselect on them being bounced off from the black particles. Then the two black particles collide with the two stationary ancilla particles and we measure and postselect on the ancilla particles being bounced off in the end. In this scheme it is easier to set the particle scatterings.  
}
\label{fig8}
\end{figure}


\begin{figure}[tp]
\centering
\begin{subfigure}{.5\textwidth}
  \centering
{\includegraphics[height=2.0in]{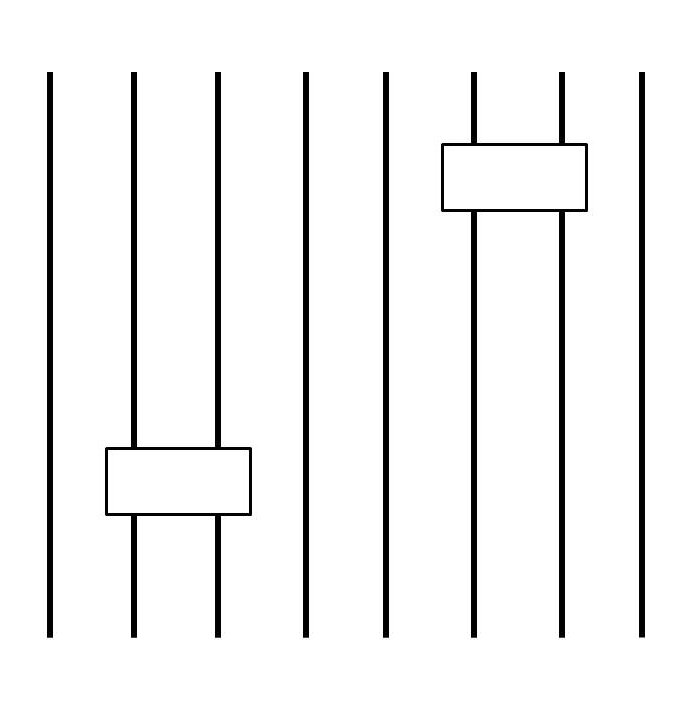}}
  \caption{}
  \label{fig9}
\end{subfigure}%
\hspace{-1cm}
\begin{subfigure}{.5\textwidth}
  \centering
{\includegraphics[height=3.0in]{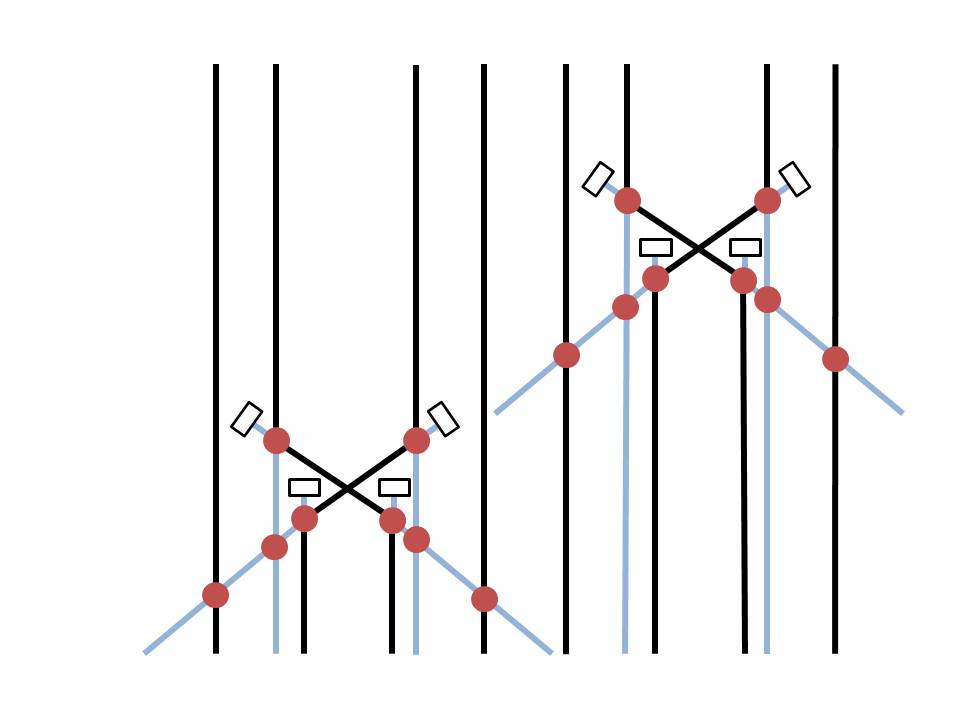}}
  \caption{}
  \label{fig10}
\end{subfigure}
\caption[An example for stationary simulation of an $X$ circuit with ball scattering]{Stationary programming of particle scattering. (a) An example of an $X$ ball permuting circuit on two gates and eight labels. (b) Stationary nondeterministic simulation of the circuit in Figure (a) with ball scattering and intermediate measurements. Each gate in Figure (a) is simulate by a gadget of Figure \ref{fig8}. Except for intermediate interactions, the black particles remain stationary at all the times.}
\label{fig9-10}
\end{figure}

\section{Detailed Proofs for Section \ref{Post}}
\label{phcol} 

\begin{theorem}
$\Post \RQBall=\Post \BQP= \PP$.
\label{thmxqbll}
\end{theorem}

\begin{proof}
$\Post \BQP = \PP$ is given by the result of Aaronson. Also, $\Post \HQBALL\subseteq \Post \BQP$. In order to see this, observe that the $\BQP$ machine first prepares the initial state $|\psi^\star\rangle \tensor |c_1, c_2, \ldots, c_m\rangle$, encoded with binary strings. Then, whenever an intermediate measurement is done, it just leaves the state along and postpones the measurement to the end of computation. This might give rise to a non-planar quantum circuit, but it is fine, since we are working with $\BQP$. The $\BQP$ measurements are done in a proper basis that encodes the ball color basis. For example, if we encode ball colors with binary representations, then it is sufficient to measure in binary basis and confirm if the digital representation of the ball color (number) is correct. Notice that in the simulation, we are not going to use the balls that have already been measured intermediately again. All the swap gates are applied accordingly. Then in the end we postselect on the desired demolition measurements and in the end we measure the encoded location of the $j$'th ball and confirm if it is among $\tilde {c}$. We can also use CNOT gates to shrink the number of postselections down to one.

In order to see the more interesting direction $\Post \HQBALL\supseteq \Post \BQP$, just we follow the postselected universality of theorem \ref{PPeqPostRQBall} to simulate any computation in $\BQP$. Then notice that any $\Post \BQP$ computation can be deformed in a way that one only needs to postselect on one qubit, and also measure one qubit in the end. $\Post \HQBALL$  uses this deformed $\Post \BQP$ protocol, instead, and uses one of its demolition measurements in the end of computation to simulate postselection of the actual $\Post \BQP$ circuit.
\end{proof}

\begin{theorem}
The existence of a $\BPP$ algorithm to create a probability distribution withing multiplicative error to the actual distribution on $c_1, c_2, \ldots, c_m$ and $c_0$ of definition ~\ref{PostScatdef} implies $\Post \RQBall\subseteq \Post \BPP$.
\end{theorem}

\begin{proof}
The proof is similar to the proof of theorem $2$ in \cite{bremner2010classical}. Suppose that there is a procedure which outputs the numbers $x_1, x_2, \ldots, x_m, y$ such that:

\begin{eqnarray*}
1/\alpha \operatorname*{Pr}[c'_0=c_0 ; c'_1=c_1, c'_2=c_2, \ldots, c'_m=c_m]&<&\\
\operatorname*{Pr}[y=c_0; x_1=c_1, x_2=c_2, \ldots, x_m=c_m]&<&\alpha \operatorname*{Pr}[c'_0=c_0 ; c'_1=c_1, c'_2=c_2, \ldots, c'_m=c_m],
\end{eqnarray*}

\noindent for every list of colors $c_1, c_2, \ldots, c_m$ and $c_0$. Notice that if this is true, it should also be true for all marginal probability distributions. Denote the vector $(c'_1, c'_2, \ldots, c'_m,  c'_0)$ by $(\tilde{c'}, c'_0)$, and $(x_1, x_2, \ldots, x_m,  y)$ by $(\tilde{x}, y)$. Then the conditional probabilities also satisfy:

$$
1/{\alpha^2}\operatorname*{Pr}[c'_0=c_0|\tilde{c'}=\tilde{c}] <\operatorname*{Pr}[y=c_0|\tilde{x}=\tilde{c}]<\alpha^2 \operatorname*{Pr}[c'_0=c_0|\tilde{c'}=\tilde{c}].
$$

\noindent Now let $L$ be any language in $\Post\HQBALL$. Then if $x\in L$, $\operatorname*{Pr}[d'=d| C]\geq 2/3$ and otherwise $\leq 1/3$. Suppose that $x\in L$, then in order for $\operatorname*{Pr}[y=d|\tilde{x}=\tilde{c}]>1/2$, it is required that:

$$
1/2 < 1/{\alpha^2} (2/3)
$$

\noindent which means if $1 <\alpha < \sqrt{4/3}-O(1)$, then the $\Post\BPP$ algorithm recognizes $x$ with bounded probability of error. The probability of success can be increased by using the majority of votes' technique.

\end{proof}

\section{Open Problems and Further Directions }
\label{openproblems}

\begin{enumerate}
\item In section \ref{Post} we showed that the model $\RQBall$ with intermediate measurements is universal under postselection, and therefore $\RQBall$ cannot be efficiently simulated classically unless the polynomial hierarchy collapses. Is this also true without intermediate measurements?

\item In section \ref{sep} we proved that single amplitudes can be efficiently approximated within additive error by $\DQC 1$ computation. However it is unknown if an efficient $\DQC 1$ (approximate) \emph{sampling} scheme exists for this model. Also a lower-bound on the power of the ball permuting model on standard initial states is not known.

\item It remains open to generalize the result of Section \ref{sep} to arbitrary quantum models with Hilbert spaces based on group algebras. As discussed, such a generalization relies on the ability to do the following: given the description of the generators of a group, encode the elements of an arbitrary element with binary strings with nearly $\log |G|$ bits so that the action of group elements on each other is implementable by reversible circuits that affect only $O(\log \log |G|)$ bits at a time. We  established this property for the symmetric group using the factorial number system. Are there other groups with this property?

\item In Section \ref{sep} we show that if the quantum ball permuting model has access to arbitrary initial states, then there is a way to efficiently sample from standard quantum circuits. The construction is based on encoding of qubits using superpositions over permutation states. More precisely $n$ qubits can be encoded using a superposition over permutations of $3n$ labels. However, a drawback of this construction is that it is not composable, in the sense that the encoding of $k$ qubits is not obtained by taking a $k$-fold tensor product of the encoding of a single qubit. Instead, to encode $k$ qubits, we use $3k$ qubits and symmetrize over all the labels representings $0$'s and all the labels representing $1$'s. A natural question is to give a composable encoding of qubits in this model, or prove such an encoding is impossible.

\item In section \ref{intermediate} we introduced the complexity class $\Samp\TQP$, in which we start with the maximally entangled initial state $\sum_x \ket{x}\ket{x}$, we apply an arbitrary quantum circuit to the left half of this state and at the end we measure both halves. We show that the power of this model is intermediate between $\DQC_1$ and $\BQP$. It is open to further classify the power of this class.

\item Is there an encoding of the Young-Yamanouchi basis of arbitrary tableaux into bit strings which is both extremely compressed (i.e. using nearly the information-theoretically optimal number of bits) and local (i.e. exchanging two labels is an operation which can be performed on $O(\log n)$ bits). We believe that this should be achievable. 
Indeed given a circuit $C$, we can compute the average trace of $C$ over the partitions $\sum_{\lambda} p_\lambda \operatorname*{Tr}_\lambda( C )$ as this is simply equal to $\langle 123 \ldots n| C |123\ldots n \rangle$. So it would be strange (but not logically impossible) if one could not compute the individual $\operatorname*{Tr}_\lambda( C )$ 's.

\item In section \ref{partialclassification1} we partially classify the computational power of the ball permuting model on arbitrary initial states. We prove that the unitary group generated by this model is as large as possible if the model starts from the initial states corresponding to Young diagrams with two rows or two columns. We conjecture that this is true for arbitrary irreducible representations. The main difficulty in proving this is that the bridge lemma (Lemma \ref{bridge}) works only if the subspaces are of \emph{different dimensions}. However, there are cases where two subspaces of equal dimensionality take part in a single branching rule, so the bridge lemma is not applicable. An additional difficulty is that the action of ball permuting gates on different irreducible subspaces can be coupled\footnote{For example,  if $\lambda^T$ is the transpose of the partition $\lambda$ then if a sequence of ball permuting gates apply the unitary $U$ on $V(\lambda)$ then the same sequence applies the unitary $U^\ast$  to $V(\lambda^T)$ after some change of basis. However, we conjecture this is the only coupling possible in our decomposition. Furthermore, even with such coupling between irreps the model we obtain can be $\BQP$ universal. For example, consider the situation where the initial state is $|\psi\rangle=\dfrac{1}{\sqrt{2}}(|x\rangle +|x'\rangle)$ where $|x\rangle$ and $|x'\rangle$ are in separate irreps. Then the coupled action of the form $U \oplus U^\star$ maps $|\psi \rangle$ to $\dfrac{1}{\sqrt{2}}(U|x\rangle +U^\star|x'\rangle)$, so $\langle \psi | C | \psi \rangle = \dfrac{1}{2} (\langle x | U | x \rangle + \langle x' | U^\star | x' \rangle)= \Re \langle x | U |x \rangle$. Note the problem of reading a real entry of a quantum circuit is already known to be $\BQP$-complete, so approximating $\langle \psi | C | \psi \rangle$ is $\BQP$-complete.}. 
We leave the full classification to future work.

\item For the randomized ball permuting model it is proved that $\BPL\subseteq \RBALL \subseteq \Almost \L$. Can $\RBALL$ be further pinned down within these classes? Moreover, with two adaptive queries to the randomized ball permuting oracle, $\RBALL$ can simulate $\Almost \L$. Can $\RBALL$ still simulate $\Almost \L$ with only one query?

\item For the classical ball permuting model, one can consider restricting  the probability of swaps to obey the Yang-Baxter equation. Is there a $\P$ simulation in this case?

\item We proved that $\BPL \subseteq \RBall \subseteq \Almost \L \subseteq \BPP$. The class $\Almost \L$ is not known to be contained in $\P$. It remains open to see if $\RBall \subseteq \P$.

\item We proved that with two adaptive queries to the $\RBall$ oracle, $\RBall = \Almost \L$. Can $\Almost \L$ be simulated with $\RBall$ using one query?

\item What is the complexity of $\#\Ball$ and $\#\Ball^\star_{adj}$ and exact computation of probabilities in $\Ball$? Since $\NBall=\NP$ there is no approximation within multiplicative error for $\#\Ball$, unless $\P= \NP$. Is this the case that $\# \Ball = \#\P$? Although deciding if some target permutation has nonzero probability is decidable in polynomial time, the counting version still seems like a hard task.

\item We defined the ball permuting oracles with all distinguishable balls. What is the complexity class if the balls are labeled with $0$ and $1$ labels only? Clearly, the original oracles can simulate the binary balls, and in the case of $\DBall$ there are reductions in both ways. It is not clear if we recover $\RBall$ and $\NBall$ in the case of binary balls.

\item The problem of deciding if a target permutation has nonzero probability in $\RBall^\star_{adj}$ is in $\P$, however the same problem for $\RBall$ is $\NP$-complete. Is $\RBall^\star_{adj}$ itself a weaker class than $\RBall$?

\end{enumerate}

\end{appendices}

\end{document}

%% file: header.tex
\newtheorem{theorem}{Theorem}[section]
\newtheorem{definition}{Definition}[section]

\newtheorem{corollary}[theorem]{Corollary}

\newtheorem{lemma}[theorem]{Lemma}

\newcommand{\TQP}{\mathsf{TQP}}
\newcommand{\VDP}{\mathsf{VDP}}
\newcommand{\Rev}{\mathsf{Rev}}
\newcommand{\EDP}{\mathsf{EDP}}
\newcommand{\WPPP}{\mathsf{WPPP}}

\newcommand{\BPP}{\mathsf{BPP}}
\newcommand{\PP}{\mathsf{PP}}
\newcommand{\BPL}{\mathsf{BPL}}

\newcommand{\Post}{\mathsf{Post}}
\newcommand{\Meas}{\mathsf{Meas}}

\renewcommand{\P}[0]{\ensuremath{\mathsf{P}}}

\renewcommand{\L}[0]{\ensuremath{\mathsf{L}}}

\newcommand{\Trace}{\mathsf{Trace}}
\newcommand{\DQC}{\mathsf{DQC}}

\newcommand{\PH}{\mathsf{PH}}

\newcommand{\LOGSPACE}{\mathsf{LOGSPACE}}
\newcommand{\NC}{\mathsf{NC}}

\newcommand{\SAT}{\mathsf{SAT}}
\newcommand{\Almost}{\mathsf{Almost}}
\newcommand{\AC}{\mathsf{AC}}
\newcommand{\NP}{\mathsf{NP}}

\newcommand{\Ball}{\mathsf{Ball}}
\newcommand{\BALL}{\mathsf{BALL}}
\newcommand{\RBALL}{\mathsf{RBall}}
\newcommand{\RBall}{\mathsf{RBall}}
\newcommand{\DBALL}{\mathsf{DBall}}
\newcommand{\DBall}{\mathsf{DBall}}

\newcommand{\NBall}{\mathsf{NBall}}
\newcommand{\Samp}{\mathsf{Samp}}

\newcommand{\HQBALL}{\mathsf{RQBall}}
\newcommand{\XQBALL}{\mathsf{QBall}}

\newcommand{\RQBall}{\mathsf{RQBall}}
\newcommand{\QBall}{\mathsf{QBall}}

\newcommand{\BQP}{\mathsf{BQP}}
\newcommand{\Poly}{\mathsf{poly}}

\newcommand{\C}{\mathbb{C}}

\newcommand{\ket}[1]{|{#1}\rangle}
\newcommand{\bra}[1]{\langle {#1}|}
\newcommand{\tensor}{\otimes}

\newcommand{\su}{\mathfrak{su}}